\documentclass[12pt, draftclsnofoot, onecolumn]{IEEEtran}
\usepackage{epsfig,graphicx,subfigure,psfrag,amsmath,cases}
\usepackage{latexsym,amssymb,amsmath,epsfig,subfigure,algorithm,mathtools}
\usepackage{algorithmic}
\usepackage{color}
\usepackage{url}
\usepackage{scrtime}
\usepackage{cite}
\usepackage{subfigure}
\usepackage{bbding}
\usepackage{multicol}
\usepackage{multirow}

\author{Zhiqiang Wei, Lei Yang, Derrick Wing Kwan Ng,\\ Jinhong Yuan, and Lajos Hanzo
\thanks{Zhiqiang Wei, Lei Yang, Derrick Wing Kwan Ng, and Jinhong Yuan are with the School of Electrical Engineering and Telecommunications, the University of New South Wales, Australia (email: zhiqiang.wei@student.unsw.edu.au; lei.yang3@unsw.edu.au; w.k.ng@unsw.edu.au; j.yuan@unsw.edu.au).
Lajos Hanzo is with the School of Electronics and Computer Science, University of Southampton, Southampton, UK (email: lh@ecs.soton.ac.uk).
The conference version of this paper has been presented at the IEEE Globecom 2018 \cite{Wei2018GlobeCom}.}}

\title{On the Performance Gain of NOMA over OMA in Uplink Communication Systems}

\newtheorem{Thm}{Theorem}
\newtheorem{Lem}{Lemma}
\newtheorem{proof}{proof}
\newtheorem{Cor}{Corollary}

\newtheorem{T-Prob}{Transformed Problem}

\newcommand{\tabincell}[2]{\begin{tabular}{@{}#1@{}}#2\end{tabular}}

\newtheorem{Remark}{Remark}

\newcommand{\abs}[1]{\lvert#1\rvert}

\begin{document}
\IEEEspecialpapernotice{(Invited Paper)\vspace{-10mm}}
\maketitle
\begin{abstract}
In this paper, we investigate and reveal the ergodic sum-rate gain (ESG) of non-orthogonal multiple access (NOMA) over orthogonal multiple access (OMA) in uplink cellular communication systems.
A base station equipped with a single-antenna, with multiple antennas, and with massive antenna arrays is considered both in single-cell and multi-cell deployments.
In particular, in single-antenna systems, we identify two types of gains brought about by NOMA:
1) a large-scale near-far gain arising from the distance discrepancy between the base station and users; 2) a small-scale fading gain originating from the multipath channel fading.
Furthermore, we reveal that the large-scale near-far gain increases with the normalized cell size, while the small-scale fading gain is a constant, given by $\gamma = 0.57721$ nat/s/Hz, in Rayleigh fading channels.
When extending single-antenna NOMA to $M$-antenna NOMA, we prove that both the large-scale near-far gain and small-scale fading gain achieved by single-antenna NOMA can be increased by a factor of $M$ for a large number of users.
Moreover, given a massive antenna array at the base station and considering a fixed ratio between the number of antennas, $M$, and the number of users, $K$, the ESG of NOMA over OMA increases linearly with both $M$ and $K$.
We then further extend the analysis to a multi-cell scenario.
Compared to the single-cell case, the ESG in multi-cell systems degrades as NOMA faces more severe inter-cell interference due to the non-orthogonal transmissions.
Besides, we unveil that a large cell size is always beneficial to the ergodic sum-rate performance of NOMA in both single-cell and multi-cell systems.
Numerical results verify the accuracy of the analytical results derived and confirm the insights revealed about the ESG of NOMA over OMA in different scenarios.
\end{abstract}


\section{Introduction}
The networked world we live in has revolutionized our daily life.
Wireless communications has become one of the disruptive technologies and it is one of the best business opportunities of the future\cite{Andrews2014,wong2017key}.
In particular, the development of wireless communications worldwide fuels the massive growth in the number of wireless communication devices and sensors for emerging applications such as smart logistics \& transportation, environmental monitoring, energy management, safety management, and industry automation, just to name a few.
It is expected that in the Internet-of-Things (IoT) era \cite{Zorzi2010}, there will be $50$ billion wireless communication devices connected worldwide with a connection density up to a million devices per $\rm{km}^{\rm{2}}$ \cite{QualComm,SunALOHA}.
The massive number of devices and explosive data traffic pose challenging requirements, such as massive connectivity \cite{SunJUICE} and ultra-high spectral efficiency for future wireless networks\cite{Andrews2014,wong2017key}.
As a result, compelling new technologies, such as massive multiple-input multiple-output (MIMO)\cite{Marzetta2010,ngo2013energy}, non-orthogonal multiple access (NOMA)\cite{Dai2015,Ding2015b,WeiSurvey2016,QiuLOMA,QiuLOMASlowFading}, and millimeter wave (mmWave) communications\cite{Rappaport2013,XiaoMing2017,zhao2017multiuser,wei2018multibeam,WeiBeamWidthControl} etc. have been proposed to address the aforementioned issues.
Among them, NOMA has drawn significant attention both in industry and in academia as a promising multiple access technique.
The principle of power-domain NOMA is to exploit the users' power difference for multiuser multiplexing together with superposition coding at the transmitter, while applying successive interference cancelation (SIC) at the receivers for alleviating the inter-user interference (IUI)\cite{WeiSurvey2016}.
In fact, the industrial community has proposed up to 16 various forms of NOMA as the potential multiple access schemes for the forthcoming fifth-generation (5G) networks\cite{ChenNOMAScheme}.

Compared to the conventional orthogonal multiple access (OMA) schemes, NOMA allows users to simultaneously share the same resource blocks and hence it is beneficial for supporting a large number of connections in spectrally efficient communications.
The concept of non-orthogonal transmissions dates back to the 1990s, e.g. \cite{Cover1991,Tse2005}, which serves as a foundation for the development of the power-domain NOMA.
Indeed, NOMA schemes relying on non-orthogonal spreading sequences have led to popular code division multiple access (CDMA) arrangements, even though eventually the so-called orthogonal variable spreading factor (OVSF) code was selected for the global third-generation (3G) wireless systems\cite{Verdu1999,YangCDMA,Hanzo2003,WangPowerEfficiency}.
To elaborate a little further, the spectral efficiency of CDMA was analyzed in \cite{Verdu1999}.
In \cite{YangCDMA}, the authors compared the benefits and deficiencies of three typical CDMA schemes: single-carrier direct-sequence CDMA (SC DS-CDMA), multicarrier CDMA (MC-CDMA), and multicarrier DS-CDMA (MC DS-CDMA).
Furthermore, a comparative study of OMA and NOMA was carried out in \cite{WangPowerEfficiency}.
It has been shown that NOMA possesses a spectral-power efficiency advantage over OMA\cite{WangPowerEfficiency} and this theoretical gain can be realized with the aid of the interleave division multiple access (IDMA) technique proposed in \cite{Ping2006IDMA}.
Despite the initial efforts on the study of NOMA, the employment of NOMA in practical systems has been developing relatively slowly due to the requirement of sophisticated hardware for its implementation.
Recently, NOMA has rekindled the interests of researchers as a benefit of the recent advances in signal processing and silicon technologies\cite{Access2015,NOMADOCOMO}.
However, existing contributions, e.g. \cite{Cover1991,Tse2005,WangPowerEfficiency}, have mainly focused their attention on the NOMA performance from the information theoretical point of view, such as its capacity region\cite{Cover1991,Tse2005} and power region\cite{WangPowerEfficiency}.
The recent work in \cite{Xu2017} intuitively explained the source of performance gain attained by NOMA over OMA via simulations.
The authors of \cite{ProceedingLiu,DaiCST2018} surveyed the state-of-the-art research on NOMA and offered a high-level discussion of the challenges and research opportunities for NOMA systems.
However, to the best of our knowledge, there is a paucity of literature on the comprehensive analysis of the achievable ergodic sum-rate gain (ESG) of NOMA over OMA relying on practical signal detection techniques.
Furthermore, the ESG of NOMA over OMA in different practical scenarios, such as single-antenna, multi-antenna, and massive antenna array aided systems relying on single-cell or multi-cell deployments has not been compared in the open literature.

As for single-antenna systems, several authors have analyzed the performance of NOMA from different perspectives, e.g.\cite{Xu2015,Ding2014,Chen2017,Yang2016}.
More specifically, based on the achievable rate region, Xu \emph{et al.} proved in \cite{Xu2015} that NOMA outperforms time division multiple access (TDMA) with a high probability in terms of both its overall sum-rate and the individual user-rate.
Furthermore, the ergodic sum-rate of single-input single-output NOMA (SISO-NOMA) was derived and the performance gain of SISO-NOMA over SISO-OMA was demonstrated via simulations by Ding \emph{et al.}\cite{Ding2014}.
Upon relying on their new dynamic resource allocation design, Chen \emph{et al.} \cite{Chen2017} proved that SISO-NOMA always outperforms SISO-OMA using a rigorous optimization technique.
In \cite{Yang2016}, Yang \emph{et al.} analyzed the outage probability degradation and the ergodic sum-rate of SISO-NOMA systems by taking into account the impact of partial channel state information (CSI).
As a further development, efficient resource allocation was designed for NOMA systems by Sun \emph{et al.} \cite{Sun2016Fullduplex} as well as by Wei \emph{et al.} \cite{WeiTCOM2017} under the assumptions of perfect CSI and imperfect CSI, respectively.
The simulation results in \cite{Sun2016Fullduplex} and \cite{WeiTCOM2017} quantified the performance gain of NOMA over OMA in terms of its spectral efficiency and power efficiency, respectively.
The aforementioned contributions studied the performance of NOMA systems or discussed the superiority of NOMA over OMA in different contexts.
However, the analytical results quantifying the ESG of SISO-NOMA over SISO-OMA has not been reported at the time of writing.
More importantly, the source of the performance gain of NOMA over OMA has not been well understood and the impact of specific system parameters on the ESG, such as the number of NOMA users, the signal-to-noise ratio (SNR), and the cell size, have not been revealed in the open literature.

To achieve a higher spectral efficiency, the concept of NOMA has also been amalgamate with multi-antenna systems, resulting in the notion of multiple-input multiple-output NOMA (MIMO-NOMA), for example, by invoking the signal alignment technique of Ding \emph{et al.} \cite{DingSignalAlignment} and the quasi-degradation-based precoding design of Chen \emph{et al.} \cite{ChenQuasiDegradation}.
Although the performance gain of MIMO-NOMA over MIMO-OMA has indeed been shown in \cite{DingSignalAlignment,ChenQuasiDegradation} with the aid of simulations, the performance gain due to additional antennas has not been quantified analytically.
Moreover, how the ESG of NOMA over OMA increases upon upgrading the system from having a single antenna to multiple antennas is still an open problem at the time of writing, which deserves our efforts to explore.
The answers to these questions can shed light on the practical implementation of NOMA in future wireless networks.
On the other hand, there are only some preliminary results on applying the NOMA principle to massive-MIMO systems.
For instance, Zhang \emph{et al.} \cite{ZhangDiMassive} investigated the outage probability of massive-MIMO-NOMA (\emph{m}MIMO-NOMA).
Furthermore, Ding and Poor \cite{DingMassive} analyzed the outage performance of \emph{m}MIMO-NOMA relying on realistic limited feedback and demonstrated a substantial performance improvement for \emph{m}MIMO-NOMA over \emph{m}MIMO-OMA.
Upon extending NOMA to a mmWave massive-MIMO system, the capacity attained in the high-SNR regime and low-SNR regime were analyzed by Zhang \emph{et al.} \cite{ZhangDimmWave}.
Yet, the ESG of \emph{m}MIMO-NOMA over \emph{m}MIMO-OMA remains unknown and the investigation of \emph{m}MIMO-NOMA has the promise attaining NOMA gains in large-scale systems in the networks of the near future.

On the other hand, although single-cell NOMA has received significant research attention \cite{Xu2015,Ding2014,Chen2017,Yang2016,Sun2016Fullduplex,WeiTCOM2017,DingSignalAlignment,ChenQuasiDegradation,ZhangDiMassive,DingMassive,ZhangDimmWave}, the performance of NOMA in multi-cell scenarios remains unexplored but critically important for practical deployment, where the inter-cell interference becomes a major obstacle\cite{ShinMulticellNOMA}.
Centralized resource optimization of multi-cell NOMA was proposed by You in \cite{YouMulticellNOMA}, while a distributed power control scheme was studied in \cite{FuMulticellNOMA}.
The transmit precoder design of MIMO-NOMA aided multi-cell networks designed for maximizing the overall sum throughput was proposed by Nguyen \emph{et al.} \cite{NguyenMulticellNOMA} and a computationally efficient algorithm was proposed for achieving a locally optimal solution.
Despite the fact that the simulation results provided by \cite{ShinMulticellNOMA,YouMulticellNOMA,FuMulticellNOMA,NguyenMulticellNOMA} have demonstrated a performance gain for applying NOMA in multi-cell cellular networks, the analytical results quantifying the ESG of NOMA over OMA for multi-cell systems relying on single-antenna, multi-antenna, and massive-MIMO arrays at the BSs have not been reported in the open literature.
Furthermore, the performance gains disseminated in the literature have been achieved for systems having a high transmit power or operating in the high-SNR regime.
However, a high transmit power inflicts a strong inter-cell interference, which imposes a challenge for the design of inter-cell interference management.
Therefore, there are many practical considerations related to the NOMA principle in multi-cell systems, while have to be investigated.

\begin{table}
	\caption{Comparison of This Work with Literature for the Results of Performance Gain of NOMA over OMA}
	\centering
	\begin{tabular}{c|l|ccccccccc}
		\hline
		Considered system setup  & Main results & \cite{Xu2017} & \cite{Ding2014} & \cite{Chen2017} &  \cite{DingSignalAlignment,ChenQuasiDegradation} &  \cite{DingMassive} & \cite{YouMulticellNOMA,FuMulticellNOMA,NguyenMulticellNOMA} & This work\\\hline
		\multirow{4}{*}{\tabincell{c}{Single-antenna \\single-cell systems}}
        & Outage probability  &  & \Checkmark & & & & & \\
        & Proof of superiority  &  &  &\Checkmark & & & & \Checkmark\\
        & Ergodic sum-rate  &  & \Checkmark & & & & & \Checkmark\\
        & Ergodic sum-rate gain  &  &  &  & & & & \Checkmark\\
		& Numerical results &  \Checkmark & \Checkmark & \Checkmark & & & & \Checkmark\\\hline
		\multirow{4}{*}{\tabincell{c}{Multi-antenna \\single-cell systems}}
        & Outage probability  &  & & &\Checkmark & & & \\
        & Proof of superiority  &  &  & & & & & \Checkmark\\
        & Ergodic sum-rate  &  &  & & & & & \Checkmark\\
        & Ergodic sum-rate gain  &  &  &  & & & & \Checkmark\\
		& Numerical results&  & & & \Checkmark & & & \Checkmark\\\hline
		\multirow{4}{*}{\tabincell{c}{Massive-MIMO \\single-cell systems}}
        & Outage probability  &  & & & & \Checkmark& & \\
        & Proof of superiority  &  &  & & & & & \Checkmark\\
        & Ergodic sum-rate  &  &  & & & & & \Checkmark\\
        & Ergodic sum-rate gain  &  &  &  & & & & \Checkmark\\
		& Numerical results&\Checkmark  & & & & \Checkmark& & \Checkmark\\\hline
		\multirow{4}{*}{Multi-cell systems}
        & Outage probability  &  & & & & & & \\
        & Proof of superiority  &  &  & & & & & \Checkmark\\
        & Ergodic sum-rate  &  &  & & & & & \Checkmark\\
        & Ergodic sum-rate gain  &  &  &  & & & & \Checkmark\\
		& Numerical results&  \Checkmark & & & & &\Checkmark & \Checkmark\\\hline
	\end{tabular}\label{LiteratureComparison}
\end{table}

In summary, the comparison of our work with the most pertinent existing contributions in the literature is shown in Table \ref{LiteratureComparison}.
Although the existing treatises have investigated the system performance of NOMA from different perspectives, such as the outage probability \cite{Ding2014,DingSignalAlignment,ChenQuasiDegradation,DingMassive} and the ergodic sum-rate \cite{Ding2014}, in various specifically considered system setups, no unified analysis has been published to discuss the performance gain of NOMA over OMA.
To fill this gap, our work offers a unified analysis on the ergodic sum-rate gain of NOMA over OMA in single-antenna, multi-antenna, and massive-MIMO systems with both single-cell and
multi-cell deployments.


This paper aims for providing answers to the above open problems and for furthering the understanding of the ESG of NOMA over OMA in the uplink of communication systems.
To this end, we carry out the unified analysis of ESG in single-antenna, multi-antenna and massive-MIMO systems.
We first focus our attention on the ESG analysis in single-cell systems and then extend our analytical results to multi-cell systems by taking into account the inter-cell interference (ICI).
We quantify the ESG of NOMA over OMA relying on practical signal reception schemes at the base station for both NOMA as well as OMA and unveil its behaviour under different scenarios.
Our simulation results confirm the accuracy of our performance analyses and provide some interesting insights, which are summarized in the following:
\begin{itemize}
  \item In all the cases considered, a high ESG can be achieved by NOMA over OMA in the high-SNR regime, but the ESG vanishes in the low-SNR regime.
  \item In the single-antenna scenario, we identify two types of gains attained by NOMA and characterize their different behaviours.
      In particular, we show that the \emph{large-scale near-far gain} achieved by exploiting the distance-discrepancy between the base station and users increases with the cell size, while the \emph{small-scale fading gain} is given by an Euler-Mascheroni constant\cite{abramowitz1964handbook} of $\gamma = 0.57721$ nat/s/Hz in Rayleigh fading channels.
  \item When applying NOMA in multi-antenna systems, compared to the MIMO-OMA utilizing zero-forcing detection, we analytically show that the ESG of SISO-NOMA over SISO-OMA can be increased by $M$-fold, when the base station is equipped with $M$ antennas and serves a sufficiently large number of users $K$.
  \item Compared to MIMO-OMA utilizing a maximum ratio combining (MRC) detector, an $\left(M-1\right)$-fold degrees of freedom (DoF) gain can be achieved by MIMO-NOMA.
      In particular, the ESG in this case increases linearly with the system's SNR quantified in dB with a slope of $\left(M-1\right)$ in the high-SNR regime.
  \item For massive-MIMO systems with a fixed ratio between the number of antennas, $M$, and the number of users, $K$, i.e., $\delta = \frac{M}{K}$, the ESG of \emph{m}MIMO-NOMA over \emph{m}MIMO-OMA increases linearly with both $K$ and $M$ using MRC detection.
  \item In practical multi-cell systems operating without joint cell signal processing, the ESG of NOMA over OMA is degraded due to the existence of ICI, especially for a small cell size with a dense cell deployment.
      Furthermore, no DoF gain can be achieved by NOMA in multi-cell systems due to the lack of joint multi-cell signal processing to handle the ICI.
      In other words, all the ESGs of NOMA over OMA in single-antenna, multi-antenna, and massive-MIMO multi-cell systems saturate in the high-SNR regime.
  \item For both single-cell and multi-cell systems, a large cell size is always beneficial to the performance of NOMA.
      In particular, in single-cell systems, the ESG of NOMA over OMA is increased for a larger cell size due to the enhanced large-scale near-far gain.
      For multi-cell systems, a larger cell size reduces the ICI level, which prevents a severe ESG degradation.
\end{itemize}
%

The notations used in this paper are as follows. Boldface capital and lower case letters are reserved for matrices and vectors, respectively. ${\left( \cdot \right)^{\mathrm{T}}}$ denotes the transpose of a vector or matrix and ${\left( \cdot \right)^{\mathrm{H}}}$ denotes the Hermitian transpose of a vector or matrix.
$\mathbb{C}^{M\times N}$ represents the set of all $M\times N$ matrices with complex entries.
$\abs{\cdot}$ denotes the absolute value of a complex scalar or the determinant of a matrix, $\|{\cdot}\|$ denotes Euclidean norm of a complex vector, $\lceil \cdot \rceil$ is the ceiling function which returns the smallest integer greater than the input value, and ${\mathrm{E}_{x}}\left\{ \cdot \right\}$ denotes the expectation over the random variable $x$.
The circularly symmetric complex Gaussian distribution with mean $\mu$ and variance $\sigma^2$ is denoted by ${\cal CN}(\mu,\sigma^2)$.

\section{System Model}
\begin{figure}[t]
\centering
\includegraphics[width=2in]{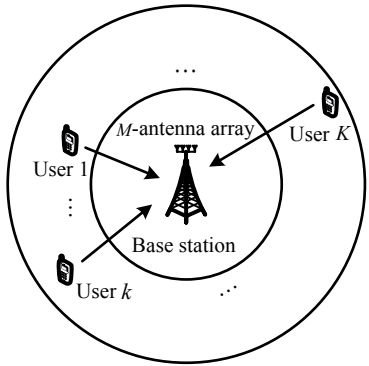}
\caption{The system model of the single-cell uplink communication with one base station and $K$ users.}
\label{NOMA_Uplink_Model}
\end{figure}
\subsection{System Model}
We first consider the uplink\footnote{We restrict ourselves to the uplink NOMA communications\cite{wei2017fairness}, as advanced signal detection/decoding algorithms of NOMA are more affordable at the base station.} of a
single-cell\footnote{We first focus on the ESG analysis for single-cell systems, which serves as a building block for the analyses for multi-cell systems presented in Section \ref{DiscussionsMulticell}.} NOMA system with a single base station (BS) supporting $K$ users, as shown in Fig. \ref{NOMA_Uplink_Model}.
The cell is modeled by a pair of two concentric ring-shaped discs.
The BS is located at the center of the ring-shaped discs with the inner radius of $D_0$ and outer radius of $D$, where all the $K$ users are scattered uniformly within the two concentric ring-shaped discs.
%
%
For the NOMA scheme, all the $K$ users are multiplexed on the same frequency band and time slot, while for the OMA scheme, $K$ users utilize the frequency or time resources orthogonally.
Without loss of generality, we consider frequency division multiple access (FDMA) as a typical OMA scheme.

In this paper, we consider three typical types of communication systems:
\begin{itemize}
  \item SISO-NOMA and SISO-OMA: the BS is equipped with a single-antenna ($M=1$) and all the $K$ users also have a single-antenna.
  \item MIMO-NOMA and MIMO-OMA: the BS is equipped with a multi-antenna array ($M>1$) and all the $K$ users have a single-antenna associated with $K > M$.
  \item Massive MIMO-NOMA (\emph{m}MIMO-NOMA) and massive-MIMO-OMA (\emph{m}MIMO-OMA): the BS is equipped with a large-scale antenna array ($M \to \infty$), while all the $K$ users are equipped with a single antenna, associated with $\frac{M}{K} = \delta < 1$, i.e., the number of antennas $M$ at the BS is lower than the number of users $K$, but with a fixed ratio of $\delta < 1$.
\end{itemize}

\subsection{Signal and Channel Model}
The signal received at the BS is given by
\begin{equation}\label{MIMONOMASystemModel}
{\bf{y}} = \sum\limits_{k = 1}^K {{{\bf{h}}_k}} \sqrt {{p_k}} {x_k} + {\bf{v}},
\end{equation}
where ${\bf{y}}\in \mathbb{C}^{ M \times 1}$, $p_k$ denotes the power transmitted by user $k$, $x_k$ is the normalized modulated symbol of user $k$ with ${\mathrm{E}}\left\{ \left|{x_k}\right|^2 \right\} = 1$, and ${\bf{v}}\sim \mathcal{CN}\left(\mathbf{0},N_0 {{\bf{I}}_M}\right)$ represents the additive white Gaussian noise (AWGN) at the BS with zero mean and covariance matrix of $N_0 {{\bf{I}}_M}$.
To emphasize the impact of the number of users $K$ on the performance gain of NOMA over OMA, we fix the total power consumption of all the uplink users and thus we have
\begin{equation}\label{SumPowerConstraint}
\sum\limits_{k = 1}^K {{p_{k}}}  \le {P_{\rm{max}}},
\end{equation}
where ${P_{\rm{max}}}$ is the maximum total transmit power for all the users.
Note that the sum-power constraint is a commonly adopted assumption in the literature\cite{Vishwanath2003,WangMUG,Xu2017} for simplifying the performance analysis of uplink communications.
In fact, the sum-power constraint is a reasonable assumption for practical cellular communication systems, where a total transmit power limitation is intentionally imposed to limit the ICI.

The uplink (UL) channel vector between user $k$ and the BS is modeled as
\begin{equation}\label{ChannelModel}
{{{\bf{h}}_k}} = \frac{{\bf{g}}_k}{\sqrt{1+d_k^{\alpha}}},
\end{equation}
where ${\bf{g}}_k \in \mathbb{C}^{ M \times 1}$ denotes the Rayleigh fading coefficients, i.e., ${\bf{g}}_k \sim \mathcal{CN}\left(\mathbf{0},{{\bf{I}}_M}\right)$, $d_k$ is the distance between user $k$ and the BS in the unit of meter, and $\alpha$ represents the path loss exponent\footnote{In this paper, we ignore the impact of shadowing to simplify our performance analysis.
Note that, shadowing only introduces an additional power factor to ${{\bf{g}}_k}$ in the channel model in \eqref{ChannelModel}.
Although the introduction of shadowing may change the resulting channel distribution of ${{{\bf{h}}_k}}$, the distance-based channel model is sufficient to characterize the large-scale near-far gain exploited by NOMA, as will be discussed in this paper.}.
We denote the UL channel matrix between all the $K$ users and the BS by ${\bf{H}} = \left[ {{{\bf{h}}_1}, \ldots ,{{\bf{h}}_K}} \right] \in \mathbb{C}^{ M \times K}$.
Note that the system model in \eqref{MIMONOMASystemModel} and the channel model in \eqref{ChannelModel} include the cases of single-antenna and massive-MIMO aided BS associated with $M = 1$ and $M \to \infty$, respectively.
For instance, when $M=1$, ${{{{h}}_k}} = \frac{{{g}}_k}{\sqrt{1+d_k^{\alpha}}}$ denotes the corresponding channel coefficient of user $k$ in single-antenna systems.
We assume that the channel coefficients are independent and identically distributed (i.i.d.) over all the users and antennas.
Since this paper aims for providing some insights concerning the performance gain of NOMA over OMA, we assume that perfect UL CSI knowledge is available at the BS for coherent detection.
%
\subsection{Signal Detection and Resource Allocation Strategy}
To facilitate our performance analyses, we focus our attention on the following efficient signal detection and practical resource allocation strategies.
\subsubsection{Signal detection}
\begin{table}
\caption{Signal Detection Techniques for NOMA and OMA Systems}
\centering
\begin{tabular}{cc|cc}
  \hline
  NOMA system  & Reception technique & OMA system & Reception technique\\ \hline
  SISO-NOMA  & SIC & SISO-OMA & FDMA-SUD \\
  MIMO-NOMA & MMSE-SIC & MIMO-OMA & FDMA-ZF, FDMA-MRC\\
  \emph{m}MIMO-NOMA & MRC-SIC & \emph{m}MIMO-OMA & FDMA-MRC\\\hline
\end{tabular}\label{TransceiverProtocol}
\end{table}

The signal detection techniques adopted in this paper for NOMA and OMA systems are shown in Table \ref{TransceiverProtocol}, which are detailed in the following.

For SISO-NOMA, we adopt the commonly used successive interference cancelation (SIC) receiver \cite{wei2017performance} at the BS, since its performance approaches the capacity of single-antenna systems\cite{Tse2005}.
On the other hand, given that all the users are separated orthogonally by different frequency subbands for SISO-OMA, the simple single-user detection (SUD) technique can be used to achieve the optimal performance.

For MIMO-NOMA, the minimum mean square error criterion based successive interference cancelation (MMSE-SIC) constitutes an appealing receiver algorithm, since its performance approaches the capacity \cite{Tse2005} at an acceptable computational complexity for a finite number of antennas $M$ at the BS.
On the other hand, two types of signal detection schemes are considered for MIMO-OMA, namely FDMA zero forcing (FDMA-ZF) and FDMA maximum ratio combining (FDMA-MRC).
Exploiting the extra spatial degrees of freedom (DoF) attained by multiple antennas at the BS, ZF can be used for multi-user detection (MUD), as its achievable rate approaches the capacity in the high-SNR regime\cite{Tse2005}.
In particular, all the users are categorized into $G = K/M$ groups\footnote{Without loss of generality, we consider the case with $G$ as an integer in this paper.} with each group containing $M$ users.
Then, ZF is utilized for handling the inter-user interference (IUI) within each group and FDMA is employed to separate all the $G$ groups on orthogonal frequency subbands.
In the low-SNR regime, the performance of ZF fails to approach the capacity \cite{Tse2005}, thus a simple low-complexity MRC scheme is adopted for single user detection on each frequency subband.
We note that there is only a single user in each frequency subband of our considered FDMA-MRC aided MIMO-OMA systems, i.e., no user grouping.

With a massive number of UL receiving antennas employed at the BS, we circumvent the excessive complexity of matrix inversion involved in ZF and MMSE detection by adopting the low-complexity MRC-SIC detection \cite{WeiLetter2018} for \emph{m}MIMO-NOMA systems and the FDMA-MRC scheme for \emph{m}MIMO-OMA systems.
Given the favorable propagation property of massive-MIMO systems\cite{Ngo2013}, the orthogonality among the channel vectors of multiple users holds fairly well, provided that the number of users is sufficiently lower than the number of antennas.
Therefore, we can assign $W \ll M$ users to every frequency subband and perform the simple MRC detection while enjoying negligible IUIs in each subband.
In this paper, we consider a fixed ratio between the group size and the number of antennas, namely, $\varsigma = \frac{W}{M} \ll 1$, and assume that the above-mentioned favorable propagation property holds under the fixed ratio $\varsigma$ considered.

\subsubsection{Resource allocation strategy}\label{ResourceAllocation}
To facilitate our analytical study in this paper, we consider an equal resource allocation strategy for both NOMA and OMA schemes.
In particular, equal power allocation is adopted for NOMA schemes\footnote{As shown in \cite{Vaezi2018non}, allocating a higher power to the user with the worse channel is not necessarily required in NOMA\cite{Vaezi2018non}.}.
On the other hand, equal power and frequency allocation is adopted for OMA schemes.
%
Note that the equal resource allocation is a typical selected strategy for applications bearing only a limited system overhead, e.g. machine-type communications (MTC).
%

We note that beneficial user grouping design is important for the MIMO-OMA system relying on FDMA-ZF and for the \emph{m}MIMO-OMA system using FDMA-MRC.
In general, finding the optimal user grouping strategy is an NP-hard problem and the performance analysis based on the optimal user grouping strategy is generally intractable.
Furthermore, the optimal SIC decoding order of NOMA in multi-antenna and massive-MIMO systems is still an open problem in the literature, since the channel gains on different antennas are usually diverse.
To avoid tedious comparison and to facilitate our performance analysis, we adopt a random user grouping strategy for the OMA systems considered and a fixed SIC decoding order for the NOMA systems investigated.
%
%
In particular, we randomly select $M$ and $W$ users for each group on each frequency subband for the MIMO-OMA and \emph{m}MIMO-OMA systems, respectively.
For NOMA systems, without loss of generality, we assume $\left\| {{{\bf{h}}_1}} \right\| \ge \left\| {{{\bf{h}}_2}} \right\|, \ldots, \ge\left\| {{{\bf{h}}_K}} \right\|$, that the users are indexed based on their channel gains, and the SIC/MMSE-SIC/MRC-SIC decoding order\footnote{Note that, in general, the adopted decoding order is not the optimal SIC decoding order for maximizing the achievable sum-rate of the considered MIMO-NOMA and \emph{m}MIMO-NOMA systems.} at the BS is $1,2,\ldots,K$.
%
%
Additionally, to unveil insights about the performance gain of NOMA over OMA, we assume that there is no error propagation during SIC/MMSE-SIC/MRC-SIC decoding at the BS.
\section{ESG of SISO-NOMA over SISO-OMA}
In this section, we first derive the ergodic sum-rate of SISO-NOMA and SISO-OMA.
Then, the asymptotic ESG of SISO-NOMA over SISO-OMA is discussed under different scenarios.

\subsection{Ergodic Sum-rate of SISO-NOMA and SISO-OMA}
When decoding the messages of user $k$, the interferences imposed by users $1,2,\ldots,(k-1)$ have been canceled in the SISO-NOMA system by SIC reception.
Therefore, the instantaneous achievable data rate of user $k$ in the SISO-NOMA system considered is given by:
\begin{equation}\label{SISONOMAIndividualAchievableRate}
R_{k}^{\rm{SISO-NOMA}} = {\ln}\left( 1 + \frac{{{p_k}{{\left| {{h_k}} \right|}^2}}}{{\sum\limits_{i = k + 1}^K {{p_i}{{\left| {{h_i}} \right|}^2}}  + {N_0}}} \right).
\end{equation}
On the other hand, in the SISO-OMA system considered, user $k$ is allocated to a subband exclusively, thus there is no inter-user interference (IUI).
As a result, the instantaneous achievable data rate of user $k$ in the SISO-OMA system considered is given by:
\begin{equation}\label{SISOOMAIndividualAchievableRate}
R_{k}^{{\rm{SISO-OMA}}} = f_k {\ln}\left(1+ {\frac{{{p_{k }}{{\left| {{{{h}}_k}} \right|}^2}}}{{f_k N_0}}} \right),
\end{equation}
with ${p_k}$ and $f_k$ denoting the power allocation and frequency allocation of user $k$.
Note that we consider a normalized frequency bandwidth for both the NOMA and OMA schemes in this paper, i.e., $\sum \limits_{k=1}^{K} f_k = 1$.
Under the identical resource allocation strategy, i.e., for ${{p_{k}}} = \frac{{P_{\rm{max}}}}{K}$ and $f_k = 1/K$, we have the instantaneous sum-rate of SISO-NOMA and SISO-OMA given by
\begin{align}
R_{\rm{sum}}^{{\rm{SISO-NOMA}}} &= \sum\limits_{k = 1}^K R_{k}^{\rm{SISO-NOMA}}=  {\ln}\left( 1 + \frac{{P_{\rm{max}}}}{K {N_0}} \sum\limits_{k = 1}^K {{\left| {{{h}_k}} \right|}^2} \right) \;\text{and} \label{InstantSumRateSISONOMA}\\
R_{\rm{sum}}^{{\rm{SISO-OMA}}} &= \sum\limits_{k = 1}^K R_{k}^{\rm{SISO-OMA}}= \frac{1}{K}\sum\limits_{k = 1}^K {\ln}\left( 1 + \frac{P_{\rm{max}}}{{N_0}} {{\left| {{{h}_k}} \right|}^2} \right), \label{InstantSumRateSISOOMA}
\end{align}
respectively.
%

Given the instantaneous sum-rates in \eqref{InstantSumRateSISONOMA} and \eqref{InstantSumRateSISOOMA}, firstly we have to investigate the channel gain distribution before embarking on the derivation of the corresponding ergodic sum-rates.
Since all the users are scattered uniformly across the pair of concentric rings between the inner radius of $D_0$ and the outer radius of $D$ in Fig. \ref{NOMA_Uplink_Model}, the cumulative distribution function (CDF) of the channel gain\footnote{As mentioned before, we assumed that the channel gains of all the users are ordered as $\left| {{{{h}}_1}} \right| \ge \left| {{{{h}}_2}} \right|, \ldots, \ge\left| {{{{h}}_K}} \right|$ in Section \ref{ResourceAllocation}.
However, as shown in \eqref{InstantSumRateSISONOMA}, the system sum-rate for the considered SISO-NOMA system is actually independent of the SIC decoding order.
Therefore, we can safely assume that all the users have i.i.d. channel distribution, which does not affect the performance analysis results.
In the sequel of this paper, the subscript $k$ is dropped without loss of generality.} ${{\left| {{h}} \right|}^2}$ is given by
\begin{equation}\label{ChannelDistributionCDF}
{F_{{{\left| {{h}} \right|}^2}}}\left( x \right) = \int_{D_0}^D {\left( {1 - {e^{ - \left( {1 + {z^\alpha }} \right)x}}} \right)} {f_{{d}}}\left( z \right)dz,
\end{equation}
where ${f_{{d}}}\left( z \right) = \frac{2z}{D^2 - D_0^2}$, $D_0 \le z \le D$, denotes the probability density function (PDF) for the random distance $d$.
%
With the Gaussian-Chebyshev quadrature approximation\cite{abramowitz1964handbook}, the CDF and PDF of ${{\left| {{h}} \right|}^2}$ can be approximated by
\begin{align}
{F_{{{\left| {{h}} \right|}^2}}}\left( x \right) &\approx 1 - \frac{1}{D+D_0}\sum\limits_{n = 1}^N {{\beta _n}{e^{ - {c_n}x}}} \; \text{and} \label{SISOChannelDistributionPDF}\\
{f_{{{\left| {{h}} \right|}^2}}}\left( x \right) &\approx \frac{1}{D+D_0}\sum\limits_{n = 1}^N {{\beta _n}{c_n}{e^{ - {c_n}x}}}, x \ge 0, \label{SISOChannelDistributionCDF}
\end{align}
respectively, where the parameters in \eqref{SISOChannelDistributionPDF} and \eqref{SISOChannelDistributionCDF} are:
\begin{align}\label{BetaCn}
{\beta_n} &= \frac{\pi }{N}\left| {\sin \frac{{2n \hspace{-1mm}-\hspace{-1mm} 1}}{{2N}}\pi } \right|\left( {\frac{D\hspace{-1mm}-\hspace{-1mm}D_0}{2}\cos \frac{{2n \hspace{-1mm}-\hspace{-1mm} 1}}{{2N}}\pi  + \frac{D\hspace{-1mm}+\hspace{-1mm}D_0}{2}} \right) \;\text{and}\notag\\
{c_n} &= 1 + {\left( {\frac{D\hspace{-1mm}-\hspace{-1mm}D_0}{2}\cos \frac{{2n \hspace{-1mm}-\hspace{-1mm} 1}}{{2N}}\pi  + \frac{D\hspace{-1mm}+\hspace{-1mm}D_0}{2}} \right)^\alpha },
\end{align}
while $N$ denotes the number of terms for integral approximation.
The larger $N$, the higher the approximation accuracy becomes.
%

Based on \eqref{InstantSumRateSISONOMA}, the ergodic sum-rate of the SISO-NOMA system considered is defined as:
\begin{align}\label{ErgodicSumRateSISONOMADefine}
\overline{R_{{\mathrm{sum}}}^{{\mathrm{SISO - NOMA}}}} = {{\mathrm{E}}_{\mathbf{H}}}\left\{ {R_{{\mathrm{sum}}}^{{\mathrm{SISO - NOMA}}}} \right\} = {{\mathrm{E}}_{\mathbf{H}}}\left\{ {\ln \left( {1 + \frac{{{P_{{\mathrm{max}}}}}}{{K{N_0}}}\sum\limits_{k = 1}^K {{{\left| {{h_k}} \right|}^2}} } \right)} \right\},
\end{align}
where the expectation ${{\mathrm{E}}_{\mathbf{H}}}\left\{ \cdot \right\}$ is averaged over both the large-scale fading and small-scale fading in the overall channel matrix ${\mathbf{H}}$.
For a large number of users, i.e., $K\rightarrow \infty$, the sum of channel gains of all the users within the $\ln\left(\cdot\right)$ in \eqref{ErgodicSumRateSISONOMADefine} becomes a deterministic value due to the strong law of large number, i.e., $\mathop {\lim } \limits_{K \to \infty } {\frac{1}{K}\sum\limits_{k = 1}^K {{{\left| {{h_k}} \right|}^2}} } = \overline{{{\left| {{h}} \right|}^2}}$, where $\overline{{{\left| {{h}} \right|}^2}}$ denotes the average channel power gain and it is given by
\begin{equation}\label{Mean_ChannelPowerGainSISO}
\overline{{{\left| {{h}} \right|}^2}} = \int_0^\infty  {x{f_{{{\left| {{h}} \right|}^2}}}\left( x \right)} dx
\approx \frac{1}{D+D_0}\sum\limits_{n = 1}^N {\frac{{{\beta _n}}}{{{c_n}}}}.
\end{equation}
Therefore, the asymptotic ergodic sum-rate of the SISO-NOMA system considered is given by
%
\begin{align}\label{ErgodicSumRateSISONOMA}
\hspace{-2mm}\mathop {\lim }\limits_{K \to \infty }  \overline{R_{\rm{sum}}^{{\rm{SISO-NOMA}}}}
& \mathop  = \limits^{(a)}  {{\mathrm{E}}_{\mathbf{H}}}\left\{ \mathop {\lim }\limits_{K \to \infty } {R_{\rm{sum}}^{{\rm{SISO-NOMA}}}} \right\}
= {\ln}\left( 1 + \frac{P_{\rm{max}}}{N_0} \overline{{{\left| {{h}} \right|}^2}} \right) \notag\\
& \approx {\ln}\left(\hspace{-1mm} {1 \hspace{-1mm}+ \hspace{-1mm} \frac{P_{\rm{max}}}{{{\left(D\hspace{-1mm}+\hspace{-1mm}D_0\right)}N_0}}\sum\limits_{n = 1}^N \hspace{-1mm} {\frac{{{\beta _n}}}{{{c_n}}}} } \hspace{-1mm}\right),
\end{align}
where the equality $(a)$ is due to the bounded convergence theorem\cite{bartle2014elements} and owing to the finite channel capacity.
Note that for a finite number of users $K$, the asymptotic ergodic sum-rate in \eqref{ErgodicSumRateSISONOMA} serves as an upper bound for the actual ergodic sum-rate in \eqref{ErgodicSumRateSISONOMADefine}, i.e., we have $\mathop {\lim }\limits_{K \to \infty }  \overline{R_{\rm{sum}}^{{\rm{SISO-NOMA}}}} \ge \overline{R_{\rm{sum}}^{{\rm{SISO-NOMA}}}}$, owing to the concavity of the logarithmic function and the Jensen's inequality.
In the Section \ref{Simulations}, we will show that the asymptotic analysis in \eqref{ErgodicSumRateSISONOMA} is also accurate for a finite value of $K$ and becomes tighter upon increasing $K$.

Similarly, based on \eqref{InstantSumRateSISOOMA}, we can obtain the ergodic sum-rate of the SISO-OMA system as follows:
\begin{align}\label{ErgodicSumRateSISOOMA}
\hspace{-2mm}\overline{R_{\rm{sum}}^{{\rm{SISO-OMA}}}}
& = {{\mathrm{E}}_{\mathbf{H}}}\left\{ \frac{1}{K}\sum\limits_{k = 1}^K {\ln}\left( 1 + \frac{P_{\rm{max}}}{{N_0}} {{\left| {{{h}_k}} \right|}^2} \right) \right\} \notag\\
& \mathop = \limits^{(a)} \int_0^\infty  {{{\ln }}\left( {1 + \frac{P_{\rm{max}}}{N_0}x} \right){f_{{{\left| {{h}} \right|}^2}}}\left( x \right)} dx \notag\\
& = \frac{1}{\left(D\hspace{-1mm}+\hspace{-1mm}D_0\right)}\sum\limits_{n = 1}^N {{\beta _n}{e^{\frac{{{c_n N_0}}}{{P_{\rm{max}}}}}} {\mathcal{E}_1}\left( {\frac{{{c_n N_0}}}{{P_{\rm{max}}}}} \right)},
\end{align}
where ${\mathcal{E}_l}\left( x \right) = \int_1^\infty  {\frac{{{e^{ - xt}}}}{t^l}} dt$ denotes the $l$-order exponential integral\cite{abramowitz1964handbook}.
The equality $(a)$ in \eqref{ErgodicSumRateSISOOMA} is obtained since all the users have i.i.d. channel distributions.
Note that in contrast to SISO-NOMA, $\overline{R_{\rm{sum}}^{{\rm{SISO-OMA}}}}$ in \eqref{ErgodicSumRateSISOOMA} is applicable to SISO-OMA supporting an arbitrary number of users.

\subsection{ESG in Single-antenna Systems}
Comparing \eqref{ErgodicSumRateSISONOMA} and \eqref{ErgodicSumRateSISOOMA}, the asymptotic ESG of SISO-NOMA over SISO-OMA with ${K \rightarrow \infty}$ can be expressed as follows:
\begin{align}\label{EPGSISOERA}
\hspace{-2mm}\mathop {\lim }\limits_{K \rightarrow \infty}  \overline{G^{\rm{SISO}}} &= \mathop {\lim }\limits_{K \to \infty }  \overline{R_{\rm{sum}}^{{\rm{SISO-NOMA}}}} - \overline{R_{\rm{sum}}^{{\rm{SISO-OMA}}}} \notag\\
&\approx {\ln}\left( {1 + \frac{{P_{\rm{max}}}}{{{\left(D+D_0\right)N_0}}}\sum\limits_{n = 1}^N {\frac{{{\beta _n}}}{{{c_n}}}} } \right)
- \frac{1}{\left(D\hspace{-1mm}+\hspace{-1mm}D_0\right)}\sum\limits_{n = 1}^N {{\beta _n}{e^{\frac{{{c_n}N_0}}{{P_{\rm{max}}}}}} {\mathcal{E}_1}\left( {\frac{{{c_n}N_0}}{{P_{\rm{max}}}}} \right)}.
\end{align}
Then, in the high-SNR regime, we can approximate the asymptotic ESG\footnote{Under the sum-power constraint, the system SNR directly depends on the total system power budget ${P_{\rm{max}}}$, and thus the system SNR and ${P_{\rm{max}}}$ are interchangeably in this paper.} in \eqref{EPGSISOERA} by applying $\mathop {\lim }\limits_{x \to 0} {\mathcal{E}_1}\left( x \right)\approx - \ln \left( x \right) - \gamma $ \cite{abramowitz1964handbook} as
\begin{equation}\label{EPGSISOERA2}
\mathop {\lim }\limits_{K \rightarrow \infty, {P_{\rm{max}}} \rightarrow \infty}  \overline{G^{\rm{SISO}}} \approx \vartheta \left( {D,{D_0}} \right) + \gamma,
\end{equation}
where $\vartheta \left( {D,{D_0}} \right)$ is given by
\begin{equation}\label{NearFarDiversity}
\vartheta \left( {D,{D_0}} \right)
 = \ln \left( {\frac{{\sum\limits_{n = 1}^N {\left( {\frac{1}{{{c_n}}}} \right)\frac{{{\beta _n}}}{{\left( {D + {D_0}} \right)}}} }}{{\mathop \Pi \limits_{n = 1}^N {{\left( {\frac{1}{{{c_n}}}} \right)}^{\frac{{{\beta _n}}}{{\left( {D + {D_0}} \right)}}}}}}} \right)
\end{equation}
and $\gamma = 0.57721$ is the Euler-Mascheroni constant\cite{abramowitz1964handbook}.
Based on the weighted arithmetic and geometric means (AM-GM) inequality\cite{kedlaya1994proof}, we can observe that $\vartheta \left( {D,{D_0}} \right) \ge 0$.
This implies that $\mathop {\lim }\limits_{K \rightarrow \infty, {P_{\rm{max}}} \rightarrow \infty}  \overline{G^{\rm{SISO}}}>0$ and SISO-NOMA provides a higher asymptotic ergodic sum-rate than SISO-OMA in the system considered.

To further simplify the expression of ESG, we consider path loss exponents $\alpha$ in the range of $\alpha \in \left[3,6\right]$ in \eqref{BetaCn}, which usually holds in urban environments\cite{Access2010}.
As a result, $c_n \gg 1$.
Hence, $\vartheta \left( {D,{D_0}} \right)$ in \eqref{NearFarDiversity} can be further simplified as follows:
\begin{equation}\label{NearFarDiversity2}
\vartheta \left( {D,{D_0}} \right)\approx \vartheta \left( \eta  \right) =\ln \left( {\frac{{\frac{\pi }{{N\left( {1 + \eta } \right)}}\sum\limits_{n = 1}^N {\left[ {{\lambda _n}\left( \eta  \right)} \right]^{1 - \alpha }}\left| {\sin \frac{{2n - 1}}{{2N}}\pi } \right| }}{{\mathop \Pi \limits_{n = 1}^N {{\left[ {{\lambda_n}\left( \eta  \right)} \right]}^{ - \frac{{\alpha \pi {\lambda_n}\left( \eta  \right)}}{{N\left( {1 + \eta } \right)}}\left| {\sin \frac{{2n - 1}}{{2N}}\pi } \right|}}}}} \right),
\end{equation}
where ${\lambda _n}\left( \eta  \right) = \left( {\frac{{\eta - 1}}{2}\cos \left( {\frac{{2n - 1}}{{2N}}\pi } \right) + \frac{{\eta +1}}{2}} \right) \in \left[1, \eta\right)$.
The normalized cell size of $\eta = \frac{D}{D_0} \ge 1$ is the ratio between the outer radius $D$ and the inner radius ${D_0}$, which also serves as a metric of the path loss discrepancy.

We can see that the asymptotic ESG of SISO-NOMA over SISO-OMA in \eqref{EPGSISOERA2} is composed of two components, i.e., $\vartheta \left( {D,{D_0}} \right)$ and $\gamma$.
As observed in \eqref{NearFarDiversity2}, the former component of $\vartheta \left( {D,{D_0}} \right) \approx \vartheta \left( \eta  \right) $ only depends on the normalized cell size of $\eta = \frac{D}{D_0}$ instead of the absolute values of $D$ and ${D_0}$.
In fact, it can characterize the \emph{large-scale near-far gain} attained by NOMA via exploiting the discrepancy in distances among NOMA users.
Interestingly, for the extreme case that all the users are randomly deployed on a circle, i.e., $D = D_0$, we have $\eta = 1$, ${\lambda _n}\left( \eta  \right) = 1$, and $\vartheta \left( \eta  \right) = 0$.
In other words, the large-scale near-far gain disappears, when all the users are of identical distance away from the BS.
With the aid of $\vartheta \left( \eta  \right) = 0$, we can observe in \eqref{EPGSISOERA2} that the performance gain achieved by NOMA is a constant value of $\gamma = 0.57721$ nat/s/Hz.
Since all the users are set to have the same distance when $D = D_0$, the minimum asymptotic ESG $\gamma$ arising from the small-scale Rayleigh fading is named as the \emph{small-scale fading gain} in this paper.
In fact, in the asymptotic case associated with $K \rightarrow \infty$ and ${P_{\rm{max}}} \rightarrow \infty$, SISO-NOMA provides \emph{at least $\gamma =  0.57721$} nat/s/Hz spectral efficiency gain over SISO-OMA for an arbitrary cell size in Rayleigh fading channels.
%
Additionally, we can see that the ESG of SISO-NOMA over SISO-OMA is saturated in the high-SNR regime.
This is because the instantaneous sum-rates of both the SISO-NOMA system in \eqref{InstantSumRateSISONOMA} and the SISO-OMA system in \eqref{InstantSumRateSISOOMA} increase logarithmically with ${P_{\rm{max}}} \rightarrow \infty$.
%

%


%
To visualize the large-scale near-far gain, we illustrate the asymptotic ESG in \eqref{EPGSISOERA2} versus $D$ and $D_0$ in Fig. \ref{ESG_ERA}.
We can observe that when $\eta = 1$,  the large-scale near-far gain disappears and the asymptotic ESG is bounded from below by its minimum value of $\gamma = 0.57721$ nat/s/Hz due to the small-scale fading gain.
Additionally, for different values of $D$ and $D_0$ but with a fixed $\eta = \frac{D}{D_0}$, SISO-NOMA offers the same ESG compared to SISO-OMA.
This is because as predicted in \eqref{NearFarDiversity2}, the large-scale near-far gain only depends on the normalized cell size $\eta$.
More importantly, we can observe that the large-scale near-far gain increases with the normalized cell size $\eta$.
In fact, for a larger normalized cell size $\eta$, the heterogeneity in the large-scale fading among users becomes higher and SISO-NOMA attains a higher near-far gain, hence improving the sum-rate performance.

%
%
%

\begin{figure}[t]
\centering
\includegraphics[width=3.5in]{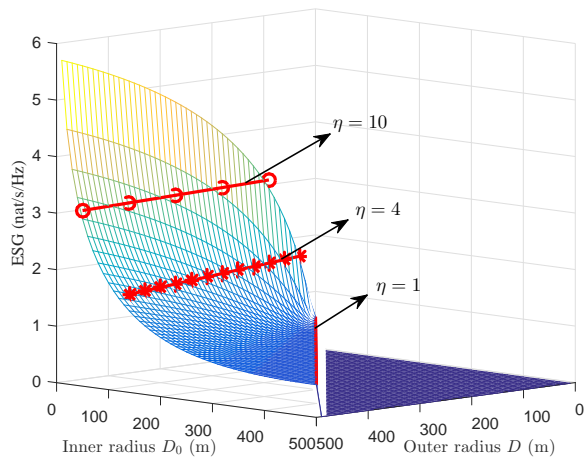}
\caption{The asymptotic ESG in \eqref{EPGSISOERA2} under equal resource allocation versus $D$ and $D_0$ with $K \rightarrow \infty$ and ${P_{\rm{max}}} \rightarrow \infty$.}
\label{ESG_ERA}
\end{figure}

\begin{Remark}
Note that it has been analytically shown in \cite{Dingtobepublished} that two users with a large distance difference (or equivalently channel gain difference) are preferred to be paired.
This is consistent with our conclusion in this paper, where a larger normalized cell size $\eta$ enables a higher ESG of NOMA over OMA.
However, \cite{Dingtobepublished} only considered a pair of two NOMA users.
%
%
In this paper, we analytically obtain the ESG of NOMA over OMA for a more general NOMA system supporting a large number of UL users.
More importantly, we identify two kinds of gains in the ESG derived and reveal their different behaviours.
\end{Remark}
\section{ESG of MIMO-NOMA over MIMO-OMA}
In this section, the ergodic sum-rates of MIMO-NOMA and MIMO-OMA associated with FDMA-ZF as well as FDMA-MRC are firstly analyzed. Then, the asymptotic ESGs of MIMO-NOMA over MIMO-OMA with FDMA-ZF and FDMA-MRC detection are investigated.
\subsection{Ergodic Sum-rate of MIMO-NOMA with MMSE-SIC}
Let us consider that an $M$-antenna BS serves $K$ single-antenna non-orthogonal users relying on MIMO-NOMA.
The BS employs MMSE-SIC detection for retrieving the messages of all the users.
The instantaneous achievable data rate of user $k$ in the MIMO-NOMA system relying on MMSE-SIC detection\footnote{The derivation of individual rates in \eqref{MIMONOMAIndividualAchievableRate} for MMSE-SIC detection of MIMO-NOMA is based on the matrix inversion lemma: \[\log \left| {{\bf{A}} + {\bf{h}}{{\bf{h}}^{\rm{H}}}} \right| - \log \left| {\bf{A}} \right| = \log \left| {1 + {{\bf{h}}^{\rm{H}}}{{\bf{A}}^{ - 1}}{\bf{h}}} \right|.\]
	Interested readers are referred to \cite{Tse2005} for a detailed derivation.} is given by\cite{Tse2005}:
\begin{align}\label{MIMONOMAIndividualAchievableRate}
R_{k}^{{\rm{MIMO-NOMA}}} = \ln \left| {{{\bf{I}}_M} + \frac{1}{{{N_0}}}\sum\limits_{i = k}^K {{p_i}{{\bf{h}}_i}{\bf{h}}_i^{\rm{H}}} } \right|
- \ln \left| {{{\bf{I}}_M} + \frac{1}{{{N_0}}}\sum\limits_{i = k + 1}^K {{p_i}{{\bf{h}}_i}{\bf{h}}_i^{\rm{H}}} } \right|.
\end{align}
As a result, the instantaneous sum-rate of MIMO-NOMA is obtained as
\begin{equation}\label{InstantSumRateMIMONOMA}
R_{\rm{sum}}^{{\rm{MIMO-NOMA}}} = \sum\limits_{k = 1}^K R_{k}^{{\rm{MIMO-NOMA}}}
= {\ln}\left| {{{\bf{I}}_M} + \frac{1}{{{N_0}}}\sum\limits_{k = 1}^K p_k {{{\bf{h}}_k}{\bf{h}}_k^{\rm{H}}} } \right|.
\end{equation}

In fact, MMSE-SIC is capacity-achieving \cite{Tse2005} and \eqref{InstantSumRateMIMONOMA} is the channel capacity for a given instantaneous channel matrix ${\bf{H}}$\cite{GoldsmithMIMOCapacity2003}.
In general, it is a challenge to obtain a closed-form expression for the instantaneous channel capacity above due to the determinant of the summation of matrices in \eqref{InstantSumRateMIMONOMA}.
To provide some insights, in the following theorem, we consider an asymptotically tight upper bound for the achievable sum-rate in \eqref{InstantSumRateMIMONOMA} associated with $K \to \infty$.

\begin{Thm}\label{Theorem1}
For the MIMO-NOMA system considered in \eqref{MIMONOMASystemModel} relying on MMSE-SIC detection, given any power allocation strategy ${\bf{p}} = \left[ {{p_1}, \ldots ,{p_K}} \right]$, the achievable sum-rate in \eqref{InstantSumRateMIMONOMA} is upper bounded by
\begin{equation}\label{InstantSumRateMIMONOMA_UpperBound}
\hspace{-2mm}R_{\rm{sum}}^{{\rm{MIMO-NOMA}}} \le M{\ln}\left( {1 + \frac{1}{{M{N_0}}}\sum\limits_{k = 1}^K {{p_k}{{\left\| {{{\bf{h}}_k}} \right\|}^2}} } \right).
\end{equation}
This upper bound is asymptotically tight, when $K \to \infty$, i.e.,
\begin{equation}\label{AsympInstantSumRateMIMONOMA}
\hspace{-2mm}\mathop {\lim }\limits_{K \rightarrow \infty} \hspace{-1mm} R_{\rm{sum}}^{{\rm{MIMO-NOMA}}}
\hspace{-0.5mm}=\hspace{-0.5mm} \mathop {\lim }\limits_{K \rightarrow \infty} M{\ln}\hspace{-0.5mm}\left( \hspace{-0.5mm}{1 \hspace{-0.5mm}+\hspace{-0.5mm} \frac{1}{{M{N_0}}}\sum\limits_{k = 1}^K {{p_k}{{\left\| {{{\bf{h}}_k}} \right\|}^2}} } \hspace{-0.5mm}\right)\hspace{-0.5mm}.
\end{equation}
\end{Thm}
\begin{proof}
Please refer to Appendix A for the proof of Theorem \ref{Theorem1}.
\end{proof}

Now, given the instantaneous achievable sum-rate obtained in \eqref{AsympInstantSumRateMIMONOMA}, we proceed to calculate the ergodic sum-rate.
Given the distance from a user to the BS as $d$, the channel gain ${{\left\| {{\mathbf{h}}} \right\|}^2}$ follows the Gamma distribution\cite{yang2017noma}, whose conditional PDF and CDF are given by\footnote{Similar to \eqref{SISOChannelDistributionPDF} and \eqref{SISOChannelDistributionCDF}, we can safely assume that all the users have i.i.d. channel distribution within the cell and drop the subscript $k$ in \eqref{MIMOChannelDistribution}, since the system sum-rate in \eqref{InstantSumRateMIMONOMA} is independent of the MMSE-SIC decoding order\cite{Tse2005}.}
\begin{equation}\label{MIMOChannelDistribution}
{f_{{{\left\| {{\mathbf{h}}} \right\|}^2 | d}}}\left( x \right) = {{\rm{Gamma}}} \left(M, 1+ d^\alpha,x\right)\; \text{and} \;
{F_{{{\left\| {{\mathbf{h}}} \right\|}^2 | d}}}\left( x \right) = \frac{{\gamma_{L} \left( {M,\left( {1 + d^\alpha } \right)x} \right)}}{\Gamma \left( M \right)},
\end{equation}
respectively, where ${\rm{Gamma}} \left(M, \lambda,x\right) = \frac{{{{ \lambda }^M}{x^{M - 1}}{e^{ -  \lambda x}}}}{\Gamma \left( M \right)}$ denotes the PDF of a random variable obeying a Gamma distribution, ${\Gamma \left( M \right)}$ denotes the Gamma function, and ${\gamma_{L} \left( {M,\left( {1 + d^\alpha } \right)x} \right)}$ denotes the lower incomplete Gamma function.
Then, the CDF of the channel gain ${{\left\| {{\mathbf{h}}} \right\|}^2}$ can be obtained by
\begin{equation}
{F_{{{\left\| {{\mathbf{h}}} \right\|}^2}}}\left( x \right) = \int_{D_0}^D \frac{{\gamma_{L} \left( {M,\left( {1 + d^\alpha } \right)x} \right)}}{{\Gamma \left(M\right)}} {f_{{d}}}\left( z \right)dz.
\end{equation}
By applying the Gaussian-Chebyshev quadrature approximation\cite{abramowitz1964handbook}, the CDF and PDF of ${{\left\| {{\mathbf{h}}} \right\|}^2}$ can be written as
\begin{align}
\hspace{-2mm}{F_{{{\left\| {{\mathbf{h}}} \right\|}^2}}}\left( x \right) &\approx 1 - \frac{1}{D+D_0}\sum\limits_{n = 1}^N \frac{{{\beta _n}\gamma_{L} \left( {M,{c_n}x} \right)}}{{\Gamma \left(M\right)}} \; \text{and} \notag\\
\hspace{-2mm}{f_{{{\left\| {{\mathbf{h}}} \right\|}^2}}}\left( x \right) &\approx \frac{1}{D+D_0}\sum\limits_{n = 1}^N {{\beta _n} {\rm{Gamma}} \left(M, c_n,x\right)}, x \ge 0,
\end{align}
respectively, where ${\beta _n}$ and ${c_n}$ are given in \eqref{BetaCn}.

According to \eqref{AsympInstantSumRateMIMONOMA}, given the equal resource allocation strategy, i.e., ${{p_{k}}} = \frac{{P_{\rm{max}}}}{K}$, the asymptotic ergodic sum-rate of MIMO-NOMA associated with $K\rightarrow \infty$ can be obtained as follows:
\begin{align} \label{ErgodicSumRateMIMONOMA}
\mathop {\lim }\limits_{K \to \infty } \overline{R_{\rm{sum}}^{{\rm{MIMO-NOMA}}}}
& = \mathop {\lim }\limits_{K \to \infty } {{\mathrm{E}}_{\mathbf{H}}}\left\{ {R_{{\mathrm{sum}}}^{{\mathrm{MIMO - NOMA}}}} \right\} = M{\ln}\left( 1 + \frac{{P_{\rm{max}}}}{MN_0} \overline{{{\left\| {{\mathbf{h}}} \right\|}^2}} \right) \\
& \approx M{\ln}\left(\hspace{-1mm}{1+ \frac{{P_{\rm{max}}}}{{{\left(D\hspace{-1mm}+\hspace{-1mm}D_0\right)N_0}}}
\sum\limits_{n = 1}^N \hspace{-1mm}{\frac{{{\beta _n}}}{{{c_n}}}} } \hspace{-1mm}\right),\notag
\end{align}
where $\overline{{{\left\| {{\mathbf{h}}} \right\|}^2}}$ denotes the average channel gain, which is given by
\begin{equation}\label{Mean_ChannelPowerGainMIMO}
\overline{{{\left\| {{\mathbf{h}}} \right\|}^2}} = \int_0^\infty  {x{f_{{{\left\| {{\mathbf{h}}} \right\|}^2}}}\left( x \right)} dx
\approx \frac{M}{D+D_0}\sum\limits_{n = 1}^N {\frac{{{\beta _n}}}{{{c_n}}}}.
\end{equation}

\begin{Remark} \label{Remark2}
Comparing \eqref{ErgodicSumRateSISONOMA} and \eqref{ErgodicSumRateMIMONOMA}, we can observe that for a sufficiently large number of users, the considered MIMO-NOMA system is asymptotically equivalent to a SISO-NOMA system with $M$-fold increases in DoF and an equivalent average channel gain of $\overline{{{\left\| {{{\bf{h}}}} \right\|}^2}}$ in each DoF.
Intuitively, when the number of UL receiver antennas at the BS, $M$, is much smaller than the number of users, $K \to \infty$, which corresponds to the extreme asymmetric case of MIMO-NOMA, the multi-antenna BS behaves asymptotically in the same way as a single-antenna BS.
Additionally, when $K \gg M$, due to the diverse channel directions of all the users, the received signals fully span the $M$-dimensional signal space\cite{WangMUG}.
Therefore, MIMO-NOMA using MMSE-SIC reception can fully exploit the system's spatial DoF, $M$, and its performance can be approximated by that of a SISO-NOMA system with $M$-fold DoF.
\end{Remark}

\subsection{Ergodic Sum-rate of MIMO-OMA with FDMA-ZF}
Upon installing more UL receiver antennas at the BS, ZF can be employed for MUD and the MIMO-OMA system using FDMA-ZF can accommodate $M$ users on each frequency subband.
As mentioned before, we adopt a random user grouping strategy for the MIMO-OMA system using FDMA-ZF detection, where we randomly select $M$ users as a group and denote the composite channel matrix of the $g$-th group by ${\bf{H}}_g = \left[ {{{\bf{h}}_{(g-1)M+1}},{{\bf{h}}_{(g-1)M+2}}, \ldots ,{{\bf{h}}_{gM}}} \right] \in \mathbb{C}^{ M \times M}$.
Then, the instantaneous achievable data rate of user $k$ in the MIMO-OMA system is given by
\begin{equation}\label{MIMOOMAZFIndividualAchievableRate}
R_{k,\rm{FDMA-ZF}}^{{\rm{MIMO-OMA}}} = f_g {\ln}\left(1+ {\frac{{{p_{k }}{{\left| {{{\mathbf{w}_{g,k}^{\rm{H}}\mathbf{h}}_k}} \right|}^2}}}{{f_g N_0}}} \right),
\end{equation}
where $f_g$ denotes the normalized frequency allocation for the $g$-th group.
The vector $\mathbf{w}_{g,k} \in \mathbb{C}^{ M \times 1}$ denotes the normalized ZF detection vector for user $k$ with ${\left\| {{{\mathbf{w}_{g,k}}}} \right\|}^2 = 1$, which is obtained based on the pseudoinverse of the composite channel matrix ${\bf{H}}_g$ in the $g$-th user group\cite{Tse2005}.

Given the equal resource allocation strategy, i.e., ${{p_{k}}} = \frac{{P_{\rm{max}}}}{K}$ and $f_g = 1/G = \frac{M}{K}$, the instantaneous sum-rate of MIMO-OMA using FDMA-ZF can be formulated as:
\begin{equation}\label{InstantSumRateMIMOOMAZF}
R_{\rm{sum,FDMA-ZF}}^{{\rm{MIMO-OMA}}}
= \sum\limits_{k = 1}^K R_{k,\rm{FDMA-ZF}}^{{\rm{MIMO-OMA}}}
= \frac{M}{K}\sum\limits_{k = 1}^K {\ln}\left(1+ \frac{{P_{\rm{max}}}}{MN_0}{\left| {{{\mathbf{w}_{g,k}^{\rm{H}}\mathbf{h}}_k}} \right|}^2 \right).
\end{equation}
Since ${\left\| {{{\mathbf{w}_{g,k}}}} \right\|}^2 = 1$ and ${\bf{g}}_k \sim \mathcal{CN}\left(\mathbf{0},{{\bf{I}}_M}\right)$, we have ${{{{\mathbf{w}_{g,k}^{\rm{H}}\mathbf{g}}_k}}} \sim \mathcal{CN}\left({0},1\right)$ \cite{Tse2005}.
As a result, ${\left| {{{\mathbf{w}_{g,k}^{\rm{H}}\mathbf{h}}_k}} \right|^2}$ in \eqref{InstantSumRateMIMOOMAZF} has an identical distribution with ${{\left| {{h}} \right|}^2}$, i.e.,  its CDF and PDF are given by \eqref{SISOChannelDistributionPDF} and \eqref{SISOChannelDistributionCDF}, respectively.
Therefore, the ergodic sum-rate of the MIMO-OMA system considered can be expressed as:
\begin{align}\label{ErgodicSumRateMIMOOMAZF}
\hspace{-2mm}\overline{R_{\rm{sum,FDMA-ZF}}^{{\rm{MIMO-OMA}}}} &= {{\mathrm{E}}_{\mathbf{H}}}\left\{ {R_{\rm{sum,FDMA-ZF}}^{{\rm{MIMO-OMA}}}} \right\} = \int_0^\infty  M{{{\ln }}\left( {1 + \frac{P_{\rm{max}}}{MN_0}x} \right){f_{{{\left| {{h}} \right|}^2}}}\left( x \right)} dx \\
& = \frac{M}{\left(D\hspace{-1mm}+\hspace{-1mm}D_0\right)}\sum\limits_{n = 1}^N \hspace{-1mm}{{\beta _n}{e^{\frac{{{c_n M N_0}}}{{P_{\rm{max}}}}}} {\mathcal{E}_1}\left( {\frac{{{c_n M N_0}}}{{P_{\rm{max}}}}} \right)}.\notag
\end{align}

\subsection{Ergodic Sum-rate of MIMO-OMA with FDMA-MRC}
The instantaneous achievable data rate of user $k$ in the MIMO-OMA system using the FDMA-MRC receiver is given by
\begin{equation}\label{MIMOOMAMRCIndividualAchievableRate}
R_{k,\rm{FDMA-MRC}}^{{\rm{MIMO-OMA}}} =
f_k {\ln}\left(1+ {\frac{{{p_{k }}{{\left\| {{{{\mathbf{h}}}_k}} \right\|}^2}}}{{f_k N_0}}} \right).
\end{equation}
Upon adopting the equal resource allocation strategy, i.e., ${{p_{k}}} = \frac{{P_{\rm{max}}}}{K}$ and $f_k = 1/K$, the instantaneous sum-rate of MIMO-OMA relying on FDMA-MRC is obtained by
\begin{equation}\label{InstantSumRateMIMOOMAMRC}
R_{\rm{sum,FDMA-MRC}}^{{\rm{MIMO-OMA}}}
= \sum\limits_{k = 1}^K R_{k,\rm{FDMA-MRC}}^{{\rm{MIMO-OMA}}}
= \frac{1}{K}\sum\limits_{k = 1}^K {\ln}\left(1+ \frac{{P_{\rm{max}}}}{N_0}{\left\| {{{\mathbf{h}}_k}} \right\|}^2 \right).
\end{equation}
Averaging $R_{\rm{sum,FDMA-MRC}}^{{\rm{MIMO-OMA}}}$ over the channel fading, we arrive at the ergodic sum-rate of MIMO-OMA using FDMA-MRC as
\begin{align}\label{ErgodicSumRateMIMOOMAMRC}
\overline{R_{\rm{sum,FDMA-MRC}}^{{\rm{MIMO-OMA}}}} &= {\mathrm{E}_{\bf{H}}}\left\{ \frac{1}{K}\sum\limits_{k = 1}^K{\ln}\left(1+ \frac{{P_{\rm{max}}}}{N_0}{\left\| {{{\mathbf{h}}_k}} \right\|}^2 \right) \right\} \notag\\
& = \int_0^\infty  {{{\ln }}\left( {1 + \frac{{P_{\rm{max}}}}{{N_0}}x} \right){f_{{{\left\| {{\mathbf{h}}} \right\|}^2}}}\left( x \right)} dx \notag\\
& = \frac{1}{{\left( {D + {D_0}} \right)}}\sum\limits_{n = 1}^N {{\beta _n}\underbrace {\int_0^\infty  {\ln \left( {1 + \frac{{P_{\rm{max}}}}{{N_0}}x} \right)} {{\rm{Gamma}}}\left( {M,{c_n}} \right)dx}_{{T_n}}},
\end{align}
with $T_n$ given by
\begin{align}\label{T_N_App}
{T_n} & \mathop = \limits^{(a)} {\int_0^\infty  {\ln \left( 1 + t \right)} {{\rm{Gamma}}}\left( {M,\frac{{{N_0}{c_n}}}{{{P_{\rm{max}}}}}} \right)dt} \notag\\
& \mathop = \limits^{(b)}  \frac{\left( \frac{{{N_0}{c_n}}}{{{P_{\rm{max}}}}} \right)^M}{{\Gamma \left(M\right)}} G_{2,3}^{3,1}\left( {\begin{array}{*{20}{c}}
{ - M, - M + 1}\\
{ - M, - M, 0}
\end{array}\left| {\frac{{{N_0}{c_n}}}{{{P_{\rm{max}}}}}} \right.} \right),
\end{align}
where $G_{p,q}^{m,n}\left(  \cdot  \right)$ denotes the Meijer G-function.
The equality $(a)$ in \eqref{T_N_App} is obtained due to $t = \frac{{{P_{\rm{max}}}}}{{N_0}}x \sim {{\rm{Gamma}}}\left( {{M},\frac{{{N_0}{c_n}}}{{{P_{\rm{max}}}}}} \right)$ and the equality $(b)$ in \eqref{T_N_App} is based on Equation (3) in \cite{Heath2011}.
Now, the ergodic sum-rate of MIMO-OMA using FDMA-MRC can be written as
\begin{equation}\label{ErgodicSumRateMIMOOMAMRC2}
\overline{R_{\rm{sum,FDMA-MRC}}^{{\rm{MIMO-OMA}}}}
= \frac{1}{{\left( {D + {D_0}} \right)}}\sum\limits_{n = 1}^N {\beta _n}\left( \frac{\left( \frac{{N_0{c_n}}}{{{P_{\rm{max}}}}} \right)^M}{{\Gamma \left(M\right)}} G_{2,3}^{3,1}\left( {\begin{array}{*{20}{c}}
{ - M, - M + 1}\\
{ - M, - M, 0}
\end{array}\left| {\frac{{N_0{c_n}}}{{P_{\rm{max}}}}} \right.} \right)  \right).
\end{equation}

Note that, the ergodic sum-rate in \eqref{ErgodicSumRateMIMOOMAMRC2} is applicable to an arbitrary number of users $K$ and an arbitrary SNR, but it is too complicated to offer insights concerning the ESG of MIMO-NOMA over MIMO-OMA.
Hence, based on \eqref{ErgodicSumRateMIMOOMAMRC}, we derive the asymptotic ergodic sum-rate of MIMO-OMA with FDMA-MRC in the low-SNR regime with ${P_{\rm{max}}} \rightarrow 0$ as follows:
\begin{equation}\label{ErgodicSumRateMIMOOMAMRC3}
\mathop {\lim }\limits_{{P_{\rm{max}}} \rightarrow 0} \overline{R_{\rm{sum,FDMA-MRC}}^{{\rm{MIMO-OMA}}}} = \frac{{P_{\rm{max}}}}{{N_0}} \overline{{{\left\| {{\mathbf{h}}} \right\|}^2}}
= \frac{M{P_{\rm{max}}}}{{N_0} \left(D+D_0\right)}\sum\limits_{n = 1}^N {\frac{{{\beta _n}}}{{{c_n}}}}.
\end{equation}
On the other hand, in the high-SNR regime, based on \eqref{ErgodicSumRateMIMOOMAMRC}, the asymptotic ergodic sum-rate of MIMO-OMA using FDMA-MRC is given by
\begin{equation}\label{ErgodicSumRateMIMOOMAMRC4}
\mathop {\lim }\limits_{{P_{\rm{max}}} \rightarrow \infty}  \overline{R_{\rm{sum,FDMA-MRC}}^{{\rm{MIMO-OMA}}}}
= {\ln}\left(\frac{{P_{\rm{max}}}}{N_0}\right) + {\mathrm{E}_{\bf{h}}}\left\{ {\ln}\left({\left\| {{{\mathbf{h}}}} \right\|}^2 \right) \right\}.
\end{equation}
\subsection{ESG in Multi-antenna Systems}
By comparing \eqref{ErgodicSumRateMIMONOMA} and \eqref{ErgodicSumRateMIMOOMAZF}, we have the asymptotic ESG of MIMO-NOMA over MIMO-OMA relying on FDMA-ZF as follows:
\begin{align}\label{EPGMIMOERA}
\hspace{-3mm}\mathop {\lim }\limits_{K \rightarrow \infty}  \overline{G^{\rm{MIMO}}_{\rm{FDMA-ZF}}} &\hspace{-1mm}=
\mathop {\lim }\limits_{K \to \infty }  \overline{R_{\rm{sum}}^{{\rm{MIMO-NOMA}}}} - \overline{R_{\rm{sum,FDMA-ZF}}^{{\rm{MIMO-OMA}}}} \notag\\
&\hspace{-1mm}\approx M{\ln}\left( {1 \hspace{-1mm}+\hspace{-1mm} \frac{{P_{\rm{max}}}}{{{\left(D\hspace{-1mm}+\hspace{-1mm}D_0\right)N_0}}}\sum\limits_{n = 1}^N {\frac{{{\beta _n}}}{{{c_n}}}} } \right)\hspace{-1mm}- \hspace{-1mm} \frac{M}{\left(D\hspace{-1mm}+\hspace{-1mm}D_0\right)}\hspace{-1mm}\sum\limits_{n = 1}^N \hspace{-1mm}{{\beta _n}{e^{\frac{{{c_n}MN_0}}{{P_{\rm{max}}}}}} \hspace{-1mm}{\mathcal{E}_1}\hspace{-1mm}\left( \hspace{-0.5mm} {\frac{{{c_n}MN_0}}{{P_{\rm{max}}}}} \hspace{-0.5mm} \right)}.
\end{align}
To unveil some insights, we consider the asymptotic ESG in the high-SNR regime as follows
\begin{equation}\label{EPGMIMOERA2}
\hspace{-2mm}\mathop {\lim }\limits_{K \rightarrow \infty, {P_{\rm{max}}} \rightarrow \infty}  \overline{G^{\rm{MIMO}}_{\rm{FDMA-ZF}}} \approx M  \underbrace{\vartheta \left( {D,{D_0}} \right)}_{\rm{large-scale\;near-far\;gain}} + M\ln\left(M\right) + M \underbrace{\gamma}_{\rm{small-scale\;fading\;gain}},\hspace{-0.5mm}
\end{equation}
where $\vartheta \left( {D,{D_0}} \right)$ denotes the large-scale near-far gain given in \eqref{NearFarDiversity}.

\begin{Remark}
The identified two kinds of gains in ESG of the single-antenna system in \eqref{EPGSISOERA2} are also observed in the ESG of MIMO-NOMA over MIMO-OMA using FDMA-ZF in \eqref{EPGMIMOERA2}.
Moreover, it can be observed that both the large-scale near-far gain $\vartheta \left( {D,{D_0}} \right)$ and the small-scale fading gain $\gamma$ are increased by $M$ times as indicated in \eqref{EPGMIMOERA2}.
In fact, upon comparing \eqref{EPGSISOERA2} and \eqref{EPGMIMOERA2}, we have
\begin{equation}\label{EPGMIMOERA3}
\mathop {\lim }\limits_{K \rightarrow \infty, {P_{\rm{max}}} \rightarrow \infty}  \overline{G^{\rm{MIMO}}_{\rm{FDMA-ZF}}} = M\mathop {\lim }\limits_{K \rightarrow \infty, {P_{\rm{max}}} \rightarrow \infty} \overline{G^{{\rm{SISO}}}} + M\ln\left(M\right),
\end{equation}
which implies that the asymptotic ESG of MIMO-NOMA over MIMO-OMA is $M$-times of that in single-antenna systems, when there are $M$ UL receiver antennas at the BS.
In fact, for $K \rightarrow \infty$, the heterogeneity in channel directions of all the users allows the received signals to fully span across the $M$-dimensional signal space.
Hence, MIMO-NOMA and MIMO-OMA using FDMA-ZF can fully exploit the system's maximal spatial DoF $M$.
Furthermore, we have an additional power gain of $\ln\left(M\right)$ in the second term in \eqref{EPGMIMOERA3}.
This is due to a factor of $\frac{1}{M}$ average power loss within each group for ZF projection to suppress the IUI in the MIMO-OMA system considered\cite{Tse2005}.
\end{Remark}

Comparing \eqref{ErgodicSumRateMIMONOMA} and \eqref{ErgodicSumRateMIMOOMAMRC2}, the asymptotic ESG of MIMO-NOMA over MIMO-OMA with FDMA-MRC is obtained by:
\begin{align}\label{EPGMIMOERAMRC}
\hspace{-3mm}\mathop {\lim }\limits_{K \rightarrow \infty}  \overline{G^{\rm{MIMO}}_{\rm{FDMA-MRC}}} &\hspace{-1mm}=
\mathop {\lim }\limits_{K \to \infty }  \overline{R_{\rm{sum}}^{{\rm{MIMO-NOMA}}}} - \overline{R_{\rm{sum,FDMA-MRC}}^{{\rm{MIMO-OMA}}}} \notag\\
&\hspace{-1mm}\approx M{\ln}\left( {1 \hspace{-1mm}+\hspace{-1mm} \frac{{P_{\rm{max}}}}{{{\left(D\hspace{-1mm}+\hspace{-1mm}D_0\right)N_0}}}\sum\limits_{n = 1}^N {\frac{{{\beta _n}}}{{{c_n}}}} } \right) \notag\\ &-\frac{1}{{\left( {D + {D_0}} \right)}}\sum\limits_{n = 1}^N {\beta _n}\left( \frac{\left( \frac{{N_0{c_n}}}{{{P_{\rm{max}}}}} \right)^M}{{\Gamma \left(M\right)}} G_{2,3}^{3,1}\left( {\begin{array}{*{20}{c}}
{ - M, - M + 1}\\
{ - M, - M, 0}
\end{array}\left| {\frac{{N_0{c_n}}}{{P_{\rm{max}}}}} \right.} \right)  \right).
\end{align}
Then, based on \eqref{ErgodicSumRateMIMONOMA} and \eqref{ErgodicSumRateMIMOOMAMRC3}, the asymptotic ESG of MIMO-NOMA over MIMO-OMA with FDMA-MRC in the low-SNR regime is given by
\begin{align}\label{EPGMIMOERA5}
\mathop {\lim }\limits_{K \rightarrow \infty, {P_{\rm{max}}} \rightarrow 0}  \overline{R_{\rm{sum,FDMA-MRC}}^{{\rm{MIMO-OMA}}}} = 0.
\end{align}
Not surprisingly, the performance gain of MIMO-NOMA over MIMO-OMA with FDMA-MRC vanishes in the low-SNR regime, which has been shown by simulations in existing works, \cite{Ding2014} for example.
In the high-SNR regime, the asymptotic ESG of MIMO-NOMA over MIMO-OMA with FDMA-MRC can be obtained from \eqref{ErgodicSumRateMIMONOMA} and \eqref{ErgodicSumRateMIMOOMAMRC4} by
\begin{align}\label{EPGMIMOERA6}
\mathop {\lim }\limits_{K \rightarrow \infty, {P_{\rm{max}}} \rightarrow \infty}  \overline{G^{\rm{MIMO}}_{\rm{FDMA-MRC}}} &\approx \left( {M - 1} \right)\ln \left( {\frac{{{P_{{\rm{max}}}}}}{{\left( {D + {D_0}} \right){N_0}}}\sum\limits_{n = 1}^N {\frac{{{\beta _n}}}{{{c_n}}}} } \right) - \ln \left( M \right) + \Delta,
\end{align}
where $\Delta = \ln \left( {{{\rm{E}}_{\bf{h}}}\left\{ {{{\left\| {\bf{h}} \right\|}^2}} \right\}} \right) - {{\rm{E}}_{\bf{h}}}\left\{ {\ln \left( {{{\left\| {\bf{h}} \right\|}^2}} \right)} \right\}$ denotes the gap between $\ln \left( {{{\rm{E}}_{\bf{h}}}\left\{ {{{\left\| {\bf{h}} \right\|}^2}} \right\}} \right)$ and ${{\rm{E}}_{\bf{h}}}\left\{ {\ln \left( {{{\left\| {\bf{h}} \right\|}^2}} \right)} \right\}$.

Although the closed-form ESG of MIMO-NOMA over MIMO-OMA is not available for the case of FDMA-MRC, the third term $\Delta$ in \eqref{EPGMIMOERA6} is a constant for a given outer radius $D$ and inner radius $D_0$.
Besides, it is expected that the first term in \eqref{EPGMIMOERA6} dominates the ESG in the high-SNR regime.
%
%
We can observe that the first term in \eqref{EPGMIMOERA6} increases linearly with the system SNR in dB with a slope of $(M-1)$ in the high-SNR regime.
In other words, there is an $(M-1)$-fold DoF gain \cite{DMTadeoff} in the asymptotic ESG of MIMO-NOMA over MIMO-OMA using FDMA-MRC.
In fact, MIMO-NOMA is essentially an $M \times K$ MIMO system on all resource blocks, i.e., time slots and frequency subbands, where the system maximal spatial DoF is limited by $M$ due to $M<K$.
On the other hand, MIMO-OMA using the FDMA-MRC reception is always an $M \times 1$ MIMO system in each resource block, and thus it can only have a  spatial DoF, which is one.
As a result, an $(M-1)$-fold DoF gain can be achieved by MIMO-NOMA compared to MIMO-OMA using FDMA-MRC.
However, MIMO-OMA is only capable of offering a power gain of $\ln\left( M \right)$ owing to the MRC detection utilized at the BS and thus the asymptotic ESG in \eqref{EPGMIMOERA6} suffers from a power reduction by a factor of $\ln\left( M \right)$ in the second term.
\section{ESG of \emph{m}MIMO-NOMA over \emph{m}MIMO-OMA}
In this section, we first derive the ergodic sum-rate of both \emph{m}MIMO-NOMA and \emph{m}MIMO-OMA and then discuss the asymptotic ESG of \emph{m}MIMO-NOMA over \emph{m}MIMO-OMA.

\subsection{Ergodic Sum-rate with $D>D_0$}
Let us now apply NOMA to massive-MIMO systems, where a large-scale antenna array ($M \to \infty$) is employed at the BS and all the $K$ users are equipped with a single antenna.
A simple MRC-SIC receiver is adopted at the BS for data detection of \emph{m}MIMO-NOMA.
The instantaneous achievable data rate of user $k$ and the sum-rate of the \emph{m}MIMO-NOMA system using the MRC-SIC reception are given by
\begin{align}
R_{k}^{\rm{\emph{m}MIMO-NOMA}} &= {\ln}\left(1+ {\frac{{{p_k}{{\left\| {{{\bf{h}}_k}} \right\|}^2}}}{{\sum\limits_{i = k + 1}^K {p_i}{{\left\| {{{\bf{h}}_i}} \right\|}^2} {{\left| {\bf{e}}_k^{\rm{H}}{{\bf{e}}_i} \right|}^2}  + {N_0}}}} \right)\;\text{and} \label{mMIMONOMAIndividualAchievableRate}\\
R_{\rm{sum}}^{\rm{\emph{m}MIMO-NOMA}} &= \sum\limits_{k = 1}^K R_{k}^{\rm{\emph{m}MIMO-NOMA}},\label{mMIMONOMASumRate}
\end{align}
respectively, where ${{\bf{e}}_k} = \frac{{{{\bf{h}}_k}}}{{\left\| {{{\bf{h}}_k}} \right\|}}$ denotes the channel direction of user $k$.
For the massive-MIMO system associated with $D>D_0$, the asymptotic ergodic sum-rate of ${K \rightarrow \infty}$ and ${M \rightarrow \infty}$ is given in the following theorem.

\begin{Thm}\label{Theorem2}
For the \emph{m}MIMO-NOMA system considered in \eqref{MIMONOMASystemModel} in conjunction with $D>D_0$ and MRC-SIC detection at the BS, under the equal resource allocation strategy, i.e., ${{p_{k}}} = \frac{{P_{\rm{max}}}}{K}$, $\forall k$, the asymptotic ergodic sum-rate can be approximated by
\begin{align}\label{ErgodicSumRatemMIMONOMA}
&\mathop {\lim }\limits_{K \rightarrow \infty, M \rightarrow \infty} \overline{R_{\rm{sum}}^{\rm{\emph{m}MIMO-NOMA}}} = \mathop {\lim }\limits_{K \rightarrow \infty, M \rightarrow \infty} {{\mathrm{E}}_{\mathbf{H}}}\left\{ {R_{\rm{sum}}^{\rm{\emph{m}MIMO-NOMA}}} \right\} \notag\\
&\approx \mathop {\lim }\limits_{K \to \infty ,M \to \infty } \sum\limits_{k = 1}^K \left( {\begin{array}{*{20}{c}}
K\\
k
\end{array}} \right){\frac{{k}}{{D + {D_0}}}} \sum\limits_{n = 1}^N {{\beta _n}\ln \left( {1 + \frac{{{\psi _k}}}{{{c_n}}}} \right)} {\left( {\frac{{\phi _n^2 \hspace{-1mm}-\hspace{-1mm} D_0^2}}{{{D^2} \hspace{-1mm}-\hspace{-1mm} D_0^2}}} \right)^{k - 1}}{\left( {\frac{{{D^2} \hspace{-1mm}-\hspace{-1mm} \phi _n^2}}{{{D^2} \hspace{-1mm}-\hspace{-1mm} D_0^2}}} \right)^{K - k}},
\end{align}
with
\begin{align}\label{ParametersformMIMONOMA}
\phi_n &= {\frac{D-D_0}{2}\cos \frac{{2n - 1}}{{2N}}\pi  + \frac{D+D_0}{2}},
{\psi _k} = \frac{{{P_{\rm{max}}}M}}{{\sum\nolimits_{i = k + 1}^K {{P_{\rm{max}}}{I_i} + K{N_0}} }}, \text{and} \notag\\
{I_k} & = {{\rm{E}}_{{d_k}}}\left\{ {\frac{1}{{1 + d_k^\alpha }}} \right\} = \left( {\begin{array}{*{20}{c}}
K\\
k
\end{array}} \right){\frac{{k}}{{D + {D_0}}}} \sum\limits_{n = 1}^N \frac{{\beta _n}}{c_n} {\left( {\frac{{\phi _n^2 \hspace{-1mm}-\hspace{-1mm} D_0^2}}{{{D^2} \hspace{-1mm}-\hspace{-1mm} D_0^2}}} \right)^{k - 1}}{\left( {\frac{{{D^2} \hspace{-1mm}-\hspace{-1mm} \phi _n^2}}{{{D^2} \hspace{-1mm}-\hspace{-1mm} D_0^2}}} \right)^{K - k}}.
\end{align}
\end{Thm}
\begin{proof}
Please refer to Appendix B for the proof of Theorem \ref{Theorem2}.
\end{proof}

%

For the \emph{m}MIMO-OMA system using the FDMA-MRC detection, we can allocate more than one user to each frequency subband due to the above-mentioned favorable propagation property\cite{Ngo2013}.
In particular, upon allocating $W = \varsigma M$ users to each frequency subband with $\varsigma = \frac{W}{M} \ll 1$, the orthogonality among channel vectors of the $W$ users holds fairly well, hence the IUI becomes negligible.
Therefore, a random user grouping strategy is adopted, where we randomly select $W = \varsigma M$ users as a group and there are $G = \frac{K}{W}$ groups\footnote{Without loss of generality, we consider that $K$ is an integer multiple of $G$ and $W$.} separated using orthogonal frequency subbands.
%
%
In each subband, low-complexity MRC detection can be employed for each individual user and thus the instantaneous achievable data rate of user $k$ can be expressed by
\begin{equation}\label{mMIMOOMAIndividualAchievableRate}
R_{k}^{\rm{\emph{m}MIMO-OMA}} = f_g{\ln}\left(1+ {\frac{{{p_k}{{\left\| {{{\bf{h}}_k}} \right\|}^2}}}{{{f_gN_0}}}} \right),
\end{equation}
where $f_g$ denotes the normalized frequency allocation of the $g$-th group.
Note that \eqref{mMIMOOMAIndividualAchievableRate} serves as an upper bound of the instantaneous achievable data rate of user $k$ in the \emph{m}MIMO-OMA system, since we assumed it to be IUI-free.
Then, under the equal resource allocation strategy, i.e., ${{p_{k}}} = \frac{{P_{\rm{max}}}}{K}$ and $f_g = 1/G = \frac{W}{K} = \delta\varsigma$, we have the asymptotic ergodic sum-rate of the \emph{m}MIMO-OMA system associated with $D > D_0$ as follows:
\begin{align}\label{ErgodicSumRatemMIMOOMA}
&\mathop {\lim }\limits_{M \rightarrow \infty} \overline{R_{\rm{sum}}^{\rm{\emph{m}MIMO-OMA}}} = \mathop {\lim }\limits_{M \rightarrow \infty} {{\mathrm{E}}_{\mathbf{H}}}\left\{ {R_{\rm{sum}}^{\rm{\emph{m}MIMO-OMA}}} \right\} \notag\\
&= \mathop {\lim }\limits_{M \to \infty } \delta\varsigma\sum\limits_{k = 1}^K \left( {\begin{array}{*{20}{c}}
K\\
k
\end{array}} \right){\frac{{k}}{{D + {D_0}}}} \sum\limits_{n = 1}^N {{\beta _n}\ln \left( {1 + \frac{{{\xi}}}{{{c_n}}}} \right)} {\left( {\frac{{\phi _n^2 - D_0^2}}{{{D^2} - D_0^2}}} \right)^{k - 1}}{\left( {\frac{{{D^2} - \phi _n^2}}{{{D^2} - D_0^2}}} \right)^{K - k}},
\end{align}
where $\phi_n$ is given in \eqref{ParametersformMIMONOMA} and ${\xi} = \frac{{{P_{\rm{max}}}}}{\varsigma N_0}$.
%

\subsection{Ergodic Sum-rate with $D=D_0$}
We note that the analytical results in \eqref{ErgodicSumRatemMIMONOMA} and \eqref{ErgodicSumRatemMIMOOMA} are only applicable to the system having $D>D_0$.
The asymptotic ergodic sum-rate of the \emph{m}MIMO-NOMA system with $D=D_0$ can be expressed using the following theorem.

\begin{Thm}\label{Theorem3}
With $D = D_0$ and the equal resource allocation strategy, i.e., ${{p_{k}}} = \frac{{P_{\rm{max}}}}{K}$ and $f_g = 1/G = \frac{W}{K}= \delta\varsigma$, the asymptotic ergodic sum-rate of the \emph{m}MIMO-NOMA system and of the \emph{m}MIMO-OMA system can be formulated by
\begin{align}
\mathop {\lim }\limits_{K \rightarrow \infty, M \rightarrow \infty} \overline{R_{\rm{sum}}^{\rm{\emph{m}MIMO-NOMA}}}
&\approx \mathop {\lim }\limits_{K \to \infty ,M \to \infty } \frac{{M}}{\varpi\delta} \left[ \ln \left( 1 + \varpi\delta + \varpi \right) \left( 1 + \varpi\delta + \varpi \right) \right. \notag\\
&- \left. \ln \left( 1 + \varpi\delta \right)\left( 1 + \varpi\delta \right) - \ln \left( 1 + \varpi \right)\left( 1 + \varpi \right)\right] \;\;\text{and}\label{DD0ErgodicSumRatemMIMONOMA}\\
\mathop {\lim }\limits_{M \rightarrow \infty} \overline{R_{\rm{sum}}^{\rm{\emph{m}MIMO-OMA}}}
&= \mathop {\lim }\limits_{M \to \infty } {\varsigma M} \ln \left( {1 + \frac{\varpi}{\varsigma}} \right),\label{DD0ErgodicSumRatemMIMOOMA}
\end{align}
respectively, where $\delta = \frac{M}{K}$ and $\varsigma = \frac{W}{M}$ are constants and $\varpi = \frac{{P_{\rm{max}}}}{\left({1+D_0^{\alpha}}\right) N_0}$ denotes the total average received SNR of all the users.
\end{Thm}
\begin{proof}
Please refer to Appendix C for the proof of Theorem \ref{Theorem3}.
\end{proof}

\subsection{ESG in Massive-antenna Systems}
Based on \eqref{ErgodicSumRatemMIMONOMA} and \eqref{ErgodicSumRatemMIMOOMA}, when $D>D_0$, the asymptotic ESG of \emph{m}MIMO-NOMA over \emph{m}MIMO-OMA associated with ${K \to \infty}$ and ${M \to \infty}$ can be expressed as follows:
\begin{align}\label{DD0EPGmMIMOERA0}
\mathop {\lim }\limits_{K \to \infty, {M \to \infty}}  \overline{G^{\rm{\emph{m}MIMO}}_{D > D_0}} &\approx \mathop {\lim }\limits_{K \to \infty ,M \to \infty } \sum\limits_{k = 1}^K {\left( \hspace{-1mm} {\begin{array}{*{20}{c}}
K\\
k
\end{array}} \hspace{-1mm}\right)} \frac{k}{{D \hspace{-1mm}+\hspace{-1mm} {D_0}}}\sum\limits_{n = 1}^N {{\beta _n}\left[ {\ln \left( {1 \hspace{-1mm}+\hspace{-1mm} \frac{{{\psi _k}}}{{{c_n}}}} \right) - \delta\varsigma\ln \left( {1 \hspace{-1mm}+\hspace{-1mm} \frac{{{\xi}}}{{{c_n}}}} \right)} \right]}  \notag\\
&\times {\left( {\frac{{\phi _n^2 - D_0^2}}{{{D^2} - D_0^2}}} \right)^{k - 1}}{\left( {\frac{{{D^2} - \phi _n^2}}{{{D^2} - D_0^2}}} \right)^{K - k}}.
\end{align}
However, the expression in \eqref{DD0EPGmMIMOERA0} is too complicated and does not provide immediate insights.
Hence, we focus on the case of $D = D_0$ to unveil some important and plausible insights on the ESG of NOMA over OMA in the massive-MIMO system.
The simulation results of Section \ref{Simulations} will show that the insights obtained from the case of $D=D_0$ are also applicable to the general scenario of $D>D_0$.

Comparing \eqref{DD0ErgodicSumRatemMIMONOMA} and \eqref{DD0ErgodicSumRatemMIMOOMA}, when $D=D_0$, we have the asymptotic ESG of \emph{m}MIMO-NOMA over \emph{m}MIMO-OMA for ${K \to \infty}$ and ${M \to \infty}$ as follows:
\begin{align}\label{DD0EPGmMIMOERA}
\hspace{-3mm}\mathop {\lim }\limits_{K \to \infty, {M \to \infty}}  \overline{G^{\rm{\emph{m}MIMO}}_{D = D_0}} &= \mathop {\lim }\limits_{K \to \infty, {M \to \infty}}  \overline{R_{\rm{sum}}^{{\rm{\emph{m}MIMO-NOMA}}}} - \overline{R_{\rm{sum}}^{{\rm{\emph{m}MIMO-OMA}}}} \notag\\
& \approx \mathop {\lim }\limits_{K \to \infty ,M \to \infty } \frac{{M}}{\varpi\delta} \left[ \ln \left( 1 + \varpi\delta + \varpi \right) \left( 1 + \varpi\delta + \varpi \right) \right. \notag\\
& - \left.\ln \left( 1 + \varpi\delta \right)\left( 1 + \varpi\delta \right) - \ln \left( 1 + \varpi \right)\left( 1 + \varpi \right)\right] -  {\varsigma M} \ln \left( {1 + \frac{\varpi}{\varsigma}} \right).
\end{align}
In the low-SNR regime, we can observe that $\mathop {\lim }\limits_{K \to \infty, {M \to \infty}, {P_{\rm{max}}} \to 0}  \overline{G^{\rm{\emph{m}MIMO}}_{D = D_0}} \to 0$.
This implies that no gain can be achieved by NOMA in the low-SNR regime, which is consistent with \eqref{EPGMIMOERA5}.
By contrast, in the high-SNR regime, we have
\begin{align}
\mathop {\lim }\limits_{K \to \infty, {M \to \infty}, {P_{\rm{max}}} \to \infty}  \overline{G^{\rm{\emph{m}MIMO}}_{D = D_0}}
&\approx \mathop {\lim }\limits_{K \to \infty ,M \to \infty ,{P_{{\rm{max}}}} \to \infty } \frac{{M}}{\delta} \zeta  -  {\varsigma M} \ln \left( {1 + \frac{\varpi}{\varsigma}} \right), \label{DD0EPGmMIMOERA2}\\
&= \mathop {\lim }\limits_{K \to \infty ,M \to \infty ,{P_{{\rm{max}}}} \to \infty } K \zeta  -  {\delta\varsigma K} \ln \left( {1 + \frac{\varpi}{\varsigma}} \right)\label{DD0EPGmMIMOERA3}
\end{align}
where $\zeta = \left[ \ln \left( {1 \hspace{-0.5mm}+\hspace{-0.5mm} \varpi \delta  \hspace{-0.5mm}+\hspace{-0.5mm} \varpi } \right)\left( {1 \hspace{-0.5mm}+\hspace{-0.5mm} \delta } \right) \hspace{-0.5mm}-\hspace{-0.5mm} \ln \left( {1 \hspace{-0.5mm}+\hspace{-0.5mm} \varpi \delta } \right)\delta  \hspace{-0.5mm}-\hspace{-0.5mm} \ln \left( {1 \hspace{-0.5mm}+\hspace{-0.5mm} \varpi } \right) \right]$ represents the extra ergodic sum-rate gain upon supporting an extra user by the \emph{m}MIMO-NOMA system considered.
Explicitly, for $K \to \infty$, ${M \to \infty}$, and ${P_{\rm{max}}} \rightarrow \infty$, the resultant extra benefit $\zeta$ is jointly determined by the average received sum SNR $\varpi$ and the fixed ratio $\delta$.
Observe in \eqref{DD0EPGmMIMOERA2} and \eqref{DD0EPGmMIMOERA3} that given the average received sum SNR $\varpi$ and the fixed ratios $\delta$ and $\varsigma$, the asymptotic ESG scales linearly with both the number of UL receiver antennas at the BS, $M$ and the number of users, $K$, respectively.
In other words, the asymptotic ESG per user and the asymptotic ESG per antenna of \emph{m}MIMO-NOMA over \emph{m}MIMO-OMA are constant and they are given by
\begin{align}
\mathop {\lim }\limits_{K \to \infty, {M \to \infty}, {P_{\rm{max}}} \to \infty}  \frac{\overline{G^{\rm{\emph{m}MIMO}}_{D = D_0}}}{K} &= \zeta  -  {\delta\varsigma } \ln \left( {1 + \frac{\varpi}{\varsigma}} \right)\;\text{and} \label{DD0EPGPerUsermMIMOERA2}\\
\mathop {\lim }\limits_{K \to \infty, {M \to \infty}, {P_{\rm{max}}} \to \infty}  \frac{\overline{G^{\rm{\emph{m}MIMO}}_{D = D_0}}}{M} &= \frac{\zeta}{\delta} -  {\varsigma} \ln \left( {1 + \frac{\varpi}{\varsigma}} \right),\; \label{DD0EPGPerAntennamMIMOERA3}
\end{align}
respectively.
We can explain this observation from the spatial DoF perspective, since it determines the pre-log factor for the ergodic sum-rate of both \emph{m}MIMO-NOMA and \emph{m}MIMO-OMA and thus also determines the pre-log factor of the corresponding ESG.
In particular, the \emph{m}MIMO-NOMA system considered is basically an $(M \times K)$ MIMO system associated with $M<K$, since all the $K$ users transmit their signals simultaneously in the same frequency band.
When scaling up the \emph{m}MIMO-NOMA system while maintaining a fixed ratio $\delta = \frac{M}{K}$, the system's spatial DoF increases linearly  both with $M$ and $K$.
On the other hand, the spatial DoF of the \emph{m}MIMO-OMA system is limited by the group size $W$, since it is always an $(M \times W)$ MIMO system associated with $W \ll M$ in each time slot and frequency subband.
Therefore, the system's spatial DoF increases linearly with both $W$, and $M$, as well as $K$, when scaling up the \emph{m}MIMO-OMA system under fixed ratios of $\delta = \frac{M}{K}$ and $\varsigma = \frac{W}{M}$.
As a result, due to the linear increase of the spatial DoF with $M$ as well as $K$ for both the \emph{m}MIMO-NOMA and \emph{m}MIMO-OMA systems, the asymptotic ESG increases linearly with both $M$ and $K$.
Note that in contrast to \eqref{EPGMIMOERA6}, there is no DoF gain, despite the fact that the asymptotic ESG scales linearly both with $M$ as well as $K$.
This is because the extra benefit $\zeta$ does not increase linearly with the system's SNR in dB.
As a result, the asymptotic ESG of \emph{m}MIMO-NOMA over \emph{m}MIMO-OMA cannot increase linearly with the system's SNR in dB, as it will be shown in Section \ref{Simulations}.

\section{ESG in Multi-cell Systems}\label{DiscussionsMulticell}
\begin{figure}[t]
\centering
\includegraphics[width=3.5in]{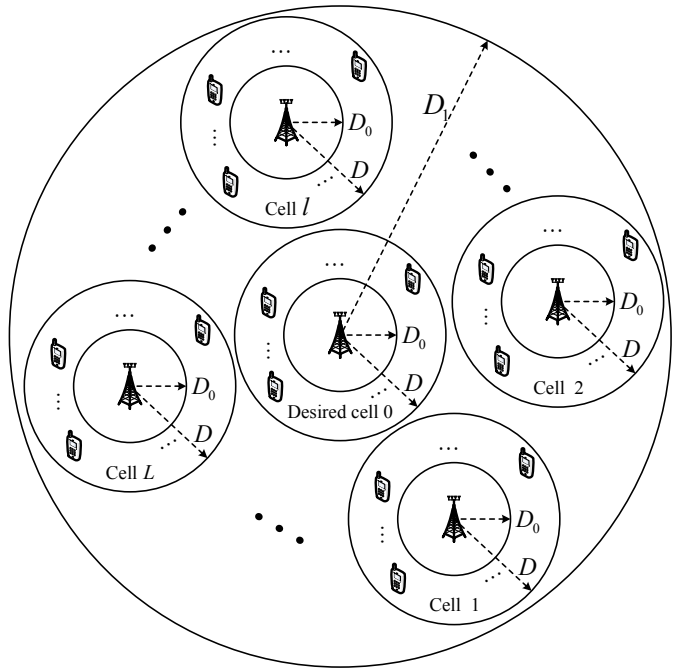}
\caption{The system model of the multi-cell uplink communication with one serving cell and $L$ adjacent cells.}
\label{Multi-cellNOMA_Uplink_Model}
\end{figure}

In Section III, the performance gain of NOMA over OMA has been investigated in single-cell systems, since these analytical results are easily comprehensible and reveal directly plausible insights.
Naturally, the performance gain of NOMA over OMA in single-cell systems serves as an upper bound on that of non-cooperative multi-cell systems, which can be approached by employing conservative frequency reuse strategy.
In practice, cellular networks consist of multiple cells where the inter-cell interference (ICI) is inevitable.
Furthermore, the characteristics of the ICI for NOMA and OMA schemes are different.
In particular, ICI is imposed by all the users in adjacent cells for NOMA schemes, while only a subset of users inflict ICI in OMA schemes, as an explicit benefit of orthogonal time or frequency allocation.
As a result, NOMA systems face more severe ICI than that of OMA, hence it remains unclear, if applying NOMA is still beneficial in multi-cell systems.
Therefore, in this section, we investigate the ESG of NOMA over OMA in multi-cell systems.
\subsection{Inter-cell Interference in NOMA and OMA Systems}
Consider a multi-cell system having multiple non-overlapped adjacent cells with index $l = 1,\ldots,L$, which are randomly deployed and  surround the serving cell $l=0$, as shown in Fig. \ref{Multi-cellNOMA_Uplink_Model}.
We assume that the $L$ interfering cells have the same structure as the serving cell and they are uniformly distributed in the pair of concentric ring-shaped discs of Fig. \ref{Multi-cellNOMA_Uplink_Model} having the inner radius of $D$ and outer radius of $D_1$.
Furthermore, we adopt the radical frequency reuse factor of 1, i.e., using the same frequency band for all cells to facilitate the performance analysis\footnote{With a less-aggressive frequency reuse strategy in multi-cell systems, both NOMA and OMA schemes endure less ICI since only the adjacent cells using the same frequency band with the serving cell are taken into account. As a result, the performance analyses derived in this paper can be extended to the case with a lower frequency reuse ratio by simply decreasing number of adjacent cells $L$. Again, the resultant performance will then approach the performance upper-bound of the single-cell scenario.}.
Again, we are assuming that in each cell there is a single $M$-antenna BS serving $K$ single-antenna users in the UL and thus there are $KL$ users imposing interference on the serving BS.
Additionally, to reduce both the system's overhead and its complexity, no cooperative multi-cell processing is included in our multi-cell system considered.
In the following, we first investigate the resultant ICI distribution and then derive the total received ICI power contaminating over NOMA and OMA systems.

Given the normalized UL receive beamforming vector of user $k$ in the serving cell at the serving BS represented by $\mathbf{w}_k \in \mathbb{C}^{M \times 1}$ with $\left\|\mathbf{w}_k\right\|^2 = 1$, the effective ICI channel spanning from user $k'$ in adjacent cell $l$ to the serving BS can be formulated as:
\begin{equation}\label{MulticelleEffectiveChannelModel}
{h_{k',l}} = {\bf{w}}_k^{\rm{H}}{\mathbf{h}_{k',l}} = \frac{{\bf{w}}_k^{\rm{H}}{\bf{g}}_{k',l}}{\sqrt{1+d_{k',l}^{\alpha}}},
\end{equation}
where ${\mathbf{h}_{k',l}} = \frac{{\bf{g}}_{k',l}}{\sqrt{1+d_{k',l}^{\alpha}}}$ denotes the channel vector from user $k'$ in adjacent cell $l$ to the serving BS, ${\bf{g}}_{k',l} \in \mathbb{C}^{ M \times 1}$ represents the Rayleigh fading coefficients, i.e., ${\bf{g}}_{k',l} \sim \mathcal{CN}\left(\mathbf{0},{{\bf{I}}_M}\right)$, and $d_{k',l}$ denotes the distance between user $k'$ in adjacent cell $l$ and the serving BS with the unit of meter.
Similar to the single-cell system considered, we assume that the CSI of all the users within the serving cell is perfectly known at the serving BS.
However, the ICI channel is unknown for the serving BS.
Note that the receive beamformer ${\bf{w}}_k$ of the serving BS depends on the instantaneous channel vector of user $k$, ${\mathbf{h}_{k}}$, and/or on the multiple access interference structure $\left[ {{{\bf{h}}_1}, \ldots ,{{\bf{h}}_{k - 1}},{{\bf{h}}_{k + 1}}, \ldots ,{{\bf{h}}_K}} \right]$ in the serving cell.
Therefore, the receive beamformer ${\bf{w}}_k$ of the serving cell is independent of the ICI channel ${\bf{g}}_{k',l}$.
As a result, owing to $\left\|\mathbf{w}_k\right\|^2 = 1$, it can be readily observed that ${{\bf{w}}_k^{\rm{H}}{\bf{g}}_{k',l}}$ obeys the circularly symmetric complex Gaussian distribution conditioned on the given $\mathbf{w}_k$, i.e., we have ${\bf{w}}_k^{\rm{H}}{{\bf{g}}_{k',l}}\left| {_{{{\bf{w}}_k}}} \right. \sim \mathcal{CN}\left(0,1\right)$.
However, since the resultant distribution $\mathcal{CN}\left(0,1\right)$ is independent of ${\bf{w}}_k$, we can safely drop the condition and directly apply ${\bf{w}}_k^{\rm{H}}{{\bf{g}}_{k',l}} \sim \mathcal{CN}\left(0,1\right)$.
Now, based on \eqref{MulticelleEffectiveChannelModel}, we can observe that the effective ICI channel ${h_{k',l}}$ is equivalent to a single-antenna Rayleigh fading channel associated with a distance of $d_{k',l}$, regardless of how many antennas are employed at the serving BS.

Given that each user is equipped with a single-antenna, the transmission of each user is omnidirectional.
Therefore, to facilitate the analysis of the ICI power, we assume that there is no gap between the adjacent cells and that the inner radius of each adjacent cell is zero, i.e., $D_0 = 0$.
Hence, we can further assume that the ICI emanates from $KL$ users uniformly distributed within the ring-shaped disc having the inner radius of $D$ and outer radius of $D_1$.
Similar to \eqref{SISOChannelDistributionPDF} and \eqref{SISOChannelDistributionCDF}, the CDF and PDF of ${{\left| {h_{k',l}} \right|}^2}$ are given by
\begin{align}
{F_{{{\left| {h_{k',l}} \right|}^2}}}\left( x \right) &\approx 1 - \frac{1}{D+D_1}\sum\limits_{n = 1}^N {{\beta' _n}{e^{ - {c'_n}x}}} \; \text{and} \label{MulticellEffectiveChannelDistributionPDF}\\
{f_{{{\left| {h_{k',l}} \right|}^2}}}\left( x \right) &\approx \frac{1}{D+D_1}\sum\limits_{n = 1}^N {{\beta' _n}{c'_n}{e^{ - {c'_n}x}}}, x \ge 0, \forall k',l \label{MulticellEffectiveChannelDistributionCDF}
\end{align}
respectively, with parameters of
\begin{align}\label{MulticellEffectiveChannelBetaCn}
{\beta'_n} &= \frac{\pi }{N}\left| {\sin \frac{{2n \hspace{-1mm}-\hspace{-1mm} 1}}{{2N}}\pi } \right|\left( {\frac{D_1\hspace{-1mm}-\hspace{-1mm}D}{2}\cos \frac{{2n \hspace{-1mm}-\hspace{-1mm} 1}}{{2N}}\pi  + \frac{D_1\hspace{-1mm}+\hspace{-1mm}D}{2}} \right) \;\text{and}\notag\\
{c'_n} &= 1 + {\left( {\frac{D_1\hspace{-1mm}-\hspace{-1mm}D}{2}\cos \frac{{2n \hspace{-1mm}-\hspace{-1mm} 1}}{{2N}}\pi  + \frac{D_1\hspace{-1mm}+\hspace{-1mm}D}{2}} \right)^\alpha}.
\end{align}
Note that all the adjacent cell users have i.i.d. channel distributions since we ignore the adjacent cells' structure.

Due to the ICI encountered in unity-frequency-reuse multi-cell systems, the performance is determined by the signal-to-interference-plus-noise ratio (SINR) instead of the SNR of single-cell systems.
Assuming that the ICI is treated as AWGN by the detector, the system's SINR can be defined as follows:
\begin{equation}\label{MulticellSINR}
{\rm{SINR}_{sum}^{multicell}} = \frac{{P_{\rm{max}}}}{{I_{{\rm{inter}}}} + N_0} {\overline{{{\left| {{{h}}} \right|}^2}}},
\end{equation}
where ${I_{{\rm{inter}}}}$ characterizes the ICI power in multi-cell systems and ${P_{\rm{max}}}$ denotes the same system power budget in each single cell.

To facilitate our performance analysis, we assume that the equal resource allocation strategy is adopted in all the adjacent cells, i.e., $p_{k',l} = \frac{{P_{\rm{max}}}}{K}$, $\forall k',l$.
When invoking NOMA in a multi-cell system, the ICI power can be modeled as
\begin{equation}\label{MulticellInterference}
{I^{\rm{NOMA}}_{{\rm{inter}}}} = \sum\limits_{l = 1}^L {\sum\limits_{k' = 1}^K {\frac{{{P_{{\rm{max}}}}}}{K}} } {\left| {{h_{k',l}}} \right|^2}.
\end{equation}
For $KL \to \infty$, ${I^{\rm{NOMA}}_{{\rm{inter}}}}$ becomes a deterministic value, which can be approximated by
\begin{equation}\label{MulticellInterference2}
\mathop {\lim }\limits_{KL \to \infty} {I^{\rm{NOMA}}_{{\rm{inter}}}} \approx L {P_{{\rm{max}}}} \overline{{\left| {{h_{k',l}}} \right|^2}} \approx \frac{L {P_{{\rm{max}}}} }{D+D_1}\sum\limits_{n = 1}^N {\frac{{{\beta' _n}}}{{{c'_n}}}}.
\end{equation}
As a result, the SINR of the multi-cell NOMA system considered is given by
\begin{equation}\label{MulticellSINRNOMA}
{\rm{SINR}_{sum,NOMA}^{multicell}} = \frac{{P_{\rm{max}}}}{\frac{L {P_{{\rm{max}}}} }{D+D_1}\sum\limits_{n = 1}^N {\frac{{{\beta' _n}}}{{{c'_n}}}} + N_0} {\overline{{{\left| {{{h}}} \right|}^2}}}.
\end{equation}

For OMA schemes, we assume that all the $K$ users in each cell are clustered into $G$ groups, with each group allocated to a frequency subband exclusively.
Since only $\frac{1}{G}$ of users in each adjacent cell are simultaneously transmitting their signals in each frequency subband, the ICI power in a multi-cell OMA system can be expressed as:
\begin{equation}\label{MulticellInterference3}
\mathop {\lim }\limits_{KL \to \infty} {I^{\rm{OMA}}_{{\rm{inter}}}} = \frac{1}{G} \mathop {\lim }\limits_{KL \to \infty} {I^{\rm{NOMA}}_{{\rm{inter}}}} \approx
\frac{L {P_{{\rm{max}}}}}{G(D+D_1)}\sum\limits_{n = 1}^N {\frac{{{\beta' _n}}}{{{c'_n}}}}.
\end{equation}
The SINR of the multi-cell OMA system considered can be written as:
\begin{equation}\label{MulticellSINROMA}
{\rm{SINR}_{sum,OMA}^{multicell}} = \frac{{P_{\rm{max}}}}{\frac{L {P_{{\rm{max}}}} }{G(D+D_1)}\sum\limits_{n = 1}^N {\frac{{{\beta' _n}}}{{{c'_n}}}} + \frac{1}{G} N_0} {\overline{{{\left| {{{h}}} \right|}^2}}}.
\end{equation}
Note that we have $G = K$ for SISO-OMA and MIMO-OMA with FDMA-MRC, $G = \frac{K}{M}$ for MIMO-OMA with FDMA-ZF, and $G = \frac{K}{W}$ for \emph{m}MIMO-OMA with FDMA-MRC.

\subsection{ESG in Multi-cell Systems}
It can be observed that ${I^{\rm{NOMA}}_{{\rm{inter}}}}$ in \eqref{MulticellInterference2} and ${I^{\rm{OMA}}_{{\rm{inter}}}}$ in \eqref{MulticellInterference3} are independent of the number of antennas employed at the serving BS, which is due to the non-coherent combining used at the serving BS ${\bf{w}}_k^{\rm{H}}{\bf{g}}_{k',l}$, thereby leading to the effective ICI channel becoming equivalent to a single-antenna Rayleigh fading channel.
Therefore, all the ergodic sum-rates of NOMA in single-antenna, multi-antenna, and massive-MIMO single-cell systems are degraded upon replacing the noise power $N_0$ by $({{I^{\rm{NOMA}}_{{\rm{inter}}}} + N_0})$.
On the other hand, since OMA schemes only face a noise power level of $\frac{1}{G}N_0$ on each subband, all the ergodic sum-rates of the OMA schemes in single-antenna, multi-antenna, and massive-MIMO single-cell systems are reduced upon substituting the noise power $\frac{1}{G}N_0$ by ${{I^{\rm{OMA}}_{{\rm{inter}}}} + \frac{1}{G}N_0} = \frac{1}{G} \left( {{I^{\rm{NOMA}}_{{\rm{inter}}}} \hspace{-1mm}+\hspace{-1mm} {N_0}} \right)$.

Given the ICI terms ${I^{\rm{NOMA}}_{{\rm{inter}}}}$ and ${I^{\rm{OMA}}_{{\rm{inter}}}}$, we have the corresponding asymptotic ESGs in single-antenna, multi-antenna, and massive-MIMO multi-cell systems as follows:
\begin{align}
\mathop {\lim }\limits_{K \rightarrow \infty} \overline{G^{{\rm{SISO}}}}'
&\approx \ln \left( {1 + \frac{{{P_{{\rm{max}}}}}}{{\left( {D + {D_0}} \right)\left( {{I^{\rm{NOMA}}_{{\rm{inter}}}} \hspace{-1mm}+\hspace{-1mm} {N_0}} \right)}}\sum\limits_{n = 1}^N {\frac{{{\beta _n}}}{{{c_n}}}} } \right) \notag\\
&-\frac{1}{{\left( {D + {D_0}} \right)}}\sum\limits_{n = 1}^N {{\beta _n}{e^{\frac{{{c_n}\left( {{I^{\rm{NOMA}}_{{\rm{inter}}}} + {N_0}} \right)}}{{{P_{{\rm{max}}}}}}}}{{\cal E}_1}\left( {\frac{{{c_n}\left( {{I^{\rm{NOMA}}_{{\rm{inter}}}} +{N_0}} \right)}}{{{P_{{\rm{max}}}}}}} \right)},\label{MultiCellEPGSISOERA2} \\
\mathop {\lim }\limits_{K \rightarrow \infty}  \overline{G^{{\rm{MIMO}}}_{\rm{FDMA-ZF}}}'
&\approx M\ln \left( {1 + \frac{{{P_{{\rm{max}}}}}}{{\left( {D + {D_0}} \right)\left( {{I^{\rm{NOMA}}_{{\rm{inter}}}} \hspace{-1mm}+\hspace{-1mm} {N_0}} \right)}}\sum\limits_{n = 1}^N {\frac{{{\beta _n}}}{{{c_n}}}} } \right) \notag\\
&-\frac{M}{{\left( {D + {D_0}} \right)}}\sum\limits_{n = 1}^N {{\beta _n}{e^{\frac{{{c_n}M\left( {{I^{\rm{NOMA}}_{{\rm{inter}}}} + {N_0}} \right)}}{{{P_{{\rm{max}}}}}}}}{{\cal E}_1}\left( {\frac{{{c_n}M\left( {{I^{\rm{NOMA}}_{{\rm{inter}}}} \hspace{-1mm}+\hspace{-1mm} {N_0}} \right)}}{{{P_{{\rm{max}}}}}}} \right)},\label{MultiCellEPGMIMOERA2}\\
\mathop {\lim }\limits_{K \rightarrow \infty}  \overline{G^{\rm{MIMO}}_{\rm{FDMA-MRC}}}' &\approx M\ln \left( {1 + \frac{{{P_{{\rm{max}}}}}}{{\left( {D + {D_0}} \right)\left( {{I^{\rm{NOMA}}_{{\rm{inter}}}} \hspace{-1mm}+\hspace{-1mm} {N_0}} \right)}}\sum\limits_{n = 1}^N {\frac{{{\beta _n}}}{{{c_n}}}} } \right) \notag\\
&\hspace{-35mm}-\frac{1}{{\left( {D + {D_0}} \right)}}\sum\limits_{n = 1}^N {{\beta _n}} \left( {\frac{{{{\left( {\frac{{\left( {{I^{\rm{NOMA}}_{{\rm{inter}}}} + {N_0}} \right){c_n}}}{{{P_{{\rm{max}}}}}}} \right)}^M}}}{{\Gamma \left( M \right)}}G_{2,3}^{3,1}\left(\hspace{-2mm} {\begin{array}{*{20}{c}}
{ - M, - M + 1}\\
{ - M, - M,0}
\end{array}\hspace{-2mm}\left| {\frac{{\left( {{I^{\rm{NOMA}}_{{\rm{inter}}}} \hspace{-1mm}+\hspace{-1mm} {N_0}} \right){c_n}}}{{{P_{{\rm{max}}}}}}} \right.} \right)} \right),\label{MultiCellEPGMIMOERAMRC}
\end{align}
\begin{align}
\mathop {\lim }\limits_{K \to \infty, {M \to \infty}}  \overline{G^{\rm{\emph{m}MIMO}}_{D > D_0}}' &\approx \mathop {\lim }\limits_{K \to \infty ,M \to \infty } \sum\limits_{k = 1}^K {\left( \hspace{-1mm} {\begin{array}{*{20}{c}}
K\\
k
\end{array}} \hspace{-1mm}\right)} \frac{k}{{D \hspace{-1mm}+\hspace{-1mm} {D_0}}}\sum\limits_{n = 1}^N {\beta _n} \notag\\
&\hspace{-25mm}\times\left[ {\ln \left( {1 \hspace{-1mm}+\hspace{-1mm} \frac{{{\psi _k}'}}{{{c_n}}}} \right) - \delta\varsigma\ln \left( {1 \hspace{-1mm}+\hspace{-1mm} \frac{{{\xi}'}}{{{c_n}}}} \right)} \right] {\left( {\frac{{\phi _n^2 - D_0^2}}{{{D^2} - D_0^2}}} \right)^{k - 1}}{\left( {\frac{{{D^2} - \phi _n^2}}{{{D^2} - D_0^2}}} \right)^{K - k}}, \;\text{and} \label{MultiCellDD0EPGmMIMOERA0} \\
\mathop {\lim }\limits_{K \to \infty, {M \to \infty}}  \overline{G^{\rm{\emph{m}MIMO}}_{D = D_0}}'
&\approx \frac{{M}}{\varpi'\delta} \left[ \ln \left( 1 + \varpi'\delta + \varpi' \right) \left( 1 + \varpi'\delta + \varpi' \right) \right. \notag\\
&\hspace{-20mm} -\left.\ln \left( 1 + \varpi'\delta \right)\left( 1 + \varpi'\delta \right) - \ln \left( 1 + \varpi' \right)\left( 1 + \varpi' \right)\right] - {{\varsigma M}} \ln \left( {1 + \frac{\varpi'}{\varsigma}} \right), \label{MultiCellDD0EPGmMIMOERA}
\end{align}
where ${\psi _k}^\prime  = \frac{{{P_{{\rm{max}}}}M}}{{\sum\nolimits_{i = k + 1}^K {{P_{{\rm{max}}}}{I_i} + K\left( {{I^{\rm{NOMA}}_{{\rm{inter}}}} + {N_0}} \right)} }}$, ${{\xi}'} = \frac{{{P_{{\rm{max}}}}M}}{{W\left( {{I^{\rm{NOMA}}_{{\rm{inter}}}} + {N_0}} \right)}}$, and $\varpi ' = \frac{{{P_{{\rm{max}}}}}}{{\left( {1 + D_0^\alpha } \right)\left( {{I^{\rm{NOMA}}_{{\rm{inter}}}} + {N_0}} \right)}}$.

It can be observed that compared to single-cell systems, the ESGs of NOMA over OMA in multi-cell systems are degraded due to the existence of ICI.
In particular, since the OMA schemes endure not only $\frac{1}{G}$ of noise power but also $\frac{1}{G}$ of ICI power, compared to NOMA schemes, the interference plus noise power of $({{I^{\rm{NOMA}}_{{\rm{inter}}}} + {N_0}})$ in multi-cell systems plays the same role as the noise power ${N_0}$ in single-cell systems.
Therefore, the performance analyses in single-cell systems are directly applicable to multi-cell systems via increasing the noise power ${N_0}$ to the interference plus noise power of $({{I^{\rm{NOMA}}_{{\rm{inter}}}} + {N_0}})$.
Upon utilizing the coordinate signal processing among multiple cells\cite{ShinMulticellNOMA}, the ICI power can be effectively suppressed, which may prevent the ESG degradation, when extending NOMA from single-cell to multi-cell systems.

\section{Simulations}\label{Simulations}
In this section, we use simulations to evaluate our analytical results.
In the single-cell systems considered, the inner cell radius is $D_0 = 50$ m and the outer cell radius is given by $D = [50, 200, 500]$ m, which corresponds to the cases of normalized cell sizes given by $\eta = [1, 4, 10]$, respectively.
The number of users $K$ ranges from $2$ to $256$ and the number of antennas employed at the BS $M$ ranges from $1$ to $128$.
The path loss exponent is $\alpha = 3.76$ according to the 3GPP path loss model\cite{Access2010}.
The noise power is set as $N_0 = -80$ dBm.
To emphasize the effect of cell size on the ESG of NOMA over OMA, in the simulations of the single-cell systems, we characterize the system's SNR with the aid of the total average received SNR of all the users at the BS as follows\cite{Xu2017}:
\begin{equation}\label{SystemSNR}
{\rm{SNR}_{sum}} = \frac{{P_{\rm{max}}}}{N_0} {\overline{{{\left| {{{h}}} \right|}^2}}} = \frac{{P_{\rm{max}}}}{N_0} \frac{\overline{{{\left\| {{\mathbf{h}}} \right\|}^2}}}{M},
\end{equation}
where ${\overline{{{\left| {{{h}}} \right|}^2}}}$ and ${\overline{{{\left\| {{\mathbf{h}}} \right\|}^2}}}$ are given by \eqref{Mean_ChannelPowerGainSISO} and \eqref{Mean_ChannelPowerGainMIMO}, respectively.
The total transmit power ${P_{\rm{max}}}$ is adjusted adaptively for different cell sizes to satisfy ${\rm{SNR}_{sum}}$ in \eqref{SystemSNR} ranging from $0$ dB to $40$ dB.
In the \emph{m}MIMO-OMA system considered, we set the ratio between the group size and the number of antennas to $\varsigma = \frac{W}{M} = \frac{1}{16}$, hence we can assume that the favorable propagation conditions prevail in the spirit of \cite{Ngo2013}.
Additionally, in the \emph{m}MIMO-NOMA system considered, the ratio between the number of receiver antennas at the BS and the number of serving users is fixed as $\delta = \frac{M}{K} =  \frac{1}{2}$.
The important system parameters adopted in our simulations are summarized in Table \ref{SysParameters}.
The specific simulation setups for each simulation scenario are shown under each figure.
All the simulation results in this paper are obtained by averaging the system performance over both small-scale fading and large-scale fading.

\begin{table}
	\caption{System Parameters Used In Simulations}
	\centering
	\begin{tabular}{l|r}
		\hline
		Inner cell radius, $D_0$                 & 50 m \\\hline
		Outer cell radius, $D$                   & [50, 200, 500] m \\\hline
		Normalized cell size, $\eta$             & [1, 4, 10] \\\hline
		Number of users, $K$                     & 2 $\sim$ 256 \\\hline
		Number of receive antennas at BS, $M$    & 1 $\sim$ 128 \\\hline
		Path loss exponent, $\alpha$         & 3.76 \\\hline
		Noise power, $N_0$                   & -80 dBm \\\hline
		System SNR, ${\rm{SNR}_{sum}}$       & 0 $\sim$ 40 dB \\\hline
		Ratio $\varsigma = \frac{W}{M}$ for \emph{m}MIMO-OMA & $\frac{1}{16}$ \\\hline
		Ratio $\delta = \frac{M}{K}$ for \emph{m}MIMO-NOMA & $\frac{1}{2}$ \\\hline
	\end{tabular}\label{SysParameters}
\end{table}


\subsection{ESG versus the Number of Users in Single-cell Systems}
\begin{figure}[t]
\centering
\vspace{-4mm}
\subfigure[Single-antenna systems]
{\label{APGVsK:a} 
\includegraphics[width=2.22in]{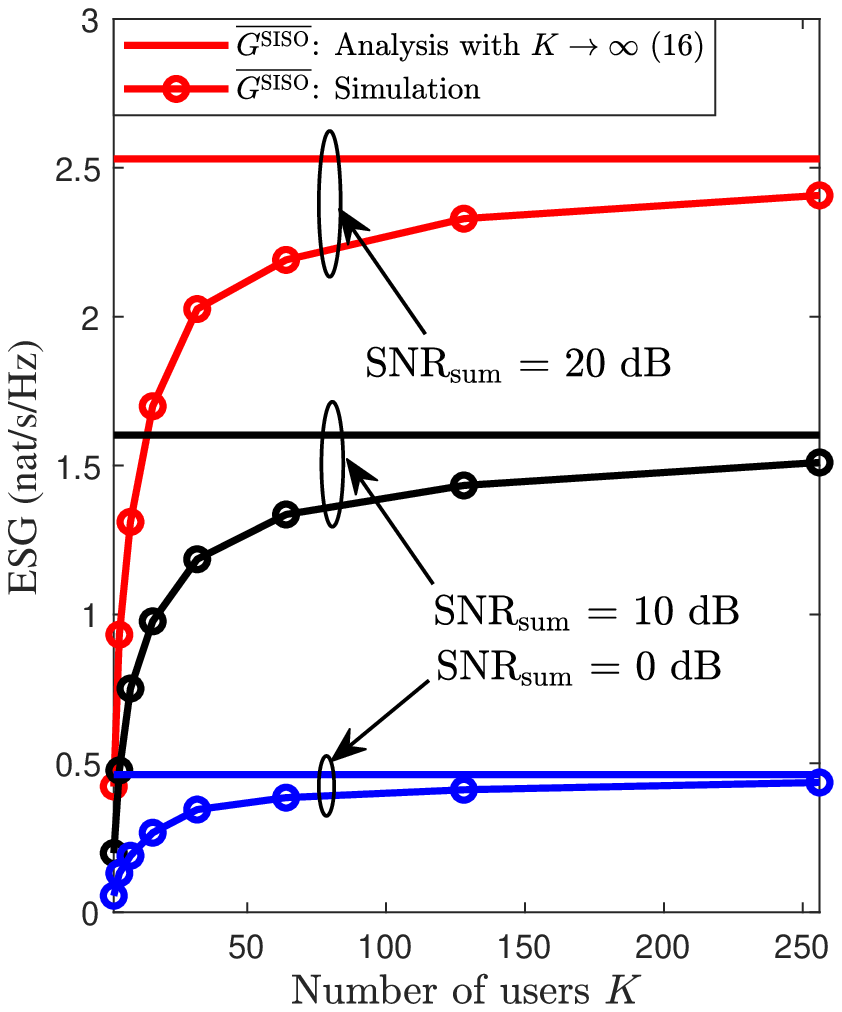}}
\hspace{-7mm}
\subfigure[Multi-antenna systems]
{\label{APGVsK:b} 
\includegraphics[width=2.22in]{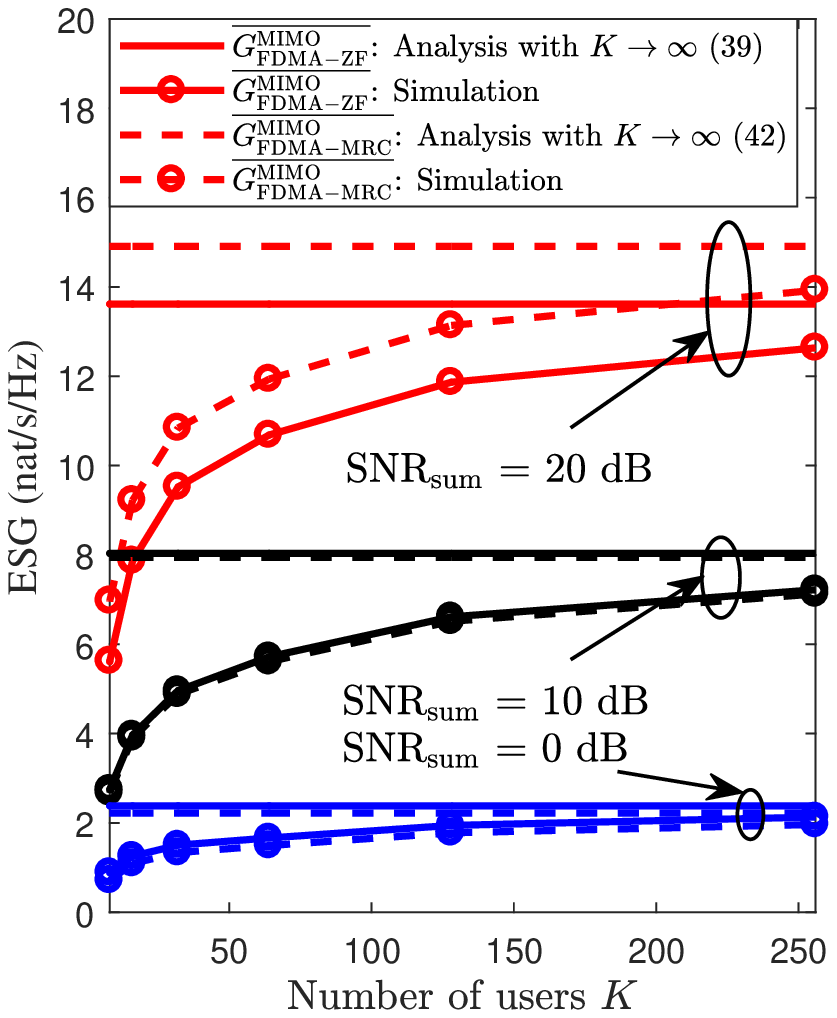}}
\hspace{-7mm}
\subfigure[Massive-antenna systems]
{\label{APGVsK:c} 
\includegraphics[width=2.22in]{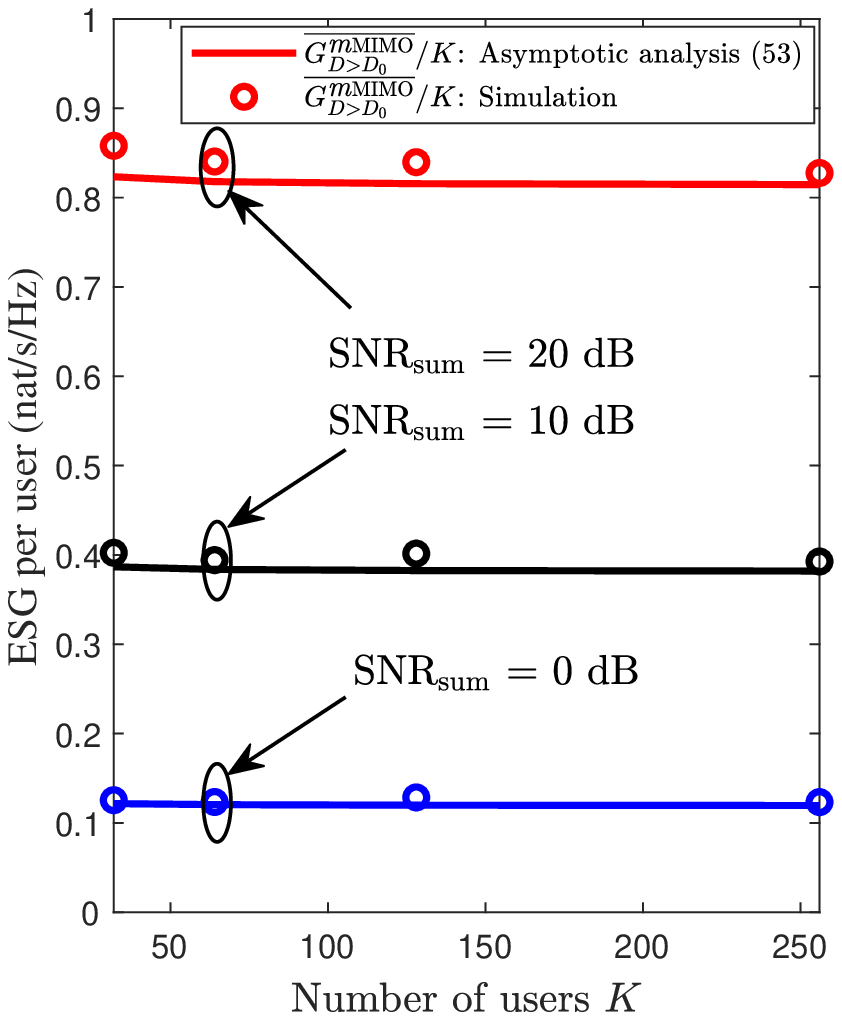}}
\caption{The ESG of NOMA over OMA versus the number of users $K$. The normalized cell size is $\eta = 10$ and the average received sum SNR is ${\rm{SNR}_{sum}} = [0,10,20]$ dB. For the considered MIMO-NOMA and MIMO-OMA systems in Fig. \ref{APGVsK:b}, we have $M=4$. For the considered \emph{m}MIMO-NOMA and \emph{m}MIMO-OMA systems in Fig. \ref{APGVsK:c}, the number of antennas equipped at the BS is adjusted according to the number of users $K$ based on $M = {K}{\delta}$ with $\delta = \frac{1}{2}$.}
\label{APGVsK}%
\end{figure}


Fig. \ref{APGVsK} illustrates the ESG of NOMA over OMA versus the number of users in the single-antenna, multi-antenna, and massive-MIMO single-cell systems.
%
In both Fig. \ref{APGVsK:a} and Fig. \ref{APGVsK:b}, we can observe that the ESG increases with the number of users $K$ and eventually approaches the asymptotic results derived for $K\to \infty$.
This is because upon increasing the number of users, the heterogeneity in channel gains among users is enhanced, which leads to an increased near-far gain.
As shown in Fig. \ref{APGVsK:c}, for massive-MIMO systems, the asymptotic ESG per user derived in \eqref{DD0EPGmMIMOERA0} closely matches with the simulations even for moderate numbers of users and SNRs.
Although \eqref{DD0EPGPerUsermMIMOERA2} is derived for massive-MIMO systems with $D=D_0$, in Fig. \ref{APGVsK:c}, we can observe a constant ESG per user in massive-MIMO systems with $D>D_0$.
In other words, the insights obtained from the massive-MIMO systems with $D=D_0$ are also applicable to the scenarios of $D>D_0$.
Compared to the ESG in the single-antenna systems of Fig. \ref{APGVsK:a}, the ESG in the multi-antenna systems of Fig. \ref{APGVsK:b} is substantially increased due to the extra spatial DoF offered by additional antennas at the BS.
Moreover, it can be observed in Fig. \ref{APGVsK:b} that we have $\mathop {\lim }\limits_{K \rightarrow \infty}  \overline{G^{\rm{MIMO}}_{\rm{FDMA-ZF}}} > \mathop {\lim }\limits_{K \rightarrow \infty}  \overline{G^{\rm{MIMO}}_{\rm{FDMA-MRC}}}$ in the low-SNR case, while $\mathop {\lim }\limits_{K \rightarrow \infty}  \overline{G^{\rm{MIMO}}_{\rm{FDMA-ZF}}} < \mathop {\lim }\limits_{K \rightarrow \infty}  \overline{G^{\rm{MIMO}}_{\rm{FDMA-MRC}}}$ in the high-SNR case.
This is because ZF detection outperforms MRC detection in the high-SNR regime for the MIMO-OMA system considered, while it becomes inferior to MRC detection in the low-SNR regime.
Furthermore, we can observe a higher ESG in Fig. \ref{APGVsK:a}, Fig. \ref{APGVsK:b}, and Fig. \ref{APGVsK:c} for the high-SNR case, e.g. ${\rm{SNR}_{sum}} = 20$ dB.
This is due to the power-domain multiplexing of NOMA, which enables multiple users to share the same time-frequency resource and motivates a more efficient exploitation of the power resource.

\subsection{ESG versus the SNR in Single-cell Systems}
\begin{figure}[t]
\centering
\subfigure[ESG of SISO-NOMA over SISO-OMA.]
{\label{APGVsSNR:a} 
\includegraphics[width=3.25in]{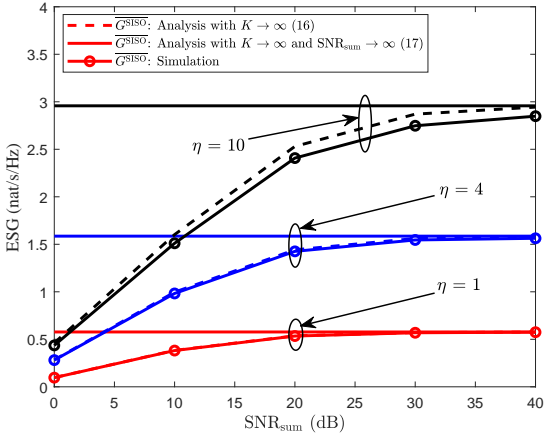}}
\hspace{-7mm}
\subfigure[ESG of MIMO-NOMA over MIMO-OMA with FDMA-ZF.]
{\label{APGVsSNR:b} 
\includegraphics[width=3.25in]{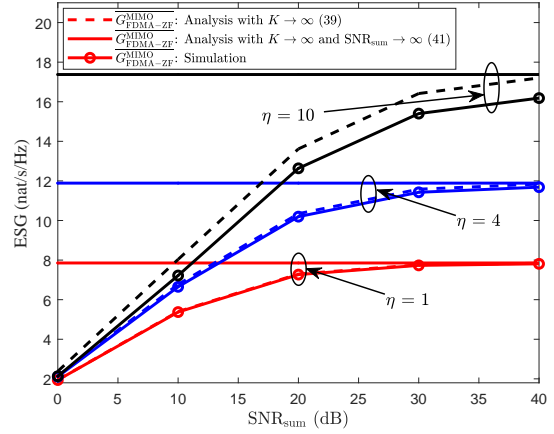}}
\subfigure[ESG of MIMO-NOMA over MIMO-OMA with FDMA-MRC.]
{\label{APGVsSNR:c} 
\includegraphics[width=3.25in]{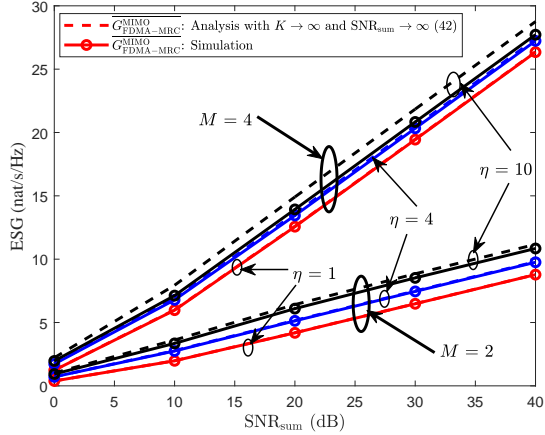}}
\hspace{-7mm}
\subfigure[ESG of \emph{m}MIMO-NOMA over \emph{m}MIMO-OMA.]
{\label{APGVsSNR:d} 
\includegraphics[width=3.25in]{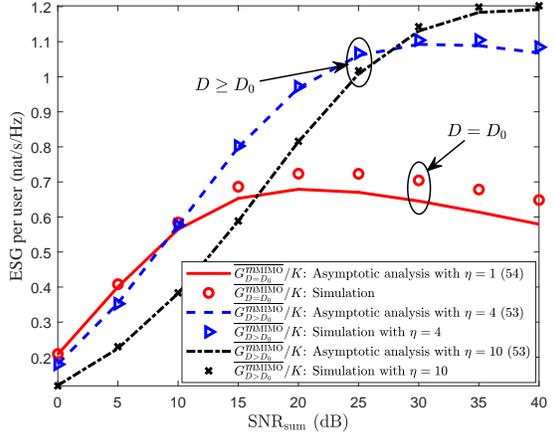}}
\caption{The ESG of NOMA over OMA versus ${\rm{SNR}_{sum}}$. The number of users is $K = 256$ and the normalized cell size is $\eta = [1,4,10]$. In Fig. \ref{APGVsSNR:b}, the number of antennas equipped at the BS $M=4$, while we have $M=[2,4]$ in Fig. \ref{APGVsSNR:c}. In Fig. \ref{APGVsSNR:d}, we have $M = 128$ such that $\delta = \frac{M}{K} = \frac{1}{2}$.}
\label{APGVsSNR}%
\end{figure}


Fig. \ref{APGVsSNR} depicts the ESG of NOMA over OMA versus the system's SNR ${\rm{SNR}_{sum}}$ within the range of ${\rm{SNR}_{sum}} = [0, 40]$ dB in the single-antenna, multi-antenna, and massive-MIMO single-cell systems.
We can observe that the simulation results match closely our asymptotic analyses in all the considered cases.
Besides, by increasing the system SNR, the ESGs seen in Fig. \ref{APGVsSNR:a} and Fig. \ref{APGVsSNR:b} increase monotonically and approach the asymptotic analyses results derived in the high-SNR regime.
In other words, the ESGs seen in Fig. \ref{APGVsSNR:a} and Fig. \ref{APGVsSNR:b} are bounded from above even if ${P_{{\rm{max}}}} \to \infty $.
This is because there is no DoF gain in the ESG of NOMA over OMA in the pair of scenarios considered.
By contrast, as derived in \eqref{EPGMIMOERA6}, the $(M-1)$-fold DoF gain in the ESG of MIMO-NOMA over MIMO-OMA with FDMA-MRC enables the ESG to increase linearly with the system's SNR in dB in the high-SNR regime, as shown in Fig. \ref{APGVsSNR:c}.
Furthermore, a higher number of antennas provides a larger DoF gain, which leads to a steeper slope of ESG versus the system SNR in dB.
In contrast to Fig. \ref{APGVsSNR:a}, Fig. \ref{APGVsSNR:b}, and Fig. \ref{APGVsSNR:c}, the ESG of \emph{m}MIMO-NOMA over \emph{m}MIMO-OMA recorded in Fig. \ref{APGVsSNR:d} first increases and then decreases with the system SNR, especially for a small normalized cell size.
In fact, the \emph{m}MIMO-NOMA system relying on MRC-SIC detection becomes interference-limited in the high-SNR regime, while the \emph{m}MIMO-OMA system remains interference-free, since favorable propagation conditions prevail for $\varsigma = \frac{W}{M} \ll 1$.
As a result, upon increasing the system SNR, the increased IUI of the \emph{m}MIMO-NOMA system considered neutralizes some of its ESG over the \emph{m}MIMO-OMA system, particularly for a small cell size associated with a limited large-scale near-far gain.

On the other hand, it is worth noticing in Fig. \ref{APGVsSNR:a}, that if all the users are randomly distributed on a circle when $D = D_0 = 50$ m, i.e., $\eta = 1$, then we have an ESG of about $0.575$ nat/s/Hz at ${\rm{SNR}_{sum}} = 40$ dB for SISO-NOMA compared to SISO-OMA.
This again verifies the accuracy of the small-scale fading gain $\gamma$ derived in \eqref{EPGSISOERA2}.
Furthermore, we can observe in Fig. \ref{APGVsSNR:a}, Fig. \ref{APGVsSNR:b}, and Fig. \ref{APGVsSNR:c}, that a larger normalized cell size $\eta$ results in a higher performance gain, which is an explicit benefit of the increased large-scale near-far gain $\vartheta \left( \eta \right)$.
By contrast, in Fig. \ref{APGVsSNR:d}, a larger cell size facilitates a higher ESG but only in the high-SNR regime, while a smaller cell size can provide a larger ESG in the low to moderate-SNR regime.
In fact, due to the large number of antennas, the IUI experienced in the \emph{m}MIMO-NOMA system is significantly reduced compared to that in single-antenna and multi-antenna systems.
As a result, in the low to moderate-SNR regime, the \emph{m}MIMO-NOMA system considered may be noise-limited rather than interference-limited, which is in line with the single-antenna and multi-antenna systems.
For instance, the noise degrades the achievable rates of the cell-edge users more severely compared to the impact of IUI in the \emph{m}MIMO-NOMA system, especially for large normalized cell sizes.
Therefore, the large-scale near-far gain cannot be fully exploited in the low to moderate-SNR regime in the massive-MIMO systems.
Moreover, it can be observed in Fig. \ref{APGVsSNR:d} that the ESG increases faster for a larger normalized cell size $\eta$.
This is due to the enhanced large-scale near-far gain observed for a larger cell size, which enables NOMA to exploit the power resource more efficiently.

\subsection{ESG versus the Number of Antennas $M$ in Single-cell Systems}

\begin{figure}[t]
\centering
\vspace{-4mm}
\subfigure[ESG of MIMO-NOMA over MIMO-OMA with FDMA-ZF.]
{\label{MIMOAPGVsM:a} 
\includegraphics[width=2.07in]{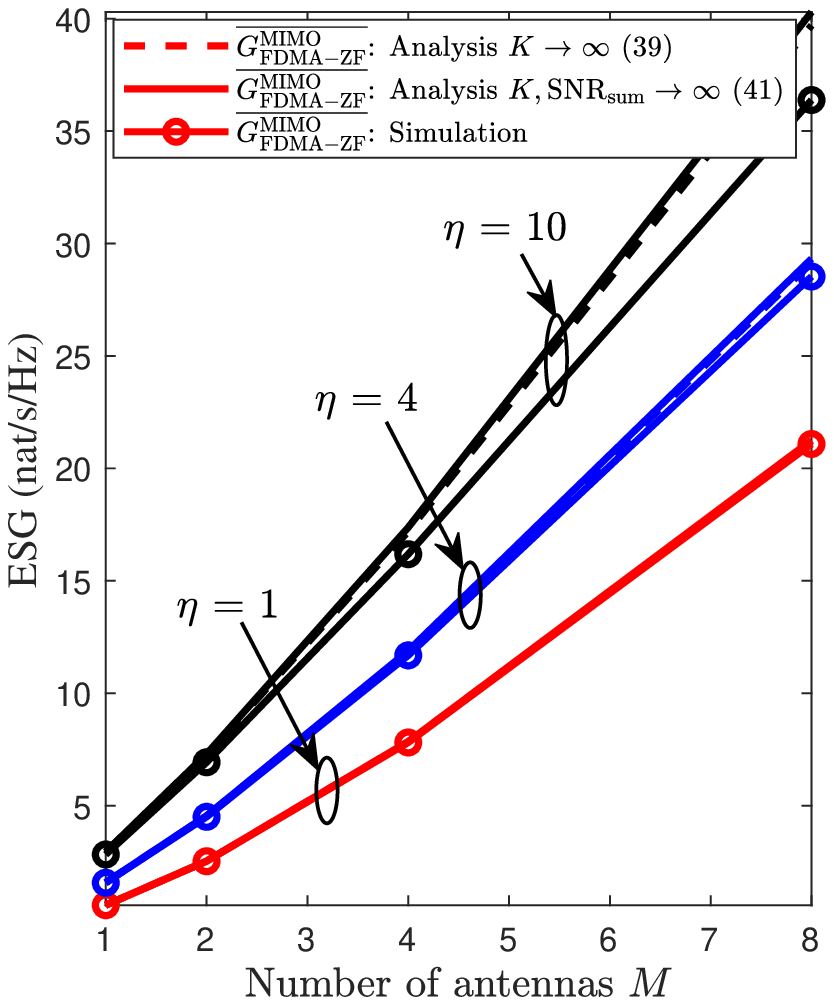}}
\subfigure[ESG of MIMO-NOMA over MIMO-OMA with FDMA-MRC.]
{\label{MIMOAPGVsM:b} 
\includegraphics[width=2.07in]{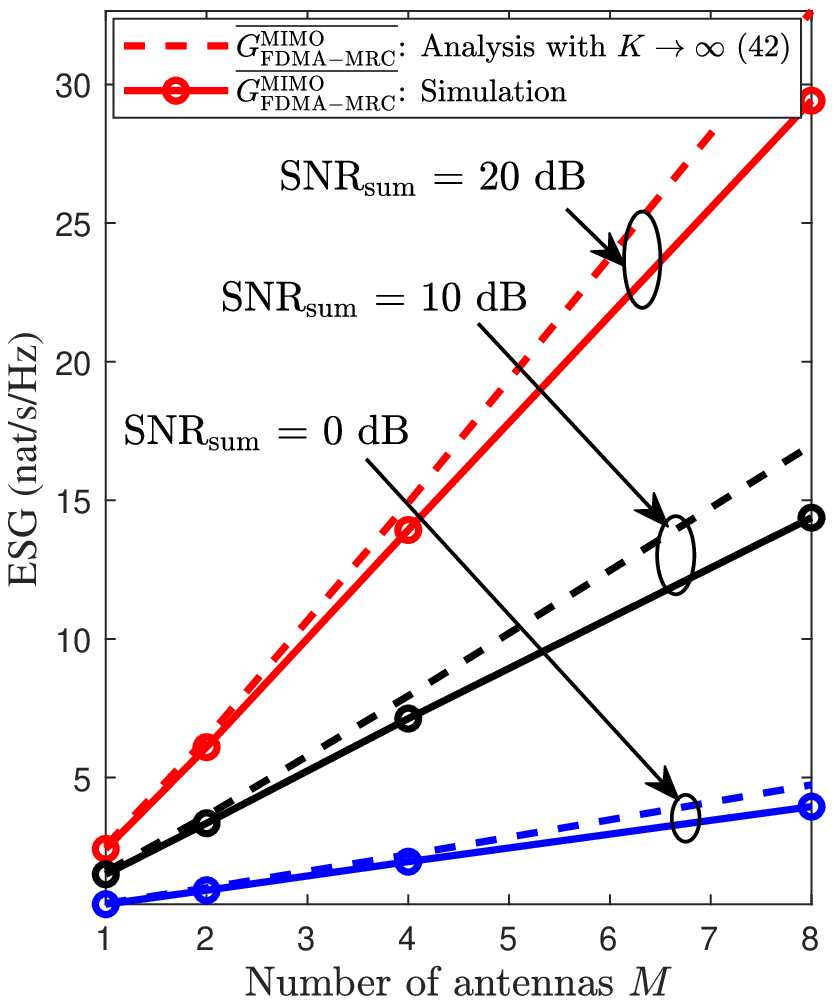}}
\subfigure[ESG of \emph{m}MIMO-NOMA over \emph{m}MIMO-OMA.]
{\label{MIMOAPGVsM:c} 
\includegraphics[width=2.07in]{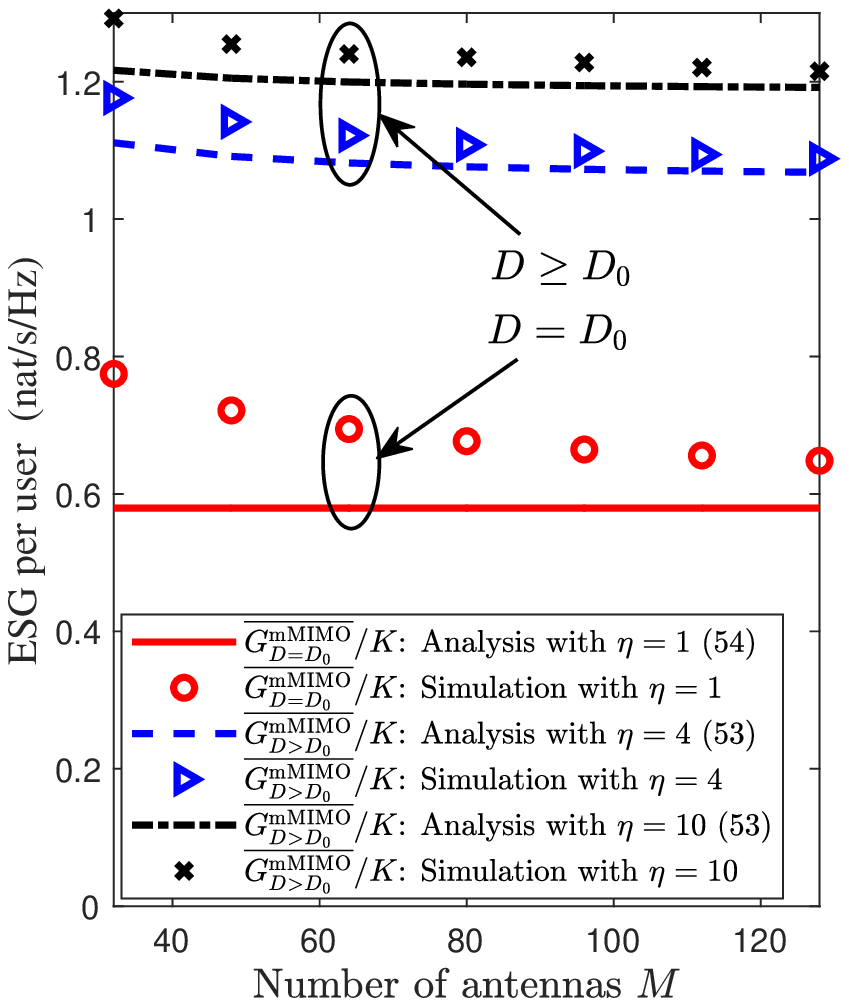}}
\caption{The ESG of NOMA over OMA versus the number of antennas $M$. The number of users is $K = 256$ in Fig. \ref{MIMOAPGVsM:a} and Fig. \ref{MIMOAPGVsM:b}. The normalized cell size is $\eta = [1,4,10]$ in Fig. \ref{MIMOAPGVsM:a} and Fig. \ref{MIMOAPGVsM:c} while it is set as $\eta = [10]$ in Fig. \ref{MIMOAPGVsM:b}. The average received sum SNR is ${\rm{SNR}_{sum}} = [0,10,20]$ dB in Fig. \ref{MIMOAPGVsM:b}, while it is set as ${\rm{SNR}_{sum}} = [40]$ dB in Fig. \ref{MIMOAPGVsM:a} and Fig. \ref{MIMOAPGVsM:c}. In multi-antenna systems in Fig. \ref{MIMOAPGVsM:a} and Fig. \ref{MIMOAPGVsM:b}, the number of antennas $M$ equipped at the BS ranges from $1$ to $8$.
In massive-MIMO systems in Fig. \ref{MIMOAPGVsM:c}, $M$ ranges from $32$ to $128$, and the number of users $K$ is adjusted according to $M$ based on $K = \frac{M}{\delta}$ with $\delta = \frac{1}{2}$.}
\label{MIMOAPGVsM}%
\end{figure}

Fig. \ref{MIMOAPGVsM} illustrates the ESG of NOMA over OMA versus the number of antennas $M$ employed at the BS in multi-antenna and massive-MIMO systems.
It can be observed that the simulation results closely match our asymptotic analyses  for all the simulation scenarios.
In particular, observe for the ESG of MIMO-NOMA over MIMO-OMA with FDMA-ZF in Fig. \ref{MIMOAPGVsM:a} that as predicted in \eqref{EPGMIMOERA3}, the asymptotic ESG $\overline{G^{{\rm{SISO}}}}$ of single-antenna systems is increased by $M$, when an $M$-antenna array is employed at the BS.
More importantly, a larger normalized cell size $\eta$ enables a steeper slope in the ESG versus the number of antennas $M$, which is due to the increased large-scale near-far gain $\vartheta \left( \eta \right)$, as shown in \eqref{EPGMIMOERA2}.
Apart from the linearly increased component of ESG vesus $M$, an additional power gain factor of $\ln\left(M\right)$ can also be observed in Fig. \ref{MIMOAPGVsM:a} as derived in \eqref{EPGMIMOERA3}.
Observe the ESG of MIMO-NOMA over MIMO-OMA with FDMA-MRC in Fig. \ref{MIMOAPGVsM:b} that the ESG grows linearly versus $M$ due to the $(M-1)$-fold of DoF gain and the corresponding slope becomes higher for a higher system SNR, as seen in \eqref{EPGMIMOERA6}.
The ESG per user seen in Fig. \ref{MIMOAPGVsM:c} for massive-MIMO systems remains almost constant upon increasing $M$, which matches for our asymptotic analysis in \eqref{DD0EPGPerUsermMIMOERA2}, and is also consistent with the results of Fig. \ref{APGVsK:c} for the fixed ratio $\delta = \frac{M}{K}$.
Furthermore, we can observe that a large cell size offers a higher ESG per user due to the improved large-scale near-far gain.

\subsection{ESG versus the Total Transmit Power in Multi-cell Systems}
In a multi-cell system, we consider a high user density scenario within a large circular area with the radius of $D_1 = 5$ km and the user density of $\rho = 1000$ devices per $\rm{km}^{\rm{2}}$.
As a result, the total number of users in the multi-cell system considered is $K' = \left\lceil\rho \pi D_1^2\right\rceil$.
Then, the number of users in each cell is given by $K = \left\lceil {\rho \pi {D^2}} \right\rceil$ with $D$ in the unit of km.
Meanwhile, the number of adjacent cells $L$ can be obtained by $L = \left\lceil \frac{K' - K}{K} \right\rceil$, so that all the $K'$ users can be covered.
Furthermore, the $K'$ users in all the cells share a given total transmit power and the total transmit power ${P_{\rm{max}}'}$ of $(L+1)$ cells is within the range spanning from $20$ dBm to $60$ dBm\footnote{Since there are a larger number of users deployed in the considered area, we set a large power budget for all the users in the considered multi-cell system.}.
In this section, we also consider an equal power allocation among multiple cells and an equal power allocation among users within each cell, i.e., we have ${P_{\rm{max}}} = \frac{P_{\rm{max}}'}{L+1}$ and ${p_{k',l}} = \frac{{P_{\rm{max}}}}{K}$.
All the other simulation parameters are the same as those adopted in the single-cell systems.

\begin{figure}[t]
\centering
\subfigure[The ESG of SISO-NOMA over SISO-OMA.]
{\label{MultiCell:a} 
\includegraphics[width=3.25in]{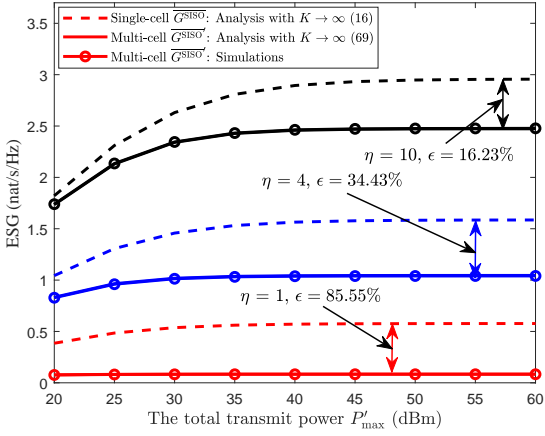}}
\hspace{-7mm}
\subfigure[ESG of MIMO-NOMA over MIMO-OMA with FDMA-ZF.]
{\label{MultiCell:b} 
\includegraphics[width=3.25in]{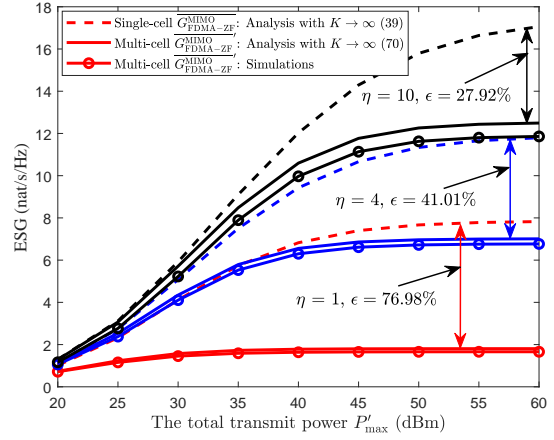}}
\subfigure[ESG of MIMO-NOMA over MIMO-OMA with FDMA-MRC.]
{\label{MultiCell:c} 
\includegraphics[width=3.25in]{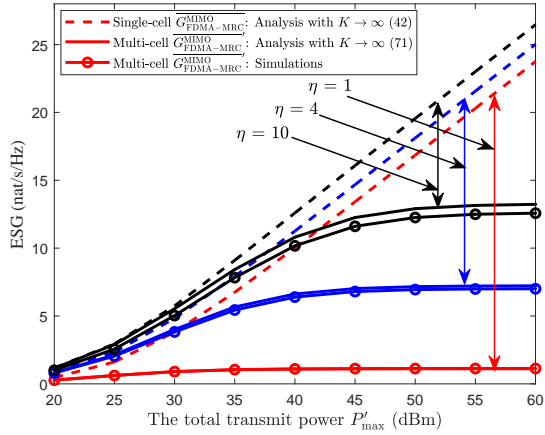}}
\hspace{-7mm}
\subfigure[ESG of \emph{m}MIMO-NOMA over \emph{m}MIMO-OMA.]
{\label{MultiCell:d} 
\includegraphics[width=3.25in]{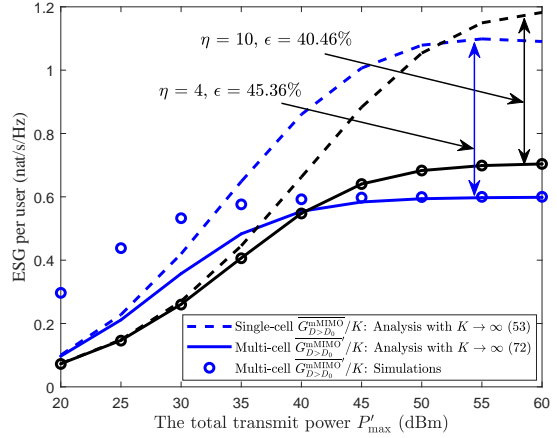}}
\caption{ESG versus the total transmit power ${P_{\rm{max}}'}$ in multi-cell systems. The normalized cell size is $\eta = [1,4,10]$ in Fig. \ref{MultiCell:a}, Fig. \ref{MultiCell:b}, and \ref{MultiCell:c}, while it is set as $\eta = [4,10]$ in Fig. \ref{MultiCell:d}. The number of antennas equipped at each BS is $M=1$ in Fig. \ref{MultiCell:a} and $M=4$ in Fig. \ref{MultiCell:b} as well as Fig. \ref{MultiCell:c}.
In Fig. \ref{MultiCell:d}, the number of antennas equipped at each BS is adjusted based on the number of users in each cell via $M = \left\lceil {K}{\delta} \right\rceil$ and the group size of the considered \emph{m}MIMO-OMA system is $W = \left\lceil {\varsigma}M \right\rceil$.
The ESG degradations due to the ICI are denoted by double-sided arrows.}
\label{MultiCell}%
\end{figure}

In contrast to the single-cell systems, the ESG versus the total transmit power ${P_{\rm{max}}'}$ trends are more interesting, which is due to the less straightforward impact of ICI on the performance gain of NOMA over OMA in multi-cell systems.
Fig. \ref{MultiCell} shows the ESG of NOMA over OMA versus the total transmit power ${P_{\rm{max}}'}$ in single-antenna, multi-antenna, and massive-MIMO\footnote{Note that, for the considered massive-MIMO multi-cell system, a small cell size leads to a small number of users $K$ in each cell and thus results in a small number of antennas $M$ due to the fixed ratio $\delta = \frac{M}{K}$. This is contradictory to our assumption of $K \to \infty$ and $M \to \infty$. Therefore, we only consider the normalized cell size of $\eta = [4,10]$ for massive-MIMO multi-cell systems in Fig. \ref{MultiCell:d}.} multi-cell systems.
The analytical results in single-cell systems are also shown for comparison.
We can observe that the performance gains of NOMA over OMA are degraded upon extending NOMA from single-cell systems to multi-cell systems.
In fact, NOMA schemes enable all the users in adjacent cells to simultaneously transmit their signals on the same frequency band and thus the ICI level in NOMA schemes is substantially higher than that in OMA schemes, as derived in \eqref{MulticellInterference3}.
For the ease of illustration, we define the normalized performance degradation of the ESG in multi-cell systems compared to that in single-cell systems as $\epsilon = \frac{{\overline G - \overline G'}}{{\overline G}}$, where $\overline G$ denotes the ESG in single-cell systems and $\overline G'$ denotes the ESG in multi-cell systems.
It can be observed that the ESG degradation is more severe for a small normalized cell size $\eta$, because multi-cell systems suffer from a more severe ICI for smaller cell sizes due to a shorter inter-site distance.
Therefore, the system performance becomes saturated even for a moderate system power budget in the case of a smaller cell size.

It is worth noting that the ESG of MIMO-NOMA over MIMO-OMA with FDMA-MRC is saturated in multi-cell systems in the high transmit power regime, as shown in Fig. \ref{MultiCell:c}, which is different from the trends seen for single-cell systems in Fig. \ref{APGVsSNR:c}.
In fact, the $(M-1)$-fold DoF gain in the ESG of MIMO-NOMA over MIMO-OMA with FDMA-MRC in single-cell systems derived in \eqref{EPGMIMOERA6} can only be achieved in the high-SNR regime.
However, due to the lack of joint multi-cell signal processing to mitigate the ICI, the multi-cell system becomes interference-limited upon increasing the total transmit power.
Therefore, the multi-cell system actually operates in the low-SINR regime, which does not facilitate the exploitation of the DoF gain in single-cell systems.

\begin{table}
	\caption{Comparison on ESG (nat/s/Hz) of NOMA over OMA in the considered scenarios. The system setup is $K = 256$, $D = 200$ m, $\eta = 4$, and $M = 4$ for  multi-antenna systems.}
	\centering
	\begin{tabular}{c|cc|cc}
		\hline
		               & \multicolumn{2}{c}{${\rm{SNR}_{sum}} = 0$ dB}  \vline& \multicolumn{2}{c}{${\rm{SNR}_{sum}} = 10$ dB} \\\hline
		               & Single-cell    & Multi-cell                          & Single-cell    & Multi-cell                           \\\hline
		Single-antenna & 0.281          & 0.2639                              & 0.983          & 0.7973                                 \\
		Multi-antenna  & 2.114          & 2.0179                              & 6.65           & 5.4113                                  \\
		Massive-MIMO (ESG per user)     & 0.1796         & 0.1702                              & 0.5765         & 0.4490                      \\
		\hline
	\end{tabular}\label{SimulationComparison}
\end{table}

\begin{Remark}
The comparison of the ESG (nat/s/Hz) results of NOMA over OMA in all the scenarios considered is summarized in Table \ref{SimulationComparison}.
We consider a practical operation setup with $K = 256$, $D = 200$ m, $\eta = 4$, and $M = 4$ for the multi-antenna systems.
For fair comparison, the total transmit power ${P_{\rm{max}}'}$ in multi-cell systems is adjusted for ensuring that the total average received SNR ${\rm{SNR}_{sum}}$ at the serving BS is identical to that in single-cell systems.
Note that the row of massive-MIMO in Table \ref{SimulationComparison} quantifies the ESG per user of \emph{m}MIMO-NOMA over \emph{m}MIMO-OMA, which is consistent with Fig. \ref{APGVsSNR:d}, Fig. \ref{MIMOAPGVsM:c} and Fig. \ref{MultiCell:d}.
We can observe that the ESG remains a near-constant at the low SNR of ${\rm{SNR}_{sum}} = 0$ dB when extending NOMA from single-cell systems to multi-cell systems, while the ESG degrades substantially at the high SNR of ${\rm{SNR}_{sum}} = 10$ dB.
In fact, the limited transmit power budget in the low-SNR regime in adjacent cells only leads to a low ICI level at the serving BS, which avoids a significant ESG degradation, when applying NOMA in multi-cell systems.
\end{Remark}

\section{Conclusions and Future Work}
In this paper, we investigated the ESG in uplink communications attained by NOMA over OMA in single-antenna, multi-antenna, and massive-MIMO systems with both single-cell and multi-cell deployments.
In the single-antenna single-cell system considered, the ESG of NOMA over OMA was quantified and two types of gains were identified in the ESG derived, i.e., the large-scale near-far gain and the small-scale fading gain.
The large-scale near-far gain increases with the cell size, while the small-scale fading gain is a constant of $\gamma = 0.57721$ nat/s/Hz in Rayleigh fading channels.
Additionally, we unveiled that the ESG of SISO-NOMA over SISO-OMA can be increased by $M$ times upon using $M$ antennas at the BS, owing to the extra spatial DoF offered by additional antennas.
In the massive-MIMO single-cell system considered, the ESG of NOMA over OMA increases linearly both with the number of users and with the number of antennas at the BS.
The analytical results derived for single-cell systems were further extended to multi-cell systems via characterizing the effective ICI channel distribution and by deriving the ICI power.
We found that a larger cell size is preferred by NOMA for both single-cell and multi-cell systems, due to the enhanced large-scale near-far gain and reduced ICI, respectively.
Extensive simulation results have shown the accuracy of our performance analyses and confirmed the insights provided above.

In this paper, as a first attempt to unveil fundamental insights on the performance gain of NOMA over OMA, we considered the ideal case associated with perfect CSI and error-propagation-free SIC detection at the BS.
In practice, it is difficult to acquire the perfect CSI due to channel estimation errors, feedback delays, and/or quantization errors.
Similarly, the error propagation during SIC decoding is usually also inevitable in practice.
In our future work, we will investigate the ESG of NOMA over OMA both in the face of imperfect CSI and error propagation during SIC detection.


\section{Acknowledgement}
The authors would like to thank Prof. Ping Li from the City University of Hong Kong for valuable discussions during this work.

\section*{Appendix}
\begin{appendices}
\subsection{Proof of Theorem \ref{Theorem1}}\label{AppendixA}
To facilitate the proof, we first consider a virtual system whose capacity serves as an upper bound to that of the system in \eqref{MIMONOMASystemModel}.
In particular, the virtual system is the uplink of a $K$-user $M \times M$ MIMO system with $M$ antennas employed at each user and the BS.
We assume that, in the virtual $K$-user $M \times M$ MIMO system, each user faces $M$ parallel subchannels with identical subchannel gain ${\left\| {{{\bf{h}}_k}} \right\|}$, i.e., the channel matrix between user $k$ and the BS is ${\left\| {{{\bf{h}}_k}} \right\|}{{\bf{I}}_M}$.
As a result, the signal received at the BS is given by
\begin{equation}\label{MIMONOMASystemModelRelaxed}
{\bf{\tilde y}} = \sum\limits_{k = 1}^K \sqrt {{p_k}} {\left\| {{{\bf{h}}_k}} \right\|}{{\bf{I}}_M}{{{\bf{\tilde x}}}_k} + {\bf{v}},
\end{equation}
where ${{{\bf{\tilde x}}}_k} = {{\bf{u}}_k}{x_k} \in \mathbb{C}^{M \times 1}$ denotes the transmitted signal after preprocessing by a precoder ${{\bf{u}}_k}\in \mathbb{C}^{M \times 1}$.
We note that the precoder should satisfy the constraint $ \rm{Tr} \left({{{\bf{u}}_k}} {{{\bf{u}}_k^{\rm{H}}}} \right)\le 1$, so that ${\rm{E}}\left\{ {{\bf{\tilde x}}_k^{\rm{H}}{{{\bf{\tilde x}}}_k}} \right\} \le {\rm{E}}\left\{ {x_k^{\rm{2}}} \right\} = 1$.
Additionally, in the virtual $K$-user $M \times M$ MIMO system, the subchannel gain between user $k$ and the BS is forced to be identical as ${\left\| {{{\bf{h}}_k}} \right\|}$, where ${\left\| {{{\bf{h}}_k}} \right\|}$ is the corresponding channel gain value between user $k$ and the BS in the original $K$-user $1 \times M$ MIMO system in \eqref{MIMONOMASystemModel}.
Furthermore, we consider an arbitrary but the identical power allocation strategy ${\bf{p}} = \left[ {{p_1}, \ldots ,{p_K}} \right]$ as that of our original system in \eqref{MIMONOMASystemModel} during the following proof.
Upon comparing \eqref{MIMONOMASystemModel} and \eqref{MIMONOMASystemModelRelaxed}, we can observe that the specific choice of the precoder ${{{\bf{u}}_k}} = \frac{{{{\bf{h}}_k}}}{\left\| {{{\bf{h}}_k}} \right\|}$ in \eqref{MIMONOMASystemModelRelaxed} would result in an equivalent system to that in \eqref{MIMONOMASystemModel}.
In other words, the capacity of the system in \eqref{MIMONOMASystemModelRelaxed} serves as an upper bound to that of the system in \eqref{MIMONOMASystemModel}, i.e., we have:
\begin{align}\label{UpperBoundMIMONOMA}
R_{\rm{sum}}^{{\rm{MIMO-NOMA}}} \mathop  =\limits^{(a)} C\left( {M, K, {\bf{p}},{\bf{H}}} \right) &\le {C}\left( {M^2, K,{\bf{p}},\widetilde{{\bf{H}}}} \right)\notag\\
&= \mathop {\max }\limits_{{\rm{Tr}}\left( {{{\bf{u}}_k}{\bf{u}}_k^{\rm{H}}} \right) \le 1} {\ln}\left| {{{\bf{I}}_M} + \frac{1}{{{N_0}}}\sum\limits_{k = 1}^K {{p_k}{{\left\| {{{\bf{h}}_k}} \right\|}^2}{{\bf{I}}_M}{{\bf{u}}_k}{\bf{u}}_k^{\rm{H}}{\bf{I}}_M^{\rm{H}}} } \right| \notag\\
&= M{\ln}\left( {1 + \frac{1}{{M{N_0}}}\sum\limits_{k = 1}^K {{p_k}{{\left\| {{{\bf{h}}_k}} \right\|}^2}} } \right),
\end{align}
where $C\left( {M, K, {\bf{p}},{\bf{H}}} \right)$ denotes the capacity for the uplink $K$-user $1 \times M$ MIMO system in \eqref{MIMONOMASystemModel} for a channel matrix ${\bf{H}} = \left[ {{{\bf{h}}_1}, \ldots ,{{\bf{h}}_K}} \right]$ and power allocation ${\bf{p}}$.
Furthermore, $C\left( {M^2, K, {\bf{p}},\widetilde{{\bf{H}}}} \right)$ denotes the capacity of the virtual $K$-user $M \times M$ MIMO system in \eqref{MIMONOMASystemModelRelaxed} associated with a channel matrix $\widetilde{{\bf{H}}} = \left[ {{\left\| {{{\bf{h}}_1}} \right\|}{\bf{I}}_M, \ldots ,{\left\| {{{\bf{h}}_K}} \right\|}{\bf{I}}_M} \right]$, while ${\bf{p}}$ is the value as in \eqref{MIMONOMASystemModel}.
The achievable sum-rate $R_{\rm{sum}}^{{\rm{MIMO-NOMA}}}$  is given in \eqref{InstantSumRateMIMONOMA} and the equality $(a)$ in \eqref{UpperBoundMIMONOMA} is obtained by a capacity-achieving MMSE-SIC\cite{Tse2005}.

Now, to prove the asymptotic tightness of the upper bound considered in \eqref{UpperBoundMIMONOMA}, we have to consider a lower bound of the achievable sum-rate in \eqref{InstantSumRateMIMONOMA} and prove that asymptotically the upper bound and the lower bound converge to the same expression.
For the uplink $K$-user $1 \times M$ MIMO system in \eqref{MIMONOMASystemModel}, we assume that all the users transmit their signals subject to the power allocation ${\bf{p}} = \left[ {{p_1}, \ldots ,{p_K}} \right]$ and the BS utilizes an MRC-SIC receiver to retrieve the messages of all the $K$ users.
Then the achievable rate for user $k$ of the MIMO-NOMA system using the MRC-SIC receiver is given by:
\begin{align}\label{MIMONOMAMRCSICIndividualAchievableRate}
R_{k,\rm{MRC-SIC}}^{{\rm{MIMO-NOMA}}} = {\ln}\left(1+ {\frac{{{p_k}{{\left\| {{{\bf{h}}_k}} \right\|}^2}}}{{\sum\limits_{i = k + 1}^K {p_i}{{\left\| {{{\bf{h}}_i}} \right\|}^2} {{\left| {\bf{e}}_k^{\rm{H}}{{\bf{e}}_i} \right|}^2}  + {N_0}}}} \right),
\end{align}
where ${{\bf{e}}_k} = \frac{{{{\bf{h}}_k}}}{{\left\| {{{\bf{h}}_k}} \right\|}}$ denotes the channel direction of user $k$.
Then, it becomes clear that the achievable sum-rate of the MIMO-NOMA system using the MRC-SIC receiver serves as a lower bound to the channel capacity in \eqref{InstantSumRateMIMONOMA}, i.e., we have
\begin{align}\label{LowerBoundMIMONOMA2}
R_{\rm{sum,MRC-SIC}}^{{\rm{MIMO-NOMA}}} = \sum \limits_{k=1}^{K} R_{k,\rm{MRC-SIC}}^{{\rm{MIMO-NOMA}}} \le R_{\rm{sum}}^{{\rm{MIMO-NOMA}}}.
\end{align}

Through the following theorem and corollaries, we first characterize the statistics of ${{\bf{e}}_k}$ as well as ${{\left| {\bf{e}}_k^{\rm{H}}{{\bf{e}}_i} \right|}^2}$ and derive the asymptotic achievable sum-rate of MIMO-NOMA employing an MRC-SIC receiver.
Then, we show that the upper bound considered in \eqref{UpperBoundMIMONOMA} and the lower bound of \eqref{LowerBoundMIMONOMA2} will asymptotically converge to the same limit for $K \to \infty$.

\begin{Lem}\label{UnitSphere}
For ${\bf{h}}_k \sim \mathcal{CN}\left(\mathbf{0},\frac{1}{1+d_k^{\alpha}}{{\bf{I}}_M}\right)$, the normalized random vector (channel direction) ${{\bf{e}}_k} = \frac{{{{\bf{h}}_k}}}{{\left\| {{{\bf{h}}_k}} \right\|}}$ is uniformly distributed on a unit sphere in $\mathbb{C}^{M}$.
\end{Lem}
\begin{proof}
According to the system model of ${{{\bf{h}}_k}} = \frac{{\bf{g}}_k}{\sqrt{1+d_k^{\alpha}}}$ with ${\bf{g}}_k \sim \mathcal{CN}\left(\mathbf{0},{{\bf{I}}_M}\right)$, we have ${\bf{h}}_k \sim \mathcal{CN}\left(\mathbf{0},\frac{1}{1+d_k^{\alpha}}{{\bf{I}}_M}\right)$.
The distribution of ${{\bf{e}}_k}$ can be proven by exploiting the orthogonal-invariance of the multivariate normal distribution.
In particular, for any orthogonal matrix $\mathbf{Q}$, we have $\mathbf{Q}{\bf{h}}_k \sim \mathcal{CN}\left(\mathbf{0},\frac{1}{1+d_k^{\alpha}}{{\bf{I}}_M}\right)$, which means that the distribution of ${\bf{h}}_k$ is invariant to rotations (orthogonal transform).
Then, ${{\bf{e}}_k} = \frac{{{\mathbf{Q}{\bf{h}}_k}}}{{\left\| {{\mathbf{Q}{\bf{h}}_k}} \right\|}} = \frac{{{\mathbf{Q}{\bf{h}}_k}}}{{\left\| {{{\bf{h}}_k}} \right\|}}$ is also invariant to rotation.
Meanwhile, we have ${{\left\| {{{\bf{e}}_k}} \right\|}} = 1$ for sure.
Therefore, ${{\bf{e}}_k}$ must be uniformly distributed on a unit sphere on $\mathbb{C}^{M}$.
\end{proof}

\begin{Cor}\label{Corollary1}
The channel direction of user $k$, ${{\bf{e}}_k}$, is independent of its channel gain ${{\left\| {{{\bf{h}}_k}} \right\|}}$.
\end{Cor}
\begin{proof}
According to Lemma \ref{UnitSphere}, the channel direction ${{\bf{e}}_k}$ is uniformly distributed on a unit sphere on $\mathbb{C}^{M}$, regardless of the value of ${{\left\| {{{\bf{h}}_k}} \right\|}}$.
Therefore, ${{\bf{e}}_k}$ is independent of ${{\left\| {{{\bf{h}}_k}} \right\|}}$.
\end{proof}

\begin{Cor}
The mean and covariance matrix of ${{\bf{e}}_k}$ are given by
\begin{align}
{\rm{E}}\left\{ {{{\bf{e}}_k}} \right\} = {\bf{0}}\; \text{and} \;
{\rm{E}}\left\{ {{{\bf{e}}_k}{\bf{e}}_k^{\rm{H}}} \right\} = \frac{1}{M}{{\bf{I}}_M},
\end{align}
respectively.
\end{Cor}
\begin{proof}
Due to the symmetry of the uniform spherical distribution, ${{{\bf{e}}_k}}$ and $-{{{\bf{e}}_k}}$ have the same distribution and thus we have ${\rm{E}}\left\{ {{{\bf{e}}_k}} \right\} =  {\rm{E}}\left\{ -{{{\bf{e}}_k}} \right\}$ and hence ${\rm{E}}\left\{ {{{\bf{e}}_k}} \right\} = {\bf{0}}$.
For the reason of symmetry, ${{{\bf{e}}_k}} = \left[ {{e_{k,1}}, \ldots ,{e_{k,m}}, \ldots ,{e_{k,M}}} \right]$ and ${{{\bf{e}}_k'}} = \left[ {{e_{k,1}}, \ldots ,-{e_{k,m}}, \ldots ,{e_{k,M}}} \right]$ have the same distribution, where ${e_{k,m}}$ denotes the $m$-th entry in ${{{\bf{e}}_k}}$.
Therefore, we have
\begin{equation}
{\rm{E}}\left\{ {{e_{k,m}}e_{k,n}^*} \right\} = {\rm{E}}\left\{ { - {e_{k,m}}e_{k,n}^*} \right\} =  - {\rm{E}}\left\{ {{e_{k,m}}e_{k,n}^*} \right\}, \forall m \neq n,
\end{equation}
which implies that the covariance terms are zero, i.e., ${\rm{E}}\left\{ {{e_{k,m}}e_{k,n}^*} \right\} = 0$, $\forall m \neq n$.
Note that, the zero covariance terms only reflect the lack of correlation between ${e_{k,m}}$ and $e_{k,n}$, but not their independence.
In fact, the entries of ${{{\bf{e}}_k}}$ are dependent on each other, i.e., increasing one entry will decrease all the other entries due to ${{\left\| {{{\bf{e}}_k}} \right\|}} = 1$.
As for the variance, since ${{{\bf{e}}_k}}$ has been normalized, we have
\begin{equation}
\sum\limits_{m = 1}^M {{\rm{E}}\left\{ {e_{k,m}^2} \right\}}  = {\rm{E}}\left\{ {\sum\limits_{m = 1}^M {e_{k,m}^2} } \right\} = 1.
\end{equation}
Again, based on the symmetry of the uniform spherical distribution, we have ${{\rm{E}}\left\{ {e_{k,m}^2} \right\}} = {{\rm{E}}\left\{ {e_{k,n}^2} \right\}}$, $\forall m,n$, and hence we have ${{\rm{E}}\left\{ {e_{k,m}^2} \right\}} = \frac{1}{M}$ and ${\rm{E}}\left\{ {{{\bf{e}}_k}{\bf{e}}_k^{\rm{H}}} \right\} = \frac{1}{M}{{\bf{I}}_M}$.
This completes the proof.
\end{proof}

Let us now define a scalar random variable as ${\nu_{k,i}} = {{\bf{e}}_k^{\rm{H}}{{\bf{e}}_i}} \in \mathbb{C}$, which denotes the projection of channel direction of user $k$ on the channel direction of user $i$.
Note that the random variable ${\nu_{k,i}}$ can characterize the IUI during MRC in \eqref{MIMONOMAMRCSICIndividualAchievableRate}.
Additionally, thanks to the independence between ${{\bf{e}}_k}$ and ${{\left\| {{{\bf{h}}_k}} \right\|}}$, ${\nu_{k,i}}$ is independent of ${{\left\| {{{\bf{h}}_k}} \right\|}}$ and ${{\left\| {{{\bf{h}}_i}} \right\|}}$.
The following Lemma characterizes the mean and variance of ${\nu_{k,i}}$.

\begin{Lem}\label{Lemma2}
For ${\bf{h}}_k \sim \mathcal{CN}\left(\mathbf{0},\frac{1}{1+d_k^{\alpha}}{{\bf{I}}_M}\right)$ and ${{\bf{e}}_k} = \frac{{{{\bf{h}}_k}}}{{\left\| {{{\bf{h}}_k}} \right\|}}$, the random variable ${\nu_{k,i}} = {{\bf{e}}_k^{\rm{H}}{{\bf{e}}_i}}$ has a zero mean and variance of $\frac{1}{M}$.
\end{Lem}
\begin{proof}
In fact, ${\nu_{k,i}}$ denotes the projection of ${\bf{e}}_k$ on ${{\bf{e}}_i}$, where ${\bf{e}}_k$ and ${{\bf{e}}_i}$ are uniformly distributed in a unit sphere on $\mathbb{C}^{M}$.
Upon fixing one channel direction ${{\bf{e}}_k}$, the conditional mean and variance of ${\nu_{k,i}}$ are given by
\begin{equation}\label{ConditionalMeanVariance}
{\rm{E}}\left\{ {{\nu_{k,i}}\left| {{{\bf{e}}_k}} \right.} \right\} = {\bf{e}}_k^{\rm{H}}{\rm{E}}\left\{ {{{\bf{e}}_i}} \right\} =0 \;\text{and}\;
{\rm{E}}\left\{ {\left|{\nu_{k,i}}\right|^2\left| {{{\bf{e}}_k}} \right.} \right\} = {{\bf{e}}_k^{\rm{H}}{\rm{E}}\left\{ {{{\bf{e}}_i}{\bf{e}}_i^{\rm{H}}} \right\}{{\bf{e}}_k}} = \frac{1}{M},
\end{equation}
respectively.
Since ${{\bf{e}}_k}$ is uniformly distributed, the integral over ${{\bf{e}}_k}$ will not change the mean and variance.
Therefore, the mean and variance of ${\nu_{k,i}}$ are given by
\begin{align}\label{ConditionalMeanVariance}
{\rm{E}}\left\{ {\nu_{k,i}} \right\} = 0 \;\text{and}\;
{\rm{E}}\left\{ \left|{\nu_{k,i}}\right|^2 \right\} = \frac{1}{M},
\end{align}
respectively, which completes the proof.
\end{proof}

Now, based on \eqref{MIMONOMAMRCSICIndividualAchievableRate}, we have the asymptotic achievable data rate of user $k$ as follows:
\begin{align}\label{AsymptoticSumRateMIMONOMA}
\mathop {\lim }\limits_{K \to \infty } R_{k,\rm{MRC-SIC}}^{{\rm{MIMO-NOMA}}} & = \mathop {\lim }\limits_{K \to \infty } {\ln}\left(1+ {\frac{{{p_k}{{\left\| {{{\bf{h}}_k}} \right\|}^2}}}{{\sum\limits_{i = k + 1}^K {p_i}{{\left\| {{{\bf{h}}_i}} \right\|}^2} \left|{\nu_{k,i}}\right|^2  + {N_0}}}} \right) \notag\\
& \mathop  =\limits^{(a)} \mathop {\lim }\limits_{K \to \infty } {\ln}\left(1+ {\frac{{{p_k}{{\left\| {{{\bf{h}}_k}} \right\|}^2}}}{{\sum\limits_{i = k + 1}^K {p_i}{{\left\| {{{\bf{h}}_i}} \right\|}^2} \frac{1}{M}  + {N_0}}}} \right) \notag\\
& \mathop  = \limits^{(b)} \mathop {\lim }\limits_{K \to \infty } M{\ln}\left(1+ {\frac{{{p_k}{{\left\| {{{\bf{h}}_k}} \right\|}^2} \frac{1}{M}}}{{\sum\limits_{i = k + 1}^K {p_i}{{\left\| {{{\bf{h}}_i}} \right\|}^2} \frac{1}{M}  + {N_0}}}} \right).
\end{align}
Note that the equality in $(a)$ holds asymptotically by applying Corollary \ref{Corollary1} and Lemma \ref{Lemma2} with $K \to \infty$.
In addition, the equality in $(b)$ holds with $K \to \infty$ since $\mathop {\lim }\limits_{x \to 0} {\ln}\left( {1 + Mx} \right) = \mathop {\lim }\limits_{x \to 0} M{\ln}\left( {1 + x} \right)$.
As a result, the asymptotic achievable sum-rate in \eqref{LowerBoundMIMONOMA2} can be obtained by
\begin{equation}\label{AchievableRateMIMONOMA2}
\mathop {\lim }\limits_{K \to \infty } R_{\rm{sum,MRC-SIC}}^{{\rm{MIMO-NOMA}}}
= \mathop {\lim }\limits_{K \to \infty } M{\ln}\left( {1 + \frac{{1}}{{M{N_0}}}\sum\limits_{k = 1}^K {p_k{{\left\| {{{\bf{h}}_k}} \right\|}^2}} } \right).
\end{equation}

Now, upon comparing \eqref{UpperBoundMIMONOMA}, \eqref{LowerBoundMIMONOMA2}, and \eqref{AchievableRateMIMONOMA2}, it can be observed that the upper bound and the lower bound considered converge when $K \to \infty$.
In other words, for any given power allocation strategy ${\bf{p}} = \left[ {{p_1}, \ldots ,{p_K}} \right]$, the upper bound in \eqref{UpperBoundMIMONOMA} is asymptotically tight.
It completes the proof.

\subsection{Proof of Theorem \ref{Theorem2}}\label{AppendixB}
Based on \eqref{AsymptoticSumRateMIMONOMA} in the proof of Theorem \ref{Theorem1} in Appendix A, under the equal resource allocation strategy, i.e., ${{p_{k}}} = \frac{{P_{\rm{max}}}}{K}$, $\forall k$, the asymptotic individual rate of user $k$ of the \emph{m}MIMO-NOMA system with the MRC-SIC detection in \eqref{mMIMONOMAIndividualAchievableRate} can be obtained by
\begin{equation}\label{mMIMONOMAIndividualAchievableRate2}
\mathop {\lim }\limits_{K \rightarrow \infty} R_{k}^{\rm{\emph{m}MIMO-NOMA}} = \mathop {\lim }\limits_{K \rightarrow \infty} {\ln}\left(1+ {\frac{{{P_{\rm{max}}}{{\left\| {{{\bf{h}}_k}} \right\|}^2}}}{{\sum\limits_{i = k + 1}^K {P_{\rm{max}}}{{\left\| {{{\bf{h}}_i}} \right\|}^2} \frac{1}{M}  + {KN_0}}}} \right).
\end{equation}
With the aid of a large-scale antenna array, i.e., $M \to \infty$, the fluctuation imposed by the small-scale fading on the channel gain can be averaged out as a benefit of channel hardening\cite{HochwaldMassiveMIMO}.
Therefore, the channel gain is mainly determined by the large-scale fading asymptotically as follows:
\begin{equation}
\mathop {\lim }\limits_{M \rightarrow \infty}\frac{{\left\| {{{\bf{h}}_k}} \right\|}^2}{M} = \frac{1}{1+d_k^{\alpha}}.
\end{equation}
As a result, the asymptotic data rate of user $k$ in \eqref{mMIMONOMAIndividualAchievableRate2} is given by:
\begin{equation}\label{mMIMONOMAIndividualAchievableRate3}
\mathop {\lim }\limits_{K \rightarrow \infty, M \rightarrow \infty} R_{k}^{\rm{\emph{m}MIMO-NOMA}} = \mathop {\lim }\limits_{K \rightarrow \infty, M \rightarrow \infty} {\ln}\left(1+ {\frac{M{{P_{\rm{max}}}\frac{1}{1+d_k^{\alpha}}}}{{\sum\limits_{i = k + 1}^K {P_{\rm{max}}}\frac{1}{1+d_i^{\alpha}}  + {KN_0}}}} \right).
\end{equation}

Based on the theory of order statistics\cite{David2003order}, the PDF of $d_k$ is given by
\begin{equation}\label{OrderedDistancePDF}
{f_{{d_k}}}\left( x \right) = k\left( {\begin{array}{*{20}{c}}
{K}\\
{k}
\end{array}} \right) {F_{{d}}}^{k-1}\left( x \right) \left(1-{F_{{d}}}\left( x \right)\right)^{K-k} {f_{{d}}}\left( x \right), D_0 \le z \le D.
\end{equation}
Thus, the mean of the large-scale fading of user $k$ can be written as
\begin{align}\label{OrderedDistanceMean}
I_k &= {{\rm{E}}_{{d_k}}}\left\{ {\frac{1}{{1 + d_k^\alpha }}} \right\} = \int_{{D_0}}^D {\frac{1}{{1 + {x^\alpha }}}} {f_{{d_k}}}\left( x \right)dx  \notag\\
& \approx \left( {\begin{array}{*{20}{c}}
K\\
k
\end{array}} \right){\frac{{k}}{{D + {D_0}}}} \sum\limits_{n = 1}^N \frac{{\beta _n}}{c_n} {\left( {\frac{{\phi _n^2 - D_0^2}}{{{D^2} - D_0^2}}} \right)^{k - 1}}{\left( {\frac{{{D^2} - \phi _n^2}}{{{D^2} - D_0^2}}} \right)^{K - k}},
\end{align}
with $\phi_n = {\frac{D-D_0}{2}\cos \frac{{2n - 1}}{{2N}}\pi  + \frac{D+D_0}{2}}$.
For a large number of users, i.e., $K \to \infty$, the random IUI term in \eqref{mMIMONOMAIndividualAchievableRate3} can be approximated by a deterministic value given by
\begin{equation}\label{DeterministicInterference}
\mathop {\lim }\limits_{K \rightarrow \infty} \sum\limits_{i = k + 1}^K {P_{\rm{max}}}\frac{1}{1+d_i^{\alpha}} \approx \mathop {\lim }\limits_{K \rightarrow \infty} \sum\limits_{i = k + 1}^K {P_{\rm{max}}} I_i.
\end{equation}
Now, the asymptotic ergodic data rate of user $k$ can be approximated by
\begin{align}\label{mMIMONOMAErgodicRate}
&\mathop {\lim }\limits_{K \rightarrow \infty, M \rightarrow \infty} \overline{R_{k}^{\rm{\emph{m}MIMO-NOMA}}}= \mathop {\lim }\limits_{K \rightarrow \infty, M \rightarrow \infty} \int_{{D_0}}^D {\ln \left( {1 + \frac{{{\psi _k}}}{{1 + {x^\alpha }}}} \right){f_{{d_k}}}\left( x \right)dx} \notag\\
& \approx \mathop {\lim }\limits_{K \to \infty ,M \to \infty } \left( {\begin{array}{*{20}{c}}
K\\
k
\end{array}} \right){\frac{{k}}{{D + {D_0}}}} \sum\limits_{n = 1}^N {{\beta _n}\ln \left( {1 + \frac{{{\psi _k}}}{{{c_n}}}} \right)} {\left( {\frac{{\phi _n^2 - D_0^2}}{{{D^2} - D_0^2}}} \right)^{k - 1}}{\left( {\frac{{{D^2} - \phi _n^2}}{{{D^2} - D_0^2}}} \right)^{K - k}},
\end{align}
with ${\psi _k} = \frac{{{P_{\rm{max}}}M}}{{\sum\nolimits_{i = k + 1}^K {{P_{\rm{max}}}{I_i} + {KN_0}} }}$.
Substituting \eqref{mMIMONOMAErgodicRate} into \eqref{mMIMONOMASumRate} yields the asymptotic ergodic sum-rate of the \emph{m}MIMO-NOMA system with the MRC-SIC detection as in \eqref{ErgodicSumRatemMIMONOMA}, which completes the proof.

\subsection{Proof of Theorem \ref{Theorem3}}\label{AppendixC}
With $D=D_0$, based on the channel hardening property\cite{HochwaldMassiveMIMO}, the channel gain can be asymptotically formulated as:
\begin{equation}\label{ChannelHardening}
\mathop {\lim }\limits_{M \rightarrow \infty}\frac{{\left\| {{{\bf{h}}_k}} \right\|}^2}{M} = \frac{1}{1+D_0^{\alpha}}, \forall k.
\end{equation}
Substituting \eqref{ChannelHardening} into \eqref{mMIMONOMAIndividualAchievableRate2}, the asymptotic individual rate of user $k$ of the \emph{m}MIMO-NOMA system with $D=D_0$ is obtained by
\begin{align}\label{DD0mMIMONOMAIndividualAchievableRate2}
\mathop {\lim }\limits_{K \rightarrow \infty, M \rightarrow \infty} R_{k}^{\rm{\emph{m}MIMO-NOMA}} &= \mathop {\lim }\limits_{K \rightarrow \infty, M \rightarrow \infty} {\ln}\left(1+ M{\frac{{{P_{\rm{max}}}\frac{1}{1+D_0^{\alpha}}}}{{\sum\limits_{i = k + 1}^K {P_{\rm{max}}} \frac{1}{1+D_0^{\alpha}}  + K{N_0}}}} \right) \notag\\
& = \mathop {\lim }\limits_{K \rightarrow \infty, M \rightarrow \infty} \ln \left( {1 + \frac{{\delta \varpi }}{{\left( {1 - \frac{k}{K}} \right)\varpi  + 1}}} \right),
\end{align}
where $\delta = \frac{M}{K}$ and $\varpi = \frac{{P_{\rm{max}}}}{\left({1+D_0^{\alpha}}\right) N_0}$.
We can observe that the asymptotic individual rate of user $k$ in \eqref{DD0mMIMONOMAIndividualAchievableRate2} becomes a deterministic value for $K \rightarrow \infty$ and $M \rightarrow \infty$ due to the channel hardening property.
As a result, we have $\mathop {\lim }\limits_{K \rightarrow \infty, M \rightarrow \infty} R_{k}^{\rm{\emph{m}MIMO-NOMA}} = \mathop {\lim }\limits_{K \rightarrow \infty, M \rightarrow \infty} \overline{R_{k}^{\rm{\emph{m}MIMO-NOMA}}}$.

Now, the asymptotic ergodic sum-rate of the \emph{m}MIMO-NOMA system with MRC-SIC receiver can be obtained by
\begin{align}\label{DD0mMIMONOMAErgodicRate}
&\mathop {\lim }\limits_{K \rightarrow \infty, M \rightarrow \infty} \overline{R_{\rm{sum}}^{\rm{\emph{m}MIMO-NOMA}}} = \mathop {\lim }\limits_{K \to \infty ,M \to \infty } \sum\limits_{k = 1}^K \ln \left( {1 + \frac{{\delta \varpi }}{{\left( {1 - \frac{k}{K}} \right)\varpi  + 1}}} \right) \\
&= \mathop {\lim }\limits_{K \to \infty ,M \to \infty } K {{\int_{0}^1 {\ln \left( {1 + \frac{{\delta \varpi }}{{\left( {1 - x} \right)\varpi  + 1}}} \right) dx} } } \notag\\
&= \mathop {\lim }\limits_{K \to \infty ,M \to \infty } \frac{{M}}{\varpi\delta} \left[ \ln \left( 1 \hspace{-1mm}+\hspace{-0.5mm} \varpi\delta \hspace{-0.5mm}+\hspace{-0.5mm} \varpi \right) \left( 1 \hspace{-0.5mm}+\hspace{-0.5mm} \varpi\delta \hspace{-0.5mm}+\hspace{-0.5mm} \varpi \right) - \ln \left( 1 \hspace{-0.5mm}+\hspace{-0.5mm} \varpi\delta \right)\left( 1 \hspace{-0.5mm}+\hspace{-0.5mm} \varpi\delta \right) - \ln \left( 1 \hspace{-0.5mm}+\hspace{-0.5mm} \varpi \right)\left( 1 \hspace{-0.5mm}+\hspace{-0.5mm} \varpi \right) \right], \notag
\end{align}
which completes the proof of \eqref{DD0ErgodicSumRatemMIMONOMA}.

On the other hand, under the equal resource allocation strategy, the asymptotic individual rate of user $k$ of the \emph{m}MIMO-OMA system with the MRC detection in \eqref{mMIMOOMAIndividualAchievableRate} can be approximated by
\begin{equation}\label{DD0mMIMOOMAIndividualAchievableRate}
R_{k}^{\rm{\emph{m}MIMO-OMA}} \approx \delta\varsigma{\ln}\left(1+ {\frac{{{P_{\rm{max}}}{{\left\| {{{\bf{h}}_k}} \right\|}^2}}}{{{\varsigma M N_0}}}} \right).
\end{equation}
Exploiting the channel hardening property as stated in \eqref{ChannelHardening}, the individual rate of user $k$ in \eqref{DD0mMIMOOMAIndividualAchievableRate} can be approximated by a deterministic value and we have the asymptotic ergodic sum-rate of the  \emph{m}MIMO-OMA system considered as
\begin{equation}\label{DD0mMIMOOMAErgodicRate}
\mathop {\lim }\limits_{M \rightarrow \infty} \overline{R_{\rm{sum}}^{\rm{\emph{m}MIMO-OMA}}}
\approx \mathop {\lim }\limits_{M \to \infty } {\varsigma M} \ln \left( {1 + \frac{\varpi}{\varsigma}}\right),
\end{equation}
which completes the proof of \eqref{DD0ErgodicSumRatemMIMOOMA}.

\end{appendices}



\end{document}